\documentclass[10 pt, journal]{IEEEtran}  

\IEEEoverridecommandlockouts                              

\usepackage{amsmath}
\usepackage{parskip}
\usepackage{listings}
\usepackage{indentfirst}
\usepackage{graphicx}
\usepackage{float}
\usepackage{extarrows}
\usepackage{amssymb}
\usepackage{amsthm}
\usepackage{subfigure}
\usepackage{epstopdf}
\usepackage{caption}
\DeclareMathOperator*{\argmax}{argmax}

\usepackage{fancyhdr}
\begin{document}
\title{\LARGE \bf
Supplementary Material for the Paper ``The Interplay of Competition and Cooperation among Service Providers (Part \uppercase\expandafter{\romannumeral1})''
 }

\author{Xingran Chen,\IEEEmembership{}
        Mohammad Hassan Lotfi,\IEEEmembership{}
        Saswati Sarkar\IEEEmembership{}
\IEEEcompsocitemizethanks {\IEEEcompsocthanksitem Xingran Chen and Saswati Sarkar are with the Electrical and System Engineering Department of the University of Pennsylvania, PA, 19104. \protect\\
E-mail: xingranc@seas.upenn.edu,\quad swati@seas.upenn.edu
\IEEEcompsocthanksitem Mohammad Hassan Lotfi is with the Institute for System Research of the University of Maryland, College Park, MD, 20740.\protect\\
E-mail: mhl@umd.edu
}
\thanks{Parts of this work was presented in Annual Conference on Information Sciences and Systems (CISS), 2017. 
}}

\IEEEoverridecommandlockouts
\newtheorem{lemma}{Lemma}
\newtheorem{note}{Note}
\newtheorem{property}{Property}
\newtheorem{theorem}{Theorem}
\newtheorem{definition}{Definition}
\newtheorem{corollary}{Corollary}
\newtheorem{proposition}{Proposition}
\newtheorem{remark}{Remark}
\newtheorem{assumption}{Assumption}

\maketitle

\begin{abstract}
This paper investigates the incentives of mobile network operators (MNOs) for acquiring
additional  spectrum to offer mobile virtual network operators (MVNOs) and thereby inviting competition
for a common pool of end users (EUs). We consider a base case and two generalizations: (\romannumeral1) one MNO and one MVNO,  (\romannumeral2) one MNO, one MVNO and an outside option, and  (\romannumeral3) two MNOs and one MVNO. In each of these cases, we model the interactions of the service providers (SPs) using a sequential game, identify when the Subgame Perfect Nash Equilibrium (SPNE) exists, when it is unique and characterize the SPNE when it exists. The characterizations are easy to compute, and are in closed form or involve optimizations in only one decision variable. We identify metrics to quantify the interplay between cooperation and competition, and evaluate those as also the SPNEs to show that  cooperation between MNO and MVNO  can enhance the payoffs of both, while increased competition due to the presence of additional MNOs  is beneficial to EUs but reduces the payoffs of the SPs.

\end{abstract}

\begin{IEEEkeywords}
Heterogeneous networks, Wireless Internet Market, Service Providers, Spectrum provisioning, Subscriber pricing, Game Theory, Hierarchical games, Nash Equilibrium
\end{IEEEkeywords}

\section{Introduction}\label{sec: introduction}

\subsection{Motivation and Overview}
\IEEEPARstart{N}{owadays} wireless  service providers (SPs)  are divided into (i) mobile network operators (MNOs) that lease spectrum from a regulator like FCC, and (ii)
mobile virtual network operators (MVNOs)  that  obtain spectrum from one or more MNOs. MVNOs can distinguish their plans from MNOs by  bundling their service with other products,  offering different pricing plans  for End-Users (EUs), or building  a good reputation through a better customer service. Although traditionally wireless service has been offered only by MNOs, in recent years, the number of MVNOs has been rapidly growing. The number of MVNOs increased by 70 percent worldwide, during June 2010-June 2015  reaching 1,017 as of June 2015 \cite{M2015}.  Even some MNOs developed their own MVNOs. An example of which is Cricket wireless which is owned by AT\&T and offers a prepaid wireless service to EUs. Another example of MVNOs is the Google's Project Fi in which the customer's service is handled using Wi-Fi hotspots wherever/whenever they exist; elsewhere the service is handled using the spectrum of a number of MNOs, eg, Sprint, T-Mobile or U.S. Cellular networks.

In this work, we consider the economics of the interaction among   MNOs  and  MVNOs. We seek to understand why and under what conditions the MNOs cooperate with the MVNOs by offering some of their spectrum to the MVNOs, and thereby inviting competition for a common pool of EUs. We consider scenarios where the MNOs decide on acquiring new spectrum, and in exchange for a fee offer those to   MVNOs, which decide to acquire some of the spectrum  offered.  The SPs decide on their pricing strategies for the EUs, and the EUs decide to opt for one of them, or neither,  if the access fees and the qualities of service are not satisfactory.  The spectrum acquisition and pricing decisions of the SPs determine their respective profits. We characterize their equilibrium choices. We obtain metrics that quantify the cooperation and competition of the SPs in terms of their spectrum investments and subscriptions of EUs, which help quantify the interplay between competition and cooperation under the equilibrium choices.

 We consider a hotelling model in which a  continuum of undecided EUs  decide which of the SPs they want to buy their wireless plan from, if at all. The EUs have different preferences for each SP. These preferences can be because of different services and qualities that SPs offer. For example, the MVNOs may be able to offer a free or cheap international call plan through VoIP, or an  SP may have  an infamous customer service. The preference for a SP also increases with  the spectrum she acquires.  If, for example,  EUs have high preferences for MVNOs, then the MNOs may prefer to lease some of their spectrum to the MVNOs and receive their share of profit through the MVNOs, instead of competing for EUs by lowering their access fees. On the other hand, if EUs have high preferences for the MNOs, the MNOs may not offer spectrum to the MVNOs and seek to attract the EUs directly.  Thus, cooperation is mutually beneficial only in some scenarios, which we seek to identify.

\subsection{Contribution}
First, we  consider a  base case  in which one MNO and one MVNO compete for EUs in a common pool, and the EUs must choose one  of the SPs. We present the system model, important definitions and terminologies, and quantify metrics such as \emph{degree of cooperation} and \emph{EU-resource-cost} that we use to assess the system from the perspective of various stake-holders throughout (Section \ref{sec: BM-independent framework}). We consider  a sequential game  in which the SPs decide their spectrum investments and access fees for the EUs  (Section~\ref{games}). We subsequently seek the Subgame Perfect Nash Equilibrium (SPNE) outcome of the game using backward induction, and identify conditions under which  the SPNE  exists and is unique, and characterize the  SPNE whenever it exists (Sections~\ref{sec: nase case-outcome}, \ref{sec: BM-SPNE Analysis}). The SPNE is simple to compute, as 1) the amount of spectrum the MNO invests turns out to be the value that maximizes a function involving only one decision variable 2) the amount of spectrum the MVNO leases from the MNO is a simple closed form expression involving the amount that the MNO offers it and the leasing fee 3) the access fees for the EUs  constitute simple closed form expressions of the spectrum the SPs acquire. The characterizations provide several insights. The spectrum acquired by the MNO never falls below  a threshold which depends only on the leasing fee to the MVNO and preferences for the SPs. When the spectrum equals this threshold, the MVNO reserves the entire spectrum that the MNO offers it. Thus cooperation is high in this case.  As the MNO acquires higher amounts of spectrum,  the MVNO  reserves progressively lower amounts, leading to lower degrees of cooperation. Numerical computations  reveal that the MNO acquires minimal amount of spectrum only when the leasing fee to the MVNO   is smaller than a threshold (Section~\ref{sec: BM-numerical}).  The SPNE characterizations show that higher degrees of cooperation invariably reduces (enhances, respectively) the efficacy of the MNO (MVNO, respectively) in competing for the EUs; yet, higher degrees of cooperation enhance the payoffs of both the SPs as  our numerical computations reveal. The MNO's loss in revenue from subscription  is more than compensated by  the leasing fees obtained from the MVNO.

Second, we  generalize the hotelling model for EU subscription in the base case by incorporating an additional demand function (Section \ref{sec: out}). The effects of the demand function are two-fold. First, the demand function models the attrition in the number of EUs of SPs if the spectrum investment or price of both SPs is not desirable for EUs. Thus, in effect, an EU may opt for neither SP if neither offers a price-quality combo that is to his satisfaction, which is equivalent to opting for outside options. Second, the demand function models an exclusive additional customer base for each of the SPs to draw from depending on her investment and the price she offers.
We characterize the unique interior SPNE outcome of the game (Section \ref{sec: out-outcome}). Numerical results reveal that the general behavior of the SPNE outcome are as in the base case and that the EU-resource-cost increases compared to the base case (Section \ref{sec: out-4-numerical}).

Finally, we  generalize the base case to include competition between MNOs.  We consider a wireless market with  two MNOs and one MVNO, in which EUs choose one of the three SPs (Section \ref{sec: 3-player model}).  We generalize the hotelling model to consider three players instead of two in the classical ones (Section \ref{sec: 3p-1-model}), and
characterize the unique SPNE outcome   (Section \ref{sec: 3p-2-outcome}).   The characterizations show  that this enhanced competition 1) increases the  degree of cooperation, as the MVNO acquires all the spectrum that the MNOs offer,  and 2) is beneficial to EUs, as the amounts of spectrum of SPs acquires are higher, and the SPs charge the EUs less. Numerical results reveal that the  additional competition enhances the EU-resource-cost  compared to the base case.

\subsection{Relation with the Sequel}
 While in this work we consider that the SPs arrive at their decisions individually, in the accompanying sequel we consider that the SPs arrive at certain   decisions as a group, and then arrive at other  decisions individually  (Part \uppercase\expandafter{\romannumeral2}). Also, here we assume that the per unit leasing fee the MVNO pays to MNO(s) is a fixed parameter, which is beyond the control of individual MNOs and MVNOs. This happens for example in two important cases: 1) when this fee is determined by an external regulator to influence the interaction between different providers (possibly to the betterment of the EUs) 2) when this fee is a market-driven parameter, for example, in a  large spectrum market with many MNOs and MVNOs.  To understand the impact of  the externals (eg, regulator, market),  we investigate the implications of different values of this fee  on the SPNE and the payoffs and the EU-resource-cost metric. This would also guide the regulatory choice of this fee for the first case.     Note that the overall market may consider several MNOs and MVNOs, whose presence we consider in the generalizations (Sections~\ref{sec: out}, \ref{sec: 3-player model}).   In the sequel  we consider that the SPs cooperatively  characterize this fee as a decision variable in a bargaining framework (Part \uppercase\expandafter{\romannumeral2}).

\subsection{Positioning vis-a-vis the State-of-the-Art}
Duan {\it et. al} made early contributions in the field of MVNOs \cite{LD2011}, \cite{LD2010}. They formulated the interactions between one cognitive mobile virtual network operator (CMVNO) and multiple end-users as a multi-stage Stackelberg game, and showed that spectrum sensing could improve the profit of the CMVNO and payoffs of the users. Since they considered only one SP, the issue of competition or cooperation between multiple SPs did not arise. We investigate the interplay of cooperation and competition between different SPs, namely MNO and MVNO.

The economics of the interactions  among multiple service providers have been extensively investigated. 
We focus on non-cooperative interactions in this paper as here we consider that the SPs arrive at their decisions individually.
Non-cooperative games were considered for example in \cite{HK2008}, \cite{LD2010},  \cite{SBO2012}, and \cite{R2017}. A general framework of strongly Pareto-inefficient Nash equilibria with noncooperative flow control was considered in \cite{HK2008}. Applying the framework to communication networks, it was shown that the Nash equilibria were not efficient.
Intervention schemes, i.e., systems where users and an intervention device interact, were formulated in \cite{PS2012}, and a solution concept of intervention equilibrium was proposed. The paper showed that intervention schemes could improve the suboptimal performance of non-cooperative equilibrium.
\cite{R2017} proposed wireless virtualization to investigate spectrum sharing in wireless networks.

However, these works did not consider both MNO and MVNO,  whose roles are fundamentally different from each other. The MNO acquires spectrum from a central regulator, which it offers  to MVNO  in exchange of money, and the MVNO  uses part  of this spectrum. Both   MNO and MVNO earn by selling  wireless plans to the EUs; the MNO earns additionally  by leasing  spectrum to the MVNO. Thus, they make different decisions, which affect their subscriptions, and  their payoffs have different expressions. Their decisions also follow different constraints: spectrum acquired by the MVNO is upper bounded by that acquired by the MNO, which constitutes the MNO's decision variable, while the spectrum acquired by the MNO depend on the availability with the regulators, the availability does not constitute the decisions of any provider. 
The interaction between the MNO and MVNO lead to an interplay of competition and cooperation  between them, which calls for innovations in the realm of modeling and analysis.

 To  our knowledge, the only papers  in the genre of non-cooperative interactions that also consider interactions of the MNOs and MVNOs are
 \cite{Banerjee2009}, \cite{le2009pricing} and \cite{CB2012}.
 In \cite{Banerjee2009}   MNOs seek to maximize the joint profit of MNO and MVNO.  The MNO's selection of access fees is formulated as a maximization in which the sales of the MNO is expressed as a function of only the fee he selects. In contrast we consider that each SP seeks to maximize his individual profit and   obtain the access fees they select and the spectrum they acquire, which also determine how  the EUs choose between the SPs. Thus we need to dwell in the realm of a hierarchical game rather than a single stage optimization. A scenario very different from ours is considered in \cite{le2009pricing}:  the SPs  \emph{do not} compete
for consumer market shares but for the proportion of resource they are going to use. The interaction between the SPs is a hierarchical game in which the MNO and MVNO choose their access fees,  the MVNO also decide investment in content/advertising. The access fees become roots of a fourth order polynomial equation which is computed numerically. The closest to our work is  \cite{CB2012}, which considers a dynamic three-level sequential game of spectrum sharing between one MNO and one MVNO. The focus is however complementary to ours. Unlike our work,   \cite{CB2012} does  not consider decisions of the 1) MNO pertaining to how much spectrum to acquire from a regulatory body 2) MVNO pertaining to how much of the MNO's spectrum offer he ought to accept  (he assumes that the MVNO uses the entire spectrum the MNO offers). We also generalize our model to consider multiple MNOs and an MVNO, which  \cite{CB2012} does not. \cite{CB2012} however considers a decision of the MVNO that we do not, i.e.,  how much the MNO would invest in content generation. The EU subscription models are also entirely different. We consider a one-shot game involving a continuum of EUs in which the SP choice of each  EU is based on his intrinsic preferences for the SPs and the spectrum investments of the SPs.  \cite{CB2012} considers  a  multi-time slot game in which a discrete number of  EUs choose between the SPs based on their experiences in the previous slots and their estimates of  the quality of service the  SPs they had not chosen apriori offer. The games we consider fundamentally differ in that the SPNE need not exist in ours (we identify necessary and sufficient conditions for its existence), while it always exists in that in \cite{CB2012}. By exploiting the structure of the game, we  obtain closed form  expressions for the various decisions we consider, in the SPNE,  whenever it exists.  \cite{CB2012} computes the SPNE only numerically through the solution of a multi-slot stochastic dynamic program (DP). Our SPNE characterization is easy to compute, while DPs usually suffer from the curse of dimensionality.

  \section{Base case}\label{sec: basic case}

 We present the system model in which we formulate the payoffs and strategies of SPs, and the utilities and decisions of EUs (Section~\ref{sec: BM-independent framework}). Next, we formulate  the interaction between different entities as a sequential game (Section~\ref{games}). Subsequently, we characterize the conditions under for  the existence and the uniqueness of the SPNE, obtain closed form expressions for the SPNE when it exists (Section~\ref{sec: nase case-outcome}).  We present numerical results in Section~\ref{sec: BM-numerical}. We prove the analytical results in Section~\ref{sec: BM-SPNE Analysis}, Appendix~\ref{Appendix-BM} (Theorems~\ref{thm: Un-conclusion-3sp}, \ref{thm: BM-prices}, \ref{thm: BM-I_F}, \ref{thm: stage 1}), and Appendix~\ref{Appendix: the interior SPNE} (Theorems~\ref{thm: Un-conclusion-sectionA}, \ref{thm: Un-corner-sectionA}).

\subsection{Model}\label{sec: BM-independent framework}

We consider one MNO ($\text{SP}_{L}$,  $L$ represents leader) and one MVNO ($\text{SP}_{F}$, $F$ represents follower) which compete for a common pool of undecided EUs.  $\text{SP}_{L}$   offers  $I_L$ amount of  spectrum  (which it acquires from a regulator) to $\text{SP}_{F}$  in exchange of money, and $\text{SP}_{F}$ uses $I_F$ amount of this spectrum. Clearly, $0\leq I_{F}\leq I_{L}$. For simplicity of analysis and formulation, we assume that $0<\delta\leq I_{L}$, where $\delta$ is a lower bound of $I_L$, which is a parameter of choice.  This assumption is not significantly restrictive as $\delta$ may be chosen as low a positive quantity as one desires\footnote{All results extend, with some modifications,  when we consider that $I_L$ is upper bounded by $M$. Such bounds may apply when   the central regulator has limited spectrum to offer.  Refer to Section~\ref{sec: Limited Spectrum from Central Regulator} for the deductions.}. 
 Both  $\text{SP}_{L}$ and $\text{SP}_{F}$ earn by selling  wireless plans to  EUs; $\text{SP}_{L}$ earns additionally  by leasing her spectrum to $\text{SP}_{F}$.   We assume that both SP$_L$ and SP$_F$ have access to separate spectrum,  which they can use to serve the EUs who join them, above and beyond the  $I_L, I_F$ amounts  they strategically acquire.  For example, a SP$_F$ like Google's Project Fi serves  customers   using Wi-Fi hotspots and the spectrum of  3 MNOs (Sprint, T-Mobile or U.S. Cellular networks). Also,  SP$_L$ may acquire additional spectrum from the regulator which it does not offer SP$_F.$

We denote the marginal leasing fee (per spectrum unit)  that $\text{SP}_{L}$ pays the regulator as $\gamma$, marginal reservation fee   $\text{SP}_{F}$ pays to $\text{SP}_{L}$ by $s$,  the fraction of EUs that $\text{SP}_{F}$ and $\text{SP}_{L}$ attract as $n_{F}$ and $n_{L}$, respectively,  and  the access fee that $\text{SP}_{F}$ and $\text{SP}_{L}$ charge the EUs as $p_{F}$ and $p_{L}$, respectively. Since $\text{SP}_{L}$ wants to lease out some of her spectrum to $\text{SP}_{F}$ with profit motive, it is reasonable to assume that  $s>\gamma$.  We assume that  $s, \gamma$ are pre-determined. The strategies of SPs are to choose the investment levels ($I_{L}$, $I_{F}$) and the access fees for EUs ($p_{L}$, $p_{F}$) so as to maximize their overall payoffs, which we formulate next.

SP$_F$ and SP$_L$ respectively earn  revenues of $n_{F}(p_{F}-c), n_{L}(p_{L}-c)$ from EU subscription, where $c$ is the transaction cost SPs incur in subscription.  The transaction cost arises due to traffic management, billing and accounting services, customer service, etc. associated with each subscription. We have assumed such costs to be equal for all SPs, as they do not significantly vary across them.
 We expect the cost of reserving spectrum to be strictly convex, i.e. the cost of investment per spectrum unit increases with the amount of spectrum. Strictly convex costs do not satisfy the economy of scale; the regulator may mandate such structures  to stop excessive acquisition by big SPs seeking to control the market, which has limited spectrum supply,  and drive out  smaller SPs or new entrants. Incidentally,
several seminal works have considered strictly convex investment costs, e.g. \cite{Hartment} and \cite{Abel}. For simplicity in analysis, we consider a specific kind of strictly convex cost function, namely  quadratic,  and  discuss generalizations in Remark~\ref{re: convex function}. That is,  SP$_{L}$ incurs a spectrum acquisition cost of  $\gamma I_{L}^{2}$, and SP$_F$ pays to SP$_L$ a leasing fee of  $sI_{F}^{2}$.
Thus, the payoffs of SPs are:
\begin{align}
 &\pi_{F}=n_{F}(p_{F}-c)-sI_{F}^{2}\label{equ: BM-L-payoff}\\
 &\pi_{L}=n_{L}(p_{L}-c)+sI_{F}^{2}-\gamma I_{L}^{2}.\label{equ: BM-F-payoff}
\end{align}

\begin{figure}
\centering
\includegraphics[width=3in]{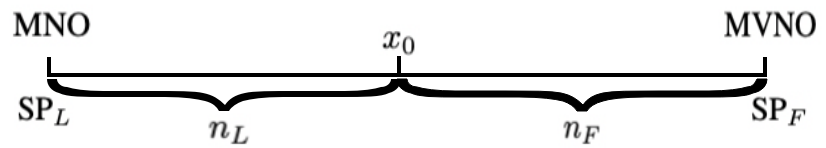}
\caption{The hotelling model for the base case. The EUs in $[0, x_0]$ ($[x_0, 1]$, respectively) prefer SP$_L$ ( SP$_F$, respectively).  The former fraction of EUs is $n_L$, the latter is $n_F.$ $x_0$ is farther off from SP$_L$ as $t_L$ becomes lower and $v_L - v_F$ become higher.  }
\label{fig:mno-mvno}
\end{figure}

\noindent{\textbf{EUs:}}
We use a hotelling model\cite{Osborne1987} to describe how EUs choose between the SPs. We assume that SP$_{L}$ is located at $0$, SP$_{F}$ is located at $1$, and EUs are distributed uniformly along the unit interval $[0,1]$ (Figure~\ref{fig:mno-mvno}). The closer an EU to a SP, the more this EU prefers this SP to the other. Note that the notion of closeness and distance is used to model the preference of EUs, and may not be the same as physical distance. Let $t_{L}$ ($t_{F}$) be the unit transport cost of EUs for SP$_{L}$ (SP$_{F}$), the EU located at $x\in[0,1]$ incurs a cost of $t_{L}x$ \big(respectively, $t_{F}(1-x)$\big) when joining SP$_{L}$ (respectively, SP$_{F}$).
\begin{equation}\label{equ: BM-utility EUs}
\begin{aligned}
u_{L}(x)=&v^{L}-\left(p_{L}+t_{L}x\right) \\
u_{F}(x)=&v^{F}-\left(p_{F}+t_{F}(1-x)\right).
\end{aligned}	
\end{equation}
The EU at $x$ receives utilities $u_{L}(x), u_{F}(x)$ respectively  from SP$_{L}$ and  SP$_{F}$, and joins the SP that gives it the higher utility.

The first component of the utility functions comprises of the ``static factors'', namely $v^{L}$ and $v^{F}$ of $\text{SP}_{L}$ and $\text{SP}_{F}$, respectively. The static factor of a SP is the same for all EUs, which depends on the local presence, its existing spectrum beyond $I_L$ or $I_F$   and its reputation in the region, quality of the customer-service, ease of usage for the online portals, etc. However, the static factors do not depend on strategies of SPs, such as the access fees,  the investment levels, etc.

The second component, i.e., $p_{L}+t_{L}x$ or $p_{F}+t_{F}(1-x)$,  is denoted as the
``strategy factor''.  The strategy factors depend on the strategies of the SPs, namely their access fees and the spectrum $I_L, I_F$  they acquire. Clearly, the utilities would decrease with the access fees, we consider the dependence to be linear. As SP$_{F}$ acquires greater fraction of the additional spectrum SP$_{L}$ offers him, SP$_F$ becomes more desirable and SP$_L$ less desirable to the EUs. Denote $t_L=I_{F}/I_{L}$ and $t_F=(I_{L}-I_{F})/I_{L}$. Then the impact of quality of service in the decision of EUs is captured through $t_{L}$ and $t_{F}$. For example, when $I_F=I_L$, i.e.,  SP$_F$ leases the entire $I_L$ spectrum from SP$_L$ and SP$_L$ can use none of it,  then $t_F=0$ and $t_L=1$. This gives SP$_F$ an advantage over SP$_L$ in attracting EUs.
Similarly, even when $I_F=0$, i.e., SP$_F$ leases no spectrum from SP$_L$, $t_F=1$ and $t_L=0$, SP$_L$ has an advantage over SP$_F$. But subscription may still be divided in both  the above extreme cases.   This happens since both SP$_F$ and SP$_L$  have access to separate spectrum as reflected in the static factors $v^F, v^L$. Note that the pair of transport cost ($t_{L}=I_{F}/I_{L}, t_{F}=1-t_{L}$) is one of the many functions that can be considered. We choose this model specifically since it captures the essence of the model, and is analytically tractable.

Finally, the strategy factors incorporate intrinsic preference of the EUs towards the SPs through the coordinate $x$, which presents the local distance in the utility model. If an EU is for example close to SP$_{F}$, $x$ is high and $1-x$ is low, and it is deemed to have a higher intrinsic preference for SP$_{F}$, as compared to SP$_{L}$. The intrinsic preference may be developed through pre-existing and ongoing relations the EU has with the SPs, e.g., if an EU is already availing of other services from a SP, the EU will have a stronger intrinsic preference for the SP, due to convenience of billing etc. Higher intrinsic preferences enhance utilities of the SP for the EUs.  The impact of the strategies of the SPs on the EUs will depend on their intrinsic preferences for the EUs, which is captured in the term $t_{L}x$ or $t_{F}(1-x)$ in the utility. Note that the intrinsic preference is different for different EUs unlike the static factor.

 We consider that $v^{L}$ and $v^{F}$ are sufficiently large so that the utility of EUs for buying a wireless plan  is positive regardless of the choice of SP\footnote{ Note that all analytical results will depend on the difference of $v^L$ and $v^F$, so absolute values of these (large or otherwise) do not have any impact on the SPNE choices of various entities.}. Thus, each EU chooses exactly one SP to subscribe to, i.e., the market is ``fully covered''. This is a common assumption for hotelling models. We would in effect relax  this assumption in Section~\ref{sec: out}.

  SP$_{F}$'s leasing of spectrum from SP$_{L}$ constitute an act of cooperation. Thus, we call $I_{F}/I_{L}$ {\it the degree of cooperation}. Since SP$_{F}$ and SP$_{L}$ compete to attract EUs,  the split of subscription $(n_{L}, n_{F})$  represent the level of competition. Since the amount of spectrum SP$_{F}$ leases from SP$_{L}$ determines the split of subscription, there is a natural interplay between cooperation and competition, that these metrics will enable us to quantify.

We develop the notion of \emph{EU-resource-cost} to capture the spectral resource per unit access fee averaged over all EUs, which  represents the ``bang-for-the-buck'' or ``value for money'' an average EU gets out of the system. For the EUs who choose the MVNO, the resource per head is $I_F/n_F$. Thus, for these EUs the resource per head per unit  fee is $I_F/(n_F p_F)$. Similarly, for the EUs who choose the MNO, the resource per head per unit fee is $(I_L-I_F)/(n_L p_L)$. Averaging over all the EUs, the resource per unit fee for an ``average'' EU then is, $\frac{n_F I_F/(n_F p_F) + n_L (I_L-I_F)/(n_L p_L)}{n_F + n_L}$, which equals $I_F/p_F+(I_L-I_F)/p_L$, since $n_L + n_F = 1.$ We therefore consider this as the expression for the EU-resource-cost. Clearly, higher values of the EU-resource-cost is beneficial for the EUs.

\subsection{The sequential game framework}

\label{games}

The interaction among SPs and EUs can be formulated as a sequential game. As a leader of the game, $SP_{L}$ makes the first move. The timing and the stages of the game are as following:
\begin{itemize}
  \item {\bf Stage 1:}  $\text{SP}_{L}$ decides on the amount of spectrum,  $I_{L}$, to acquire.
  \item {\bf Stage 2: }  $\text{SP}_{F}$ decides on the amount of spectrum to lease from $\text{SP}_{L}$, $I_{F}$.
  \item {\bf Stage 3:} $\text{SP}_{L}$ and $\text{SP}_{F}$ determine the access fees for the EUs, $p_{L}$ and $p_{F}$, respectively.
  \item {\bf Stage 4:}  Each EU  subscribes to the SP that gives it the higher utility.
\end{itemize}
\begin{remark}
 We assume that the decision of investments ($I_{L}$ and $I_{F}$) happens before the decisions of access fees ($p_{L}$ and $p_{F}$), guided by the fact that spectrum investment decisions are long-term ones, and are therefore expected to be constants over longer time horizons in comparison to subscription pricing decisions.
\end{remark}

\begin{definition}\cite[Chapter 6.2]{Osborne1994}\label{def: BM-SPNE}
A strategy is a {\it Subgame Perfect Nash Equilibrium} (SPNE) if and only if it constitutes a Nash Equilibrium (NE) of every subgame of the game.
\end{definition}

We refer to a SPNE choice of spectrum investments and access fees by the SPs  as $(I_L^*, I_F^*, p_L^*, p_F^*)$, and the EU subscriptions for the SPs under the same as $n_L^*, n_F^*$, should a SPNE exist.

\subsection{The SPNE outcome}\label{sec: nase case-outcome}

We next identify the conditions under which SPNE exists,  characterize the SPNE when it exists, and examine its uniqueness.

We denote $v^L - v^F$ as $\Delta$.   Since $0 \leq t_L, t_F \leq 1$, $0 \leq x \leq 1,$
in the expressions for utilities in \eqref{equ: BM-utility EUs}, $|\Delta| \geq 1$ provides a near insurmountable disadvantage to one of the SPs through the static factors; this SP might have to choose a significantly lower price to recoup. Thus, we first focus on the range $|\Delta| < 1.$ As stated before, we assume $\delta$ is small, and let $\delta<\sqrt{\frac{2-\Delta}{9s}}$, which reduces to $\delta<\sqrt{\frac{2}{9s}}$ in the special case that $v^L = v^F$.

\begin{theorem}\label{thm: Un-conclusion-sectionA}
Let $|\Delta| < 1$. The SPNE  is:

\noindent{\bf (1)} any solution of the following maximization is $I_{L}^{*}$,
 \begin{equation*}
\begin{aligned}
\max_{I_{L}}\,\,&\pi_{L}(I_{L})=(\frac{2+\Delta}{3}-\frac{1-\Delta}{27sI_{L}^{2}-3})^{2} \\
&+s(\frac{(1-\Delta)I_{L}}{9sI_{L}^{2}-1})^{2}-\gamma I_{L}^{2}\\
s.t\,\,&\sqrt{\frac{2-\Delta}{9s}}\leq I_L\leq M,
\end{aligned}
\end{equation*}
\noindent{\bf (2)}  $I_{F}^{*}$ is characterized in
 \begin{equation*}
 \begin{aligned}
  I_{F}^{*}=\left\{\begin{aligned}
  &\frac{(1-\Delta)I_{L}}{9I_{L}^{2}s-1}\,\,& \text{if}\,\, I_{L}>\sqrt{\frac{2-\Delta}{9s}}\\
  &I_{L}\,\, & \text{if}\,\, I_{L}=\sqrt{\frac{2-\Delta}{9s}}\end{aligned}\right.,
\end{aligned}
\end{equation*}
\noindent{\bf (3)}
$p_{L}^{*}=c+\frac{2}{3}-\frac{I_{F}^{*}}{3I_{L}^{*}}+\frac{\Delta}{3},\quad
p_{F}^{*}=c+\frac{1}{3}+\frac{I_{F}^{*}}{3I_{L}^{*}}-\frac{\Delta}{3}$,

\noindent{\bf (4)} $n_{L}^{*}=\frac{\Delta}{3}+\frac{2}{3}-\frac{I_{F}^{*}}{3I_{L}^{*}},\, n_{F}^{*}=\frac{I_{F}^{*}}{3I_{L}^{*}}+\frac{1}{3}-\frac{\Delta}{3}$.
\end{theorem}

\begin{remark}
\label{remark2}
From (2), $I_F^*$ is unique once $I_L^*$ is given; from (3) and (4), $(p_L^*,p_F^*,n_L^*,n_F^*)$ is unique once $I_L^*$ and $I_F^*$ are given. Thus, every solution of the maximization in Theorem \ref{thm: Un-conclusion-sectionA}~(1) leads to a distinct SPNE. Thus, the SPNE is unique if and only if this maximization has a unique solution. Our extensive numerical computations suggest that this is the case.
\end{remark}

The SPNE is easy to compute, despite the expressions being cumbersome.  Otherwise, $I_L^*$ can be obtained as a maximizer of an expression that involves only one decision variable, $I_L$, and fixed parameters $ s, \gamma,  \Delta$.   $I_F^*$ has been expressed as a closed form function involving $I_L^*$ and the fixed parameters $ s, \Delta.$ $p_L^*, p_F^*, n_L^*, n_F^*$ have been expressed  as closed form functions of $I_F^*/I_L^*$ and the fixed parameters  $c,  \Delta.$

From Theorem \ref{thm: Un-conclusion-sectionA}~(3), the price the EUs receive from SP$_L$ (respectively, SP$_F$) decrease (respectively, increase) with increase in the degree of cooperation ($I_{F}/I_{L}$). Thus, since at least one of the SPs reduce the price, the EUs benefit from  higher degree of cooperation.

From Theorem \ref{thm: Un-conclusion-sectionA} (3) and (4), $n_{L}^{*}=p_{L}^{*}-c, n_{F}^{*}=p_{F}^{*}-c.$ Thus, SPNE subscriptions of the SPs increase with increase in the access fees they announce. This counter-intuitive feature arises because the subscriptions also depend on the spectrum acquisitions of the SPs, through the transport costs $t_L=I_{F}/I_{L}$ and $t_F=1-t_F$ in the utilities specified in (\ref{equ: BM-utility EUs}).

From Theorem \ref{thm: Un-conclusion-sectionA}~(1), in the SPNE,  SP$_L$ acquires at least   $\sqrt{\frac{2-\Delta}{9s}}$ amount of spectrum.
From Theorem \ref{thm: Un-conclusion-sectionA}~(2), when $I_L^{*}$ equals this minimum, then SP$_F$ reserves all the available spectrum, i.e., $I_L^{*}=I_F^{*}$ (note that $I^*_F$ is continuous at $I_L=\sqrt{\frac{2-\Delta}{9s}}$). Thus, SP$_L$ can not use any of $I_L^{*}.$ However,  from Theorem \ref{thm: Un-conclusion-sectionA}~(4),   SP$_L$ is still able to attract a positive fraction of EUs: $n_{L}^{*} = \frac{\Delta+1}{3} > 0$  since $|\Delta|<1$. This is because EUs have spectrum other than $I_L^{*}, I_F^{*}$ as captured in the values of $v^{L}, v^{F}$.

 From Theorem \ref{thm: Un-conclusion-sectionA}~(1) and (2), when $I_L^{*}$ exceeds its minimum value,  then SP$_F$ reserves only a fraction of available spectrum ($I^*_F<I_L^*$). Note that in this case, $\frac{d I^*_F}{d I_L}<0$. Thus, the higher the amount of available spectrum, the lower would be the amount of spectrum reserved by SP$_F$. Also, $I^*_F$ is decreasing with $s$.

 The SPNE depends on the static factors $v^L, v^F$ only through their difference $\Delta$. As expected, with increase (respectively, decrease) in  $\Delta$, SP$_L$ (respectively, SP$_F$) can increase his (respectively, her) access fee $p_L^*$ (respectively, $p_F^*$).
The minimum value of his spectrum acquisition $I_L^*$ increases with decrease in $\Delta$,  to offset the competitive advantage  the static factors provide.  Through our numerical computations, we elucidate how $I_L^*, I_F^*$ and the payoffs otherwise vary with $\Delta$.

The results illustrate the interplay between cooperation and competition.  From Theorem \ref{thm: Un-conclusion-sectionA}~(4), the subscription $n_{L}^{*}$ (respectively, $n_{F}^{*}$) of SP$_L$ (respectively, SP$_F$) decreases (respectively, increases) with the degree of cooperation ($I_{F}^*/I_{L}^*$).
Thus, the higher the degree of cooperation, lesser (respectively, greater)  is the competition efficacy of SP$_L$ (respectively, SP$_F$).   A natural question arises: why would the SP$_L$ then cooperate with the SP$_F$? From (\ref{equ: BM-L-payoff}) and (\ref{equ: BM-F-payoff}), Theorem \ref{thm: Un-conclusion-sectionA}~(3),  (4),   $\pi_{L}=n_{L}^{*2}+sI_{F}^{*2}-\gamma I_{L}^{*2}$, and $\pi_{F}=n_{F}^{*2}-sI_{F}^{*2}$. On the one hand, if the degree of cooperation increases, then the amount of subscribers of SP$_{L}$ decreases, thus the revenue  SP$_{L}$  earn from the subscribers decreases. On the other hand, the payoff of SP$_{L}$  increases through $sI_{F}^{*2}$. Thus the second factor may offset the first, and the payoff of SP$_L$ may increase due to cooperation.   Note that it is not a zero sum game, thus, the payoffs  of both  players may simultaneously increase due to cooperation. We illustrate these phenomena definitively through our numerical computations in the next section.

Then, in the extreme case that $| \Delta | \geq 1$:
 \begin{theorem}\label{thm: Un-corner-sectionA}
{\bf (1) } $\Delta\geq 1$: The  SPNE  is
\[ I_{L}^{*}=\delta, I_{F}^{*}=0, p_{F}^{*}=p_{L}^{*}-\Delta,
n_{L}^{*}=1,\, n_{F}^{*}=0,\] and
$p_L^*$ can be chosen any value in $[c+1, c+\Delta].$ \newline
{\bf (2) } $\Delta = 1:$ The following interior strategy constitute an additional  SPNE:
\[I_L^*=I_F^*=\frac{1}{3\sqrt{s}}, p_L^* - c = n_L^* = 2/3, p_F^* - c = n_F^*= 1/3.\]

{\bf (3) } $\Delta\leq -1:$ The SPNE strategy is:
 \[I_L^*=I_F^*=\frac{1}{\sqrt{2s}}, p_{L}^{*}=p_{F}^{*}+\Delta-1, n_{L}^{*}=0,\, n_{F}^{*}=1,\] and
$p_L^*$ can be chosen any value in $[c+1, c-\Delta].$ \newline
\end{theorem}
 We prove this theorem in Appendix~\ref{sec: corner SPNE}. As is intuitive, for large $\Delta$, all EUs subscribe to SP$_L$, despite lower access fees selected by  SP$_F$; the reverse happens in the other extreme, despite lower access fees selected by  SP$_F$. The extremes therefore lead to ``corner equilibria'', which
  correspond to $0, 1$ as the degrees of cooperation. The SPNE is non-unique in both these extremes.

\subsection{Numerical results}\label{sec: BM-numerical}

\begin{figure}
\centering
\includegraphics[width=1.7in]{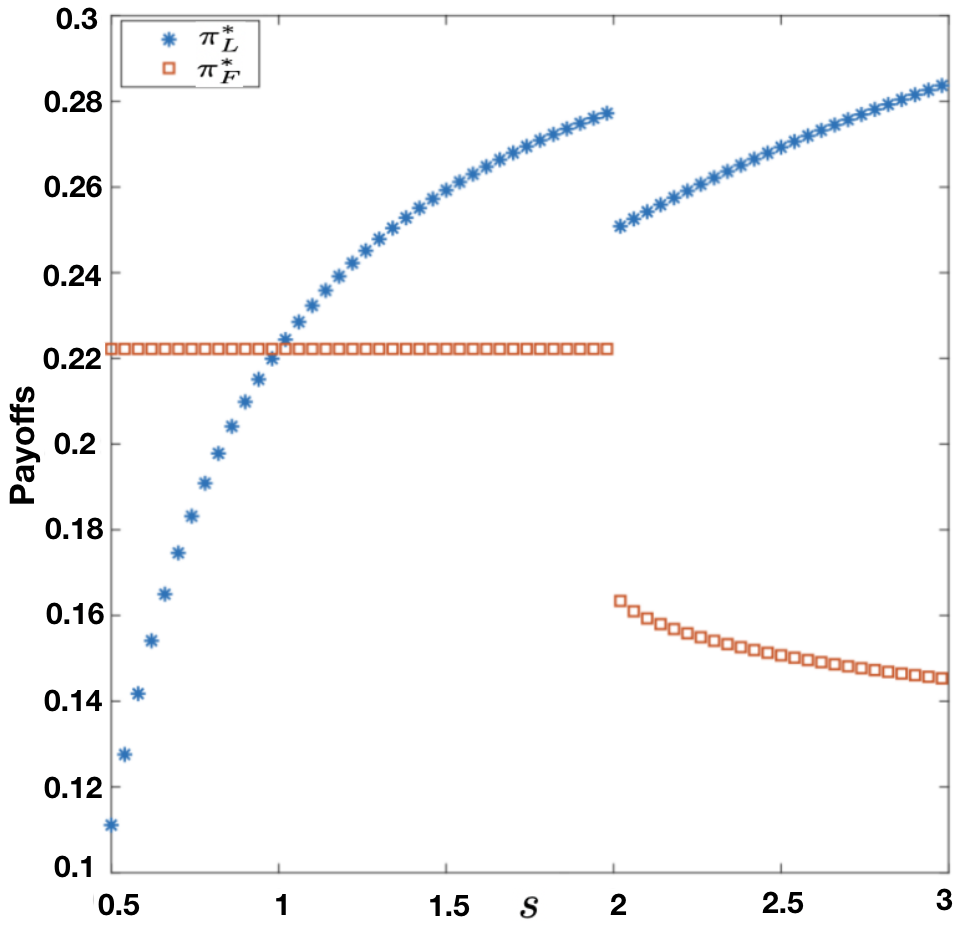}
\includegraphics[width=1.7in]{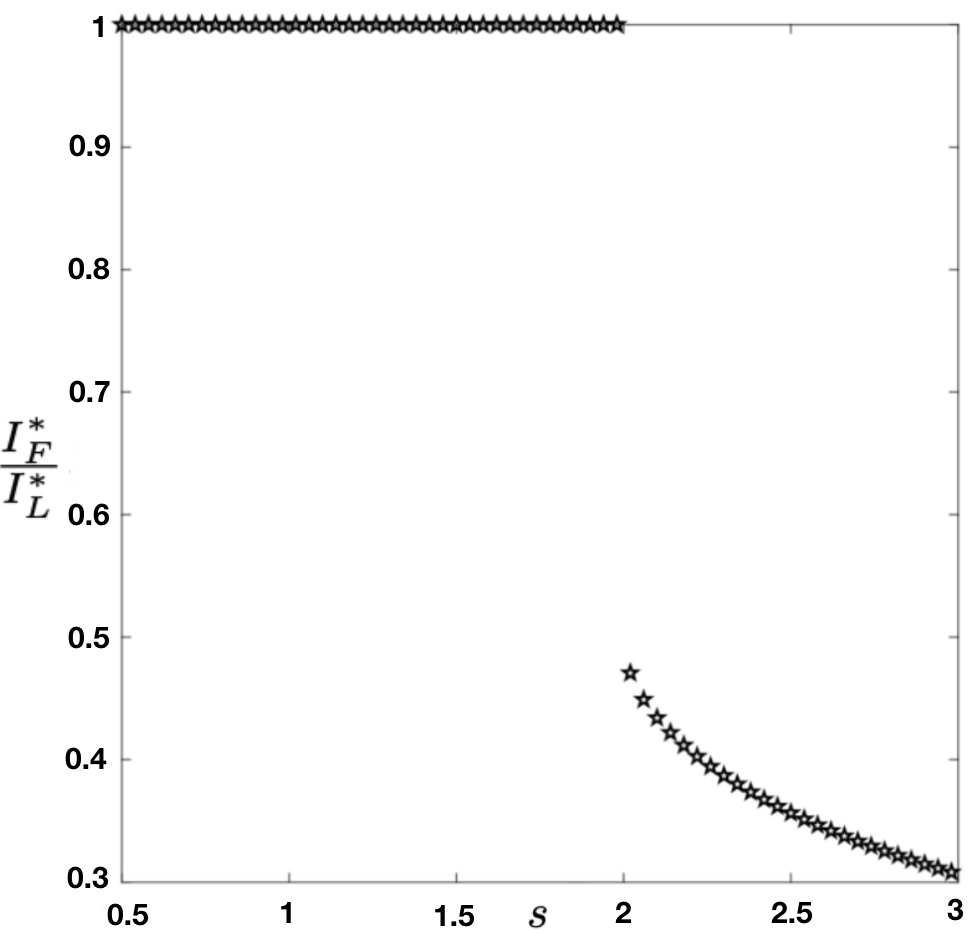}
\caption{Payoffs (left) and the degree of cooperation (right) vs. $s$. Here, $\gamma=0.5$, $c=1$, $\Delta = 0$.}
\label{figBM1}
\end{figure}

Figure \ref{figBM1} shows the payoffs (left) and the degree of cooperation (right) under different  $s$ when $\Delta = 0$. The degree of cooperation reaches the maximum ($=1$), i.e., $I_{F}^{*}=I_{L}^{*}$ when $s$ is less than a threshold ($\approx2$). In this case,  SP$_{L}$ generates most of its revenue from  the reservation fee paid by  SP$_{F}$. As expected, $\pi_{L}^{*}$ increases with $s$. From Theorem \ref{thm: Un-conclusion-sectionA} (1), (2), (4), when $I_{F}^{*}=I_{L}^{*}$,  $I_{L}^{*}$ equals its minimum value $\sqrt{\frac{2}{9s}}$, and  $n_{F}^{*}=1/3+I_{F}^{*}/3I_{L}^{*}=2/3$, thus $\pi_{F}^{*}=n_{F}^{*2}-sI_{F}^{*2}$ is a constant which is independent of $s$.
When $s$ is larger than this threshold, $I_{F}^{*}/I_{L}^{*}<1$, and decreases with $s$. In this case, $I_{L}^{*}$ exceeds its minimum value, and SP$_F$ leases only a portion of the new spectrum invested by SP$_L$, i.e., $I_{F}^{*} < I_{L}^{*}$. Thus, SP$_{L}$ generates more of its revenue from EUs. The payoff of  SP$_{L}$ (SP$_{F}$) first jumps to a lower value at this threshold, and then increases (decreases) with $s$. At this threshold, the degree of cooperation also jumps to a lower value ($<1$). Thus, higher degrees of cooperation can enhance the payoff of both SPs, and the reservation fee $s$ enhances (reduces) the payoff of SP$_{L}$ (SP$_{F}$). Also,  SP$_F$ earns more than SP$_L$  for lower values of $s$; hence SP$_F$ gets more from the spectrum sharing between the 2 SPs in this case.   For higher values of $s$,  the reverse happens. 

\begin{figure}
\centering
\includegraphics[width=2in, height=1.55in]{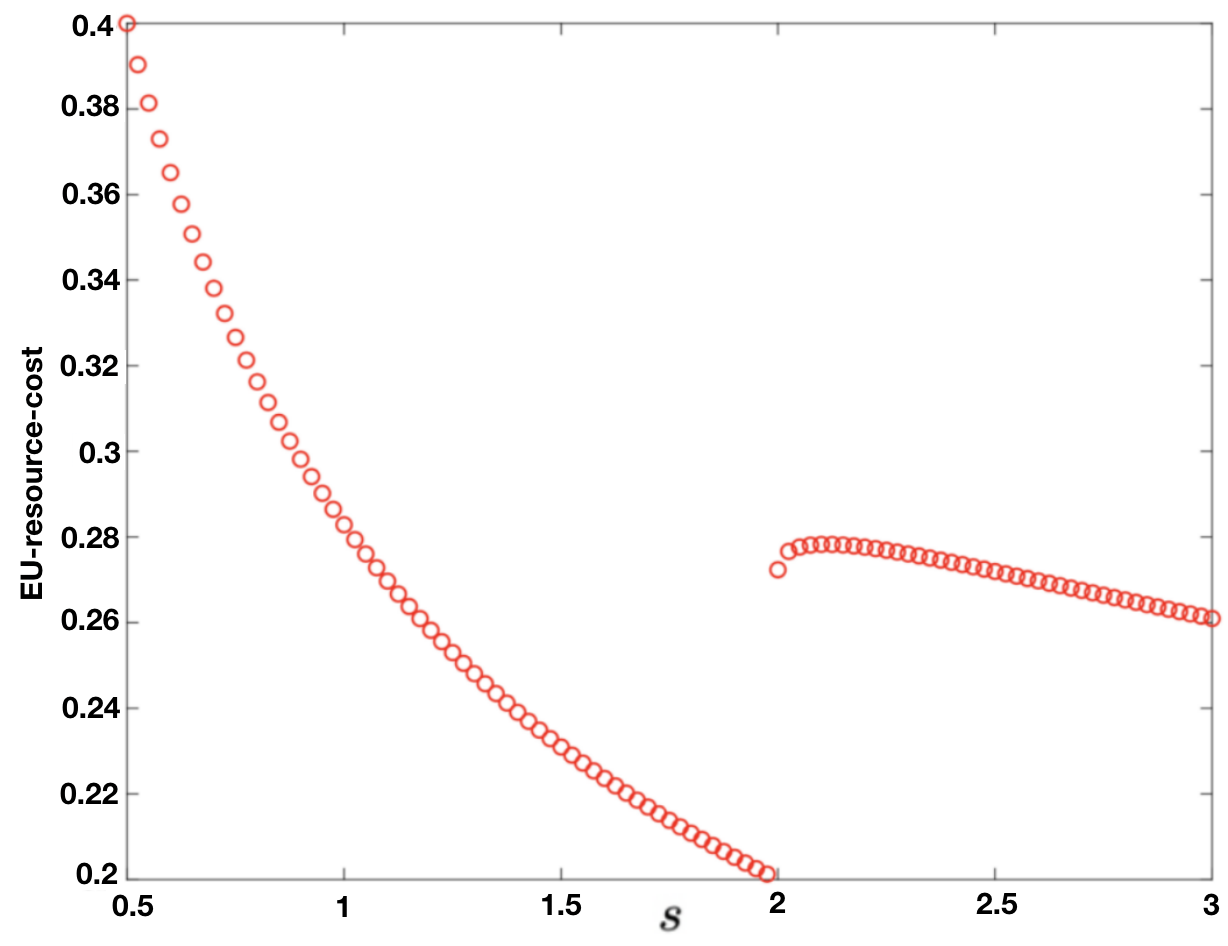}
\caption{EU-resource-cost vs. $s$. Here, $\gamma=0.5$, $c=1$.}
\label{figBM_a1}
\end{figure}

  $s$ has significant impact on the EU-resource-cost, as  depicted in Figure~\ref{figBM_a1}. We first explain the jump at the  threshold value of $s$.  When $s$ is less than the threshold, $I_L^* = I_F^*$, as seen in Figure~\ref{figBM1}~(right). Thus the EU-resource-cost  is $I_F^*/p_F^*$. At the threshold,  $I_F^* < I_L^*$, so the second term in EU-resource-cost \big($(I_L^*-I_F^*)/p_L^*$\big) jumps to a positive value from $0$,  leading to the  jump in the EU-resource-cost. The EU-resource-cost otherwise decreases in $s$, thus if a regulator chooses $s$, it ought to opt for a low value of $s$, though if $s$ is really low, then  SP$_{L}$ may not have enough incentive to cooperate due to low  $\pi_{L}^{*}$ (Figure~\ref{figBM1}~(left)). Note that the degree of cooperation is $1$ at low values of $s$, thus high degree of cooperation coincides with high EU-resource-cost.

\begin{figure}
\centering
\includegraphics[width=1.7in]
{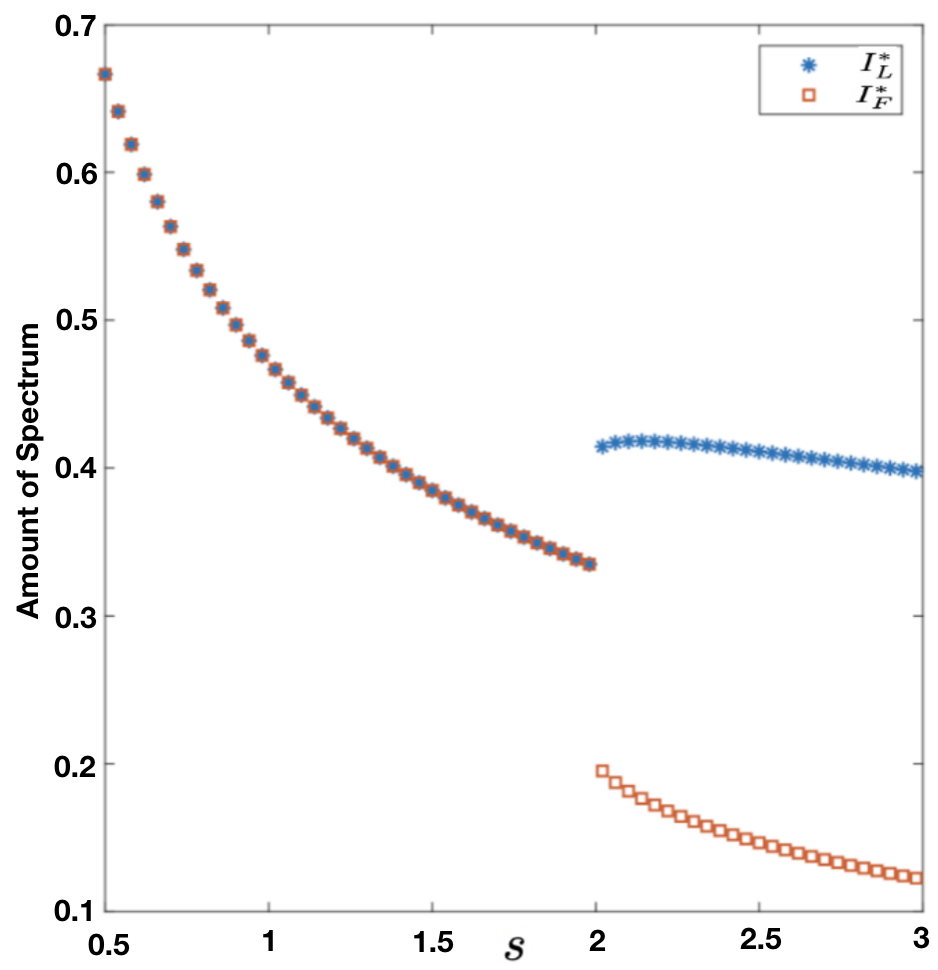}
\includegraphics[width=1.7in]
{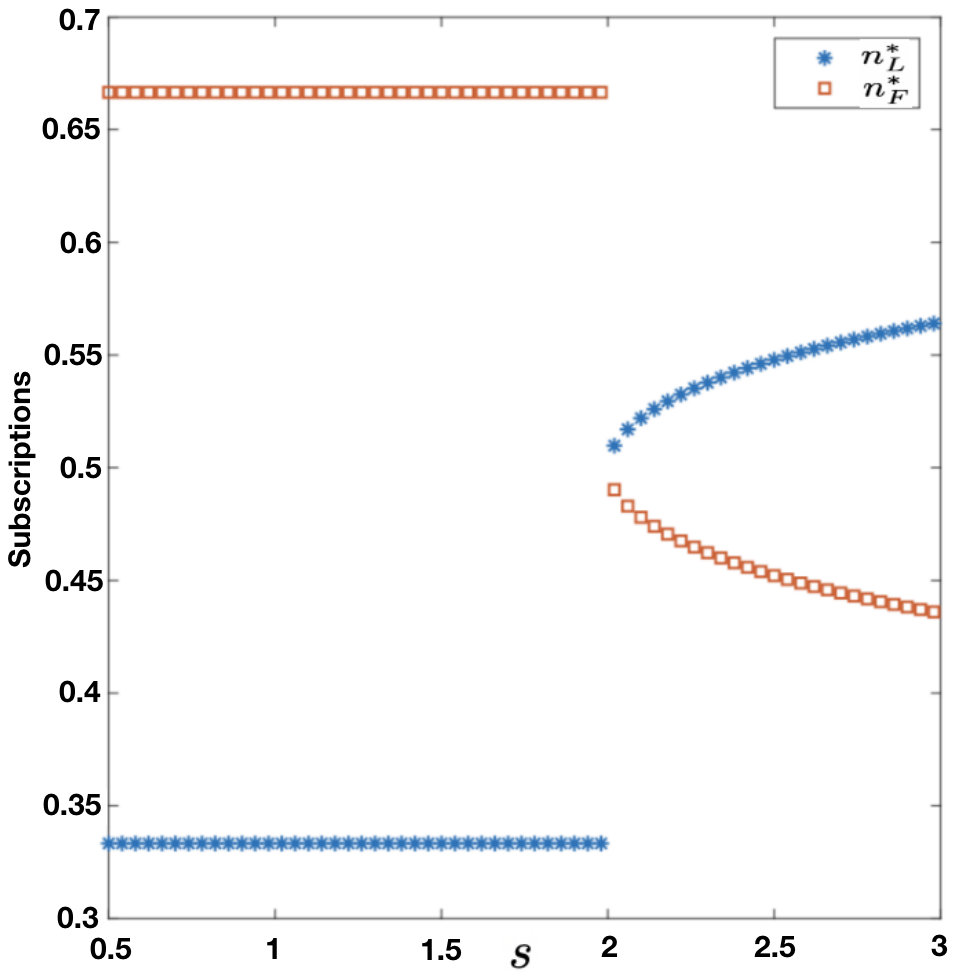}
\caption{Investment decisions (left), the split of subscription (right) vs. $s$. Here, $\gamma=0.5$, $c=1$, $\Delta = 0$.}
\label{figBM2}
\end{figure}

\begin{figure}
\centering
\includegraphics[width=1.7in]{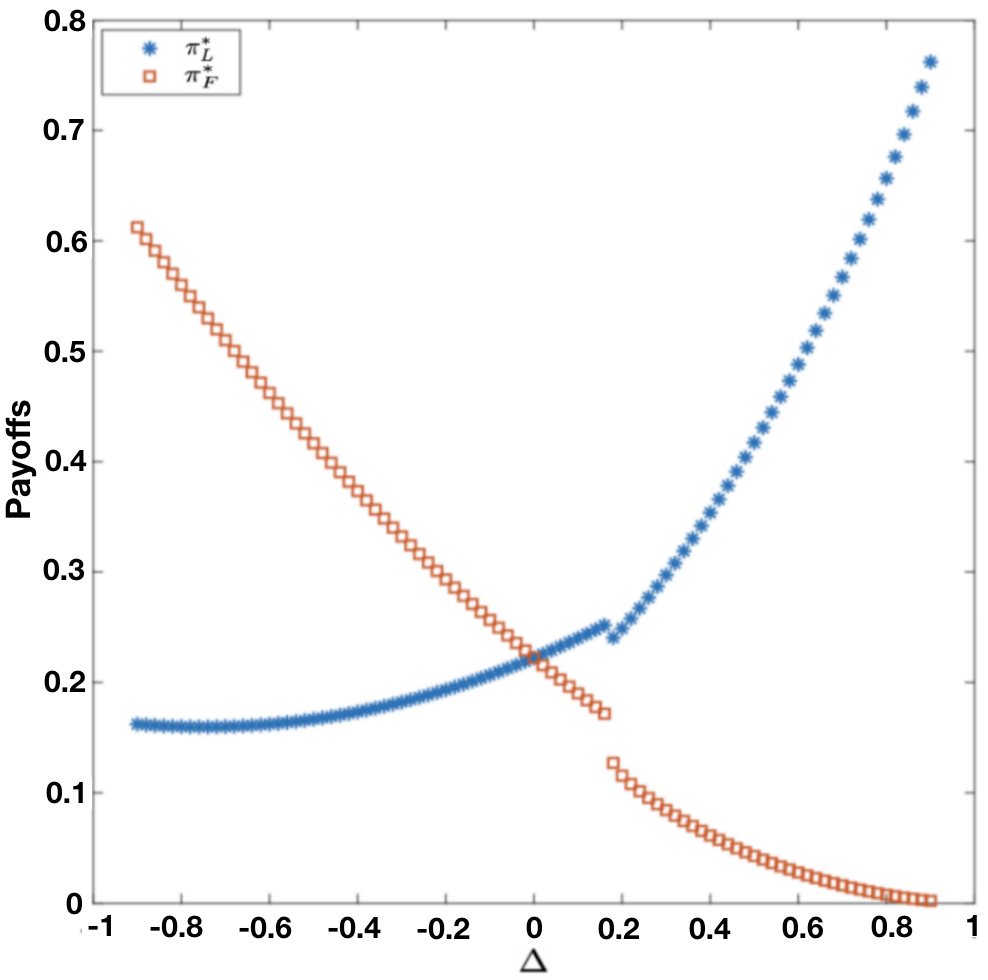}
\includegraphics[width=1.7in]{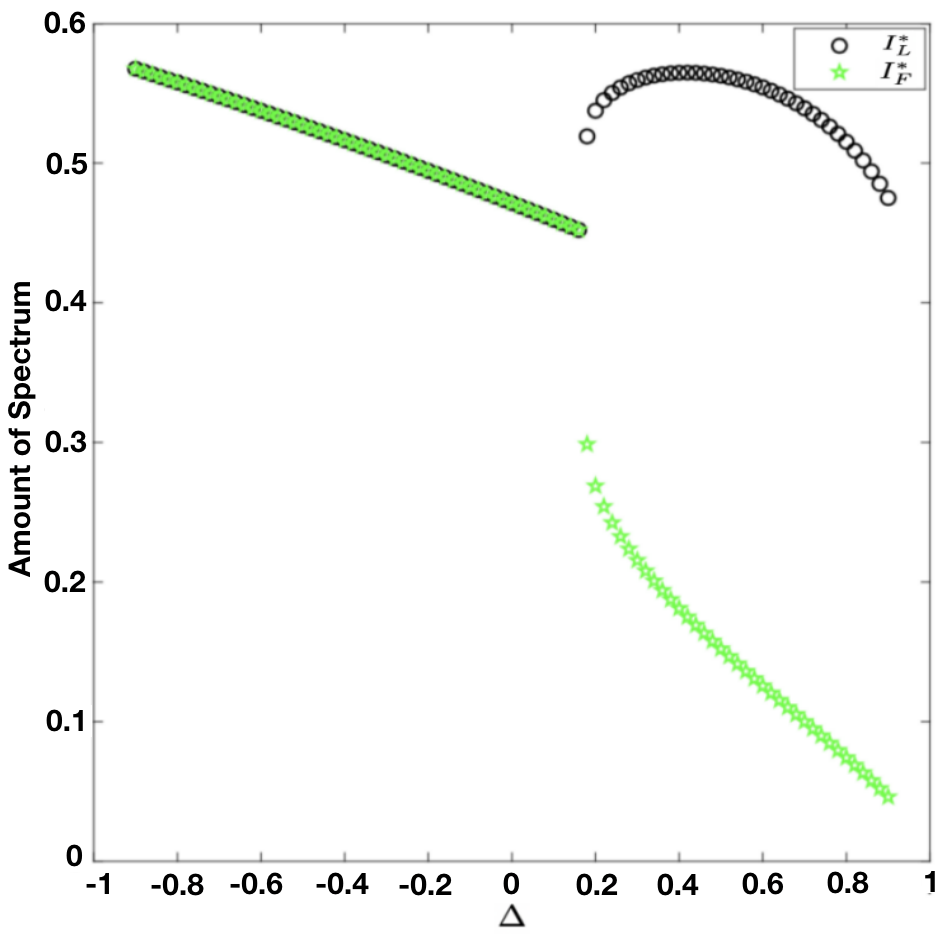}
\caption{Payoffs (left), investment decisions (right) vs. $\Delta$. Here, $\gamma=0.5$, $c=1$, $s=1$.}
\label{figBM3}
\end{figure}

Figure~\ref{figBM2} shows the SPNE level of investment  (left) and subscriptions of SPs (right) when $\Delta = 0$. It reconfirms that  when $s$ is smaller than a threshold, SP$_F$ leases the entire spectrum SP$_L$ offers, and after that threshold, SP$_F$ leases only a portion of the new spectrum offered by  SP$_L$. Also,  $I^*_L$  strictly decreases with $s$ throughout.
 When $s$ is small, $I_{F}^{*}=I_{L}^{*}$, $n^*_F$ and $n^*_L$ are constant ($n_{L}^{*}=1/3$, $n_{F}^{*}=2/3$) independent of $\gamma$ and $s$, and $n^*_F>n^*_L$. After the threshold, $n^*_F$ decreases and $n^*_L$ increases with $s$ (because $I_{F}^{*}/I_{L}^{*}$ decreases with $s$ in Figure~\ref{figBM1} (right)). Comparing Figure~\ref{figBM1} (right) and Figure~\ref{figBM2} (right) we note that higher degrees of cooperations increase (decrease, respectively) the competition efficacy of SP$_F$ (SP$_L$, respectively).

Figure \ref{figBM3} plots the payoffs (left) and $I_L$, $I_F$ (right) as a function of $\Delta$ when $ |\Delta| < 1$, the region in which the SPNE  exists uniquely. We set $s=1$. As expected, the payoff of SP$_L$ (SP$_F$, respectively) increase (decrease, respectively) with increase in $\Delta$. Also,  SP$_F$ earns more than SP$_L$  for lower values of $\Delta$; hence SP$_F$ gets more from the spectrum sharing between the 2 SPs in this case.  For higher values of $s$,  the reverse happens. With increase in $\Delta$, $I_L$, $I_F$ may either increase or decrease, depending on whether additional spectrum provides ``bang for the buck'' by enticing commensurate number of EUs  which depends on the EUs' prior biases (static factors)  for or against the SPs. The figure shows which is the case.


\begin{figure}
\centering
\includegraphics[width=1.7in, height=1.55in]{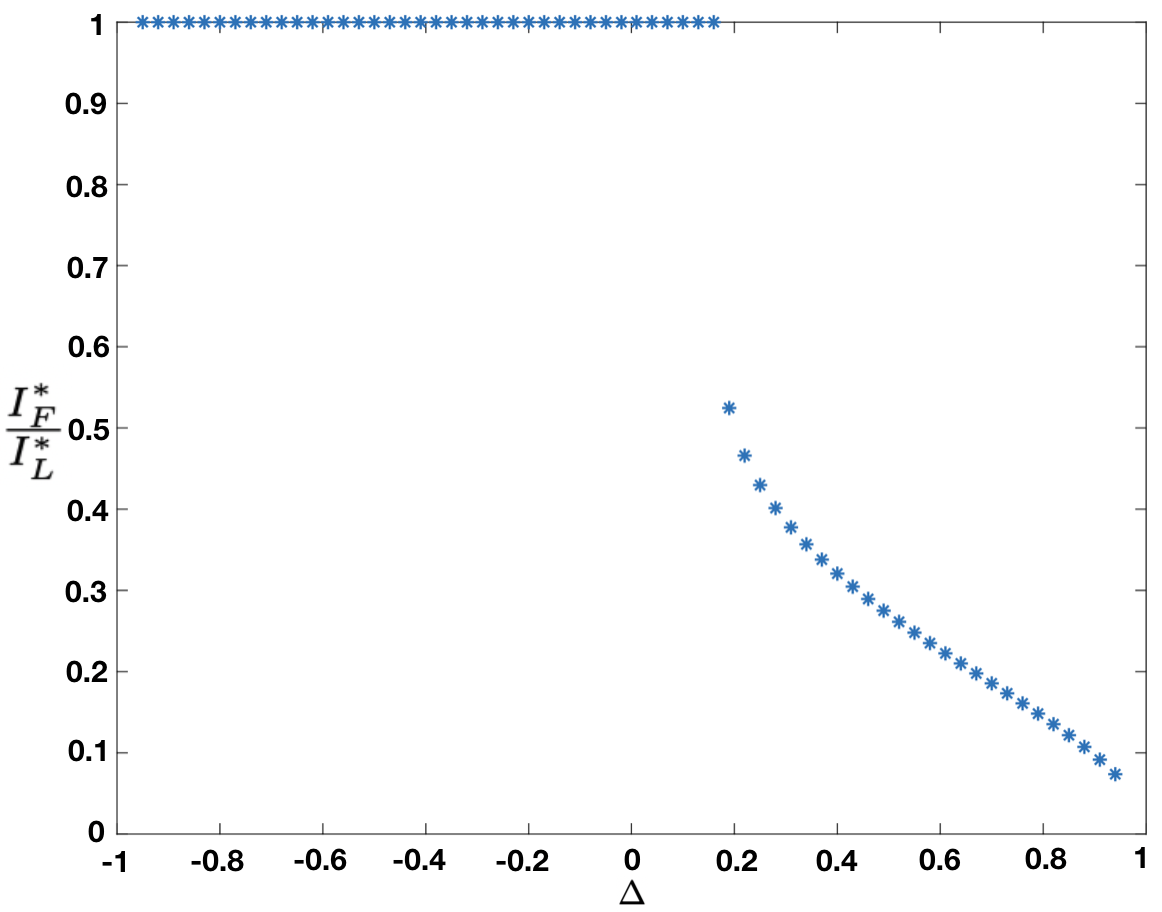}
\includegraphics[width=1.7in, height=1.55in]{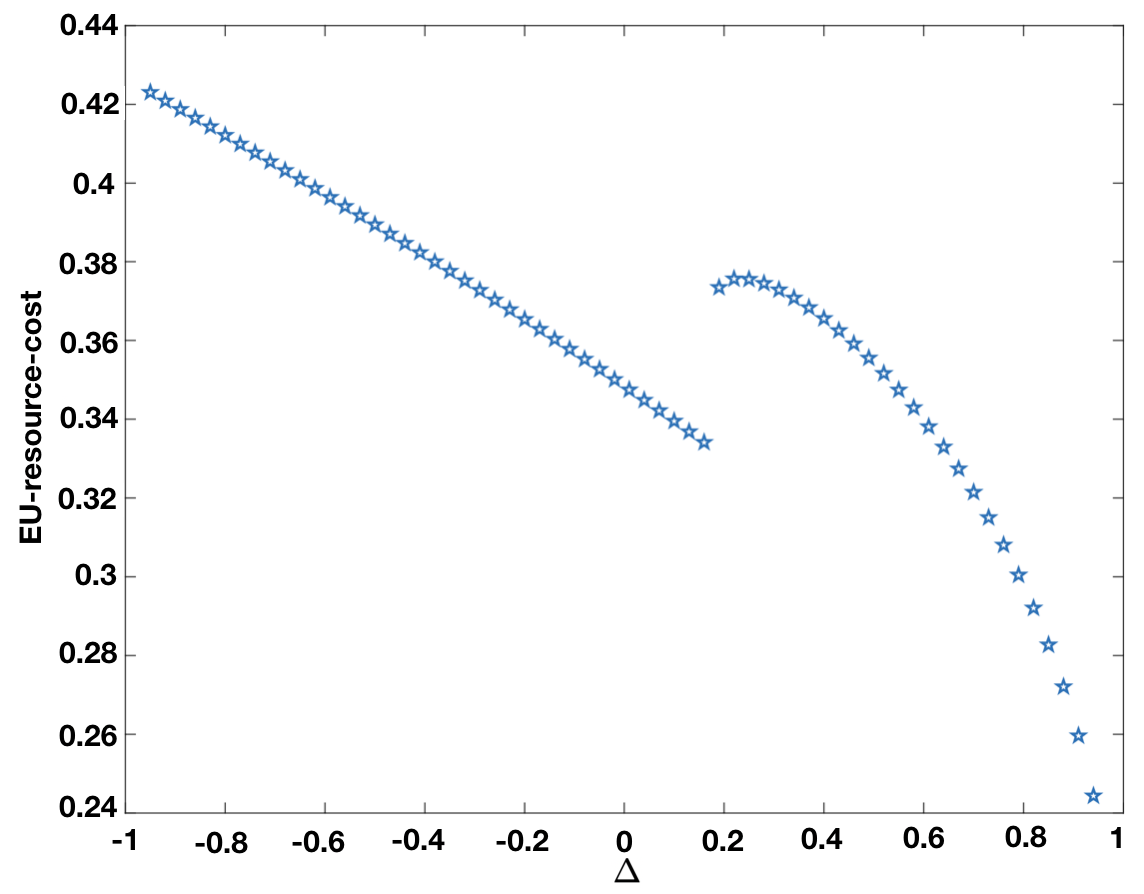}
\caption{Degree of cooperation (left), EU-resource-cost (right) vs. $\Delta$. Here, $\gamma=0.5$, $c=1$, $\Delta = 0$.}
\label{figBM_unequal}
\end{figure}

Figure~\ref{figBM_unequal} plots the degree of cooperation (left) and the EU-resource-cost  (right) as a function of $\Delta$ when $ |\Delta| < 1$. Figure~\ref{figBM_unequal} (left) shows that the degree of cooperation is a constant $1$ when $\Delta$ is less than a threshold, and decreases when $\Delta$ is larger than this threshold. The amount of spectrum SP$_F$ leases from SP$_L$ decreases when SP$_L$ has larger common preference.
 The jump in the EU-resource-cost at the threshold value of $\Delta$ may be explained similar to that for Figure~\ref{figBM_a1}, considering Figure~\ref{figBM_unequal}~(left) instead of Figure~\ref{figBM1}~(right). Other than this jump, the EU-resource-cost  decreases in $\Delta.$ Again, note that  high degree of cooperation coincides with high EU-resource-cost.

\subsection{SPNE Analysis}\label{sec: BM-SPNE Analysis}
We use backward induction to characterize SPNE strategies,  starting from the last stage of the game and proceeding backward. For simplicity and brevity,  we present this analysis only for the important special case of $\Delta = 0$, and defer the general case to Appendix~\ref{Appendix-BM}. Thus, we prove Theorem~\ref{thm: Un-conclusion-sectionA} while applying $\Delta = 0$ in the corresponding expressions. Specific Theorems~\ref{thm: Un-conclusion-3sp}, \ref{thm: BM-I_F}, \ref{thm: stage 1} are proven in Appendix~\ref{Appendix-BM}.

\noindent{{\bf Stage 4:}}
We first characterize the equilibrium division of EUs between SPs, i.e., $n_L^*$ and $n_F^*$, using the knowledge of the strategies chosen by the  SPs in Stages 1$\sim$3.

\begin{definition}\label{def: BM-indifferent location}
$x_{0}$ is the {\it indifferent location} between the two service providers if $u_{L}(x_{0})=u_{F}(x_{0})$ (Figure~\ref{fig:mno-mvno}).
\end{definition}

By the full market coverage assumption, if  $0<x_{0}<1$, then EUs in the interval $[0, x_{0}]$ join $\text{SP}_{L}$ and those in the interval $[x_{0}, 1]$ join $\text{SP}_{F}$. If $x_{0}\leq0$,  all EUs choose $\text{SP}_{F}$; and if $x_{0}\geq1$,  all  EUs choose $\text{SP}_{L}$ (Figure~\ref{fig:mno-mvno}).

From Definition \ref{def: BM-indifferent location},
$u_{F}(x_{0})=v-t_{F}(1-x_{0})-p_{F}=v-t_{L}x_{0}-p_{L}=u_{L}(x_{0})$.
Since $t_{L}+t_{F}=1$, then
$x_{0}=\frac{t_{F}+p_{F}-p_{L}}{t_{L}+t_{F}}=t_{F}+p_{F}-p_{L}$. Thus,
\begin{align}\label{equ: BM-indifferent location}
x_{0}=t_{F}+p_{F}-p_{L}
\end{align}

Thus, since EUs are distributed uniformly along $[0,1]$,
the fraction of EUs with each SP  is:
\begin{equation}\label{equ: BM-demand}
\begin{aligned}
&n_{L}=\left\{\begin{aligned}
&0,&\,\text{if}\quad&x_{0}\leq0\\
&x_{0},&\,\text{if}\quad&0<x_{0}<1\\
&1,&\,\text{if}\quad&x_{0}\geq1\\
\end{aligned}\right.,\, n_{F}=1-n_{L},
\end{aligned}
\end{equation}
where $x_{0}$ is defined in (\ref{equ: BM-indifferent location}) and $n_F = 1-n_L$ (Figure~\ref{fig:mno-mvno}).

Only ``interior'' strategies may be SPNE, as:

\begin{theorem}\label{thm: Un-conclusion-3sp}
In the SPNE it must be that $0<x_{0}<1.$
\end{theorem}

\noindent{{\bf Stage 3:}}
SP$_L$ and SP$_F$ determine their access fees for EUs, $p_L$ and $p_F$, respectively, to maximize their payoffs.

\begin{lemma}\label{lem: BM-stage-3-payoffs}
The payoffs of SPs are:

\begin{equation}\label{equ: BM-interior-payoffs}
\begin{aligned}
\pi_{L}=&(t_{F}+p_{F}-p_{L})(p_{L}-c)+sI_{F}^{2}-\gamma I_{L}^{2}\\
\pi_{F}=&(t_{L}+p_{L}-p_{F})(p_{F}-c)-sI_{F}^{2}
\end{aligned}
\end{equation}
\end{lemma}

\begin{proof}
From (\ref{equ: BM-demand}), substitute $(n_{L}, n_{F})=(t_{F}+p_{F}-p_{L}, 1-n_{L})$ into (\ref{equ: BM-L-payoff}) and (\ref{equ: BM-F-payoff}), and get (\ref{equ: BM-interior-payoffs}).
\end{proof}

We next obtain the SPNE $p_{F}^{*}$ and $p_{L}^{*}$ which maximize the payoffs $\pi_{L}$ and $\pi_{F}$ of the SPs respectively.

\begin{theorem}\label{thm: BM-prices}
The SPNE pricing strategies  are:
\begin{equation}\label{equ: BM-prices}
p_{L}^{*}=c+\frac{2}{3}-\frac{I_{F}}{3I_{L}},\quad p_{F}^{*}=c+\frac{1}{3}+\frac{I_{F}}{3I_{L}}
\end{equation}
\end{theorem}
\begin{proof}
$p^*_F$ and $p^*_L$ must satisfy the first order condition, i.e., $\frac{d \pi_F}{dp_F}=0$ and   $\frac{d \pi_L}{dp_L}=0$.  Thus,
$p^*_F=c+\frac{I_L+I_F}{3I_L} \quad \& \quad p^*_L=c+\frac{2I_L-I_F}{3I_L}$.
 $p^*_F$ and $p^*_L$ are the unique SPNE strategies if they yield $0< x_0< 1$ and no unilateral deviation is profitable for SPs. We establish these respectively in  Parts A and B.

\noindent{{\bf Part A}.}
From (\ref{equ: BM-prices}),
$x_0=\frac{I_L^*-I_F^*}{I_L^*}+p^*_F-p^*_L=\frac{2I_L^*-I_F^*}{3I_L^*}$. Since $I_L^*\geq I_F^*$ and $I_L^*>0$, then $0<x_{0}<1$.

\noindent{{\bf Part B}.}
Since $\frac{d^2 \pi_F}{dp^2_F}<0, \frac{d^2 \pi_L}{dp^2_L}<0$,  a local maxima is also a global maximum, and any solution to the first order conditions  maximize the payoffs   when $0<x_0<1$,  and no unilateral deviation by which $0<x_0<1$ would be profitable for the SPs.
Now, we show that unilateral deviations of the SPs leading to  $n_L=0, n_F=1$ and $n_L=1, n_F=0$  is not profitable. Note that the payoffs of the SPs, \eqref{equ: BM-L-payoff} and \eqref{equ: BM-F-payoff}, are continuous  as $n_L\downarrow 0$, and $n_L\uparrow 1$ (which subsequently yields  $n_F\uparrow 1$ and $n_F\downarrow 0$, respectively). Thus, the payoffs of both SPs when selecting $p_L$ and $p_F$ as the  solutions of the first order conditions are greater than or equal to the payoffs when $n_L=0$ and $n_L=1$. Thus, the unilateral deviations under consideration are not profitable for the SPs.
\end{proof}

\begin{remark}\label{re: convex function}
The proof shows that $x_0, p_{L}^{*}, p_{F}^{*}$ do not depend on the specific nature of the costs of leasing spectrum $I_F, I_L$, neither does $n_L^*, n_F^*$ from \eqref{equ: BM-demand}. Thus the SPNE expressions for these would remain the same for any other cost function. But, the SPNE of investment levels ($I_{L}^{*}$, $I_{F}^{*}$) as obtained in the next results depend on the specific nature of these functions.  \end{remark}

\noindent{{\bf Stage 2:}}
SP$_F$ decides on  the amount of spectrum to be leased from SP$_L$, $I_F$, with the condition that $0\leq I_F\leq I_L$,  to maximize $\pi_F$.

\begin{theorem}\label{thm: BM-I_F}
The SPNE spectrum acquired by $\text{SP}_{F}$ is:
\begin{align}\label{equ: BM-I_F}
  I_{F}^{*}=\left\{\begin{aligned}&\frac{I_{L}}{9I_{L}^{2}s-1}\quad&\text{when}&\quad I_{L}>\sqrt{\frac{2}{9s}}\\&I_{L}\quad&\text{when}&\quad \delta\leq I_{L}\leq\sqrt{\frac{2}{9s}}\end{aligned}\right.
\end{align}
\end{theorem}

\noindent{{\bf Stage 1:}}
 SP$_L$ chooses the amount of spectrum $I_L$ to lease from the regulator, to maximize $\pi_L$.
\begin{theorem}\label{thm: stage 1}
The SPNE spectrum acquired by SP$_L$, $I^*_L$ is the solution of the following maximization
 \begin{equation}\label{equ: BM-I_L}
\footnotesize
\begin{aligned}
\max_{I_L}\,\,&\pi_{L}=\frac{1}{9}(2-\frac{1}{9sI^2_L-1})^2+s(\frac{I_L}{9sI^{2}_{L}-1})^2-\gamma I^2_L\\
s.t\,\,&\sqrt{\frac{2}{9s}}\leq I_L.
\end{aligned}
\end{equation}
\end{theorem}

 Let $\Delta = 0$. Theorem~\ref{thm: Un-conclusion-sectionA} follows from  Theorems~\ref{thm: Un-conclusion-3sp}, \ref{thm: BM-prices}, \ref{thm: BM-I_F}, \ref{thm: stage 1}. Theorem~\ref{thm: Un-conclusion-3sp} allows us to consider only  interior SPNE. Parts (1) and (2)
of Theorem~\ref{thm: Un-conclusion-sectionA}  follow respectively from Theorems \ref{thm: stage 1} and \ref{thm: BM-I_F}. Part (3) follows from Theorem~\ref{thm: BM-prices}, part (4) from Theorem~\ref{thm: BM-prices} and (\ref{equ: BM-demand}).

\section{EUs with Outside Options}\label{sec: out}

We now generalize our framework to consider a scenario in which the EUs from the common pool the SPs are competing over,  may not choose either of the two SPs if the service quality-price tradeoff they offer is not satisfactory. In effect, there is an outside option for the EUs. Also, each SP has an exclusive additional customer base which can provide customers beyond the common pool   depending on the service quality and access fees they offer. We introduce these modifications through demand functions we describe next.

\begin{definition}\label{definition: new_demand}
The fraction\footnote{The fraction may be replaced with actual number (of EUs) in this case, by altering scale factors in this expression and in those of the payoffs. Our results hold for both interpretations as we do not use $0 \leq \tilde{n}_{L}, \tilde{n}_{F} \leq 1$ in any derivation.
  We use $0 \leq n_{L}, n_{F} \leq 1$ though.} of EUs  with each SP is
\begin{equation*}
\begin{aligned}
\tilde{n}_{L}=\alpha n_{L}+\tilde{\varphi}_{L}(p_{L}, I_{L}),\ \
\tilde{n}_{F}=\alpha n_{F}+\tilde{\varphi}_{F}(p_{F}, I_{F}),
\end{aligned}
\end{equation*}
where
\begin{equation*}
\begin{aligned}
&\tilde{\varphi}_{L}(p_{L}, I_{L})=k'-\theta' p_{L}+b'(I_{L}-I_{F}),\\
&\tilde{\varphi}_{F}(p_{F}, I_{F})=k'-\theta' p_{F}+b'I_{F}
\end{aligned}
\end{equation*}
and $\alpha>0$, $k'$, $\theta'$ and $b'$ are constants.
\end{definition}

Here, $n_L, n_F$ represent fractional subscriptions from the common pool as before, and are determined in Stage 4 of the sequential game described in Section~\ref{games},  based on the utilities specified in (\ref{equ: BM-utility EUs}), with $v^L = v^F$ for simplicity.  The demand functions $\tilde{\varphi}_L(.,.)$ and $\tilde{\varphi}_F(.,.)$ can be positive or negative. A positive value denotes attracting EUs presumably from an exclusive additional customer base beyond the common pool, and a negative value denotes losing some of the EUs in the common pool  to an outside option. The size of the common pool may be different from the exclusive additional customer bases of the SPs; to account for this disparity, we multiply the fractional subscriptions from the common pool, $n_L, n_F$ with a constant $\alpha.$

Considering $\theta' = \alpha$, for analytical tractability:
\begin{equation}\label{equ: Out-demand}
\begin{aligned}
&\tilde{n}_{L}=\alpha\big(n_{L}+\varphi_{L}(p_{L}, I_{L})\big),\\
&\tilde{n}_{F}=\alpha\big(n_{F}+\varphi_{F}(p_{F}, I_{F})\big),
\end{aligned}
\end{equation}
with $k = k'/\alpha$, $b = b'/\alpha$, and
\begin{equation}
\begin{aligned}
&\varphi_{L}(p_{L}, I_{L})=k-p_{L}+b(I_{L}-I_{F}),\\
&\varphi_{F}(p_{F}, I_{F})=k-p_{F}+bI_{F}
\end{aligned}
\end{equation}

The formulation is the same as in Sections~\ref{sec: BM-independent framework}, \ref{games}, with $\tilde{n}_L, \tilde{n}_F$  replacing $n_L, n_F$ in \eqref{equ: BM-L-payoff} and \eqref{equ: BM-F-payoff}.  Using the  argument that led us to the expression for the  As in Section~\ref{sec: BM-independent framework}, the  EU-resource-cost is $I_F^*/p_F^*+(I_L^*-I_F^*)/p_L^*$, following the argument in the last paragraph of Section~\ref{sec: BM-independent framework}.  We characterize the SPNE strategies in Section~\ref{sec: out-outcome}, and provide numerical results in Section \ref{sec: out-4-numerical}.

\subsection{The SPNE outcome}\label{sec: out-outcome}

 For simplicity, we consider only interior SPNE strategies, that is, $0 < n_L^*, n_F^* < 1$.
 We define functions $f(I_L)$, $g(I_L)$, $\pi_{L}(I_{F})$ and sets $\mathbb{L}_{1}$, $\mathbb{L}_{2}$ as follows:
\begin{equation*}
\footnotesize
\begin{aligned}
&g(I_{L})=\frac{b}{15}I_{L}+\frac{1}{15}-\frac{c}{3}+\frac{k}{3}, \,\, f(I_{L})=\frac{1}{5I_{L}}+\frac{b}{5},\\
&\theta(y)=2\alpha \big(\frac{b}{5}I_{L}+\frac{1}{5}+g(I_{L})-f(I_{L})y\big)^{2}+sy^{2}-\gamma I_{L}^{2},	
\end{aligned}
\end{equation*}
\begin{equation*}
\footnotesize
\begin{aligned}
\mathbb{L}_{1}=&\{s>2\alpha f^{2}(I_{L})+2\alpha f(I_{L})g(I_{L})/I_{L},\, g(I_{L})\geq0,\\
&\delta\leq I_{L}, I_L<4/b\},	
\end{aligned}
\end{equation*}
\begin{equation*}
\footnotesize
\begin{aligned}
\mathbb{L}_{2}=&\{0\leq I_{L}, I_L<4/b\}\cap\Big(\{g(I_{L})\geq0,\\
&\,2\alpha f^{2}(I_{L})\leq s\leq2\alpha f^{2}(I_{L})+2\alpha f(I_{L})g(I_{L})/I_{L}\}\\
\cup&\{2\alpha f^{2}(I_{L})+4\alpha f(I_{L})g(I_{L})/I_{L}\geq s,\,2\alpha f^{2}(I_{L})>s\}\Big).	
\end{aligned}
\end{equation*}
With $\delta<4/b$, we prove in  Appendix~\ref{Appendix: outside option}:
\begin{theorem}
\label{outsideoptiontheorem}
The interior SPNE strategies are:
\begin{itemize}
  \item [(1)]  $I_{L}^{*}$ is characterized in
 \begin{equation*}
\small
\begin{aligned}
I_{L}^{*}=\argmax_{I_{L}}\Big(\max_{I_{L}\in\mathbb{L}_{1}}\theta(\frac{-2\alpha f(I_{L})g(I_{L})}{2\alpha f^{2}(I_{L})-s}),\max_{I_{L}\in\mathbb{L}_{2}}\theta(I_{L})\Big)
\end{aligned}
\end{equation*}

  \item [(2)]  $I_{F}^{*}$ is characterized in
\begin{equation*}
 \begin{aligned}
  I_{F}^{*}=\left\{\begin{aligned}
  &\frac{-2\alpha f(I_{L})g(I_{L})}{2\alpha f^{2}(I_{L})-s}\,\,& \text{if}\,\, I_{L}\in\mathbb{L}_{1}\\
  &I_{L}\,\, & \text{if}\,\, I_{L}\in\mathbb{L}_{2}\end{aligned}\right.
\end{aligned}
\end{equation*}
 \item [(3)]
$p_{L}^{*}=\frac{1}{15}+\frac{2c}{3}+\frac{k}{3}+\frac{I_{L}^{*}-I_{F}^{*}}{5I_{L}^{*}}-\frac{b}{5}I_{F}^{*}+\frac{4b}{15}I_{L}^{*}$,
$p_{F}^{*}=\frac{1}{15}+\frac{2c}{3}+\frac{k}{3}+\frac{I_{F}^{*}}{5I_{L}^{*}}+\frac{b}{15}I_{L}^{*}+\frac{b}{5}I_{F}^{*}$.

\item [(4)] $\tilde{n}_{L}^{*}=\frac{I_{L}^{*}-I_{F}^{*}}{I_{L}^{*}}+p_{F}^{*}-2p_{L}^{*}+k+bI_{L}^{*}-bI_{F}^{*}$,
  $\tilde{n}_{F}^{*}=\frac{I_{F}^{*}}{I_{L}^{*}}+p_{L}^{*}-2p_{F}^{*}+k+bI_{F}^{*}$
  \end{itemize}
\end{theorem}

Remark~\ref{remark2} holds here with Theorem~\ref{outsideoptiontheorem} substituting  Theorem~\ref{thm: Un-conclusion-sectionA}.

Despite the expressions being cumbersome, the characterization is easy to compute, as in Theorem~\ref{thm: Un-conclusion-sectionA},  and lead to important insights, as enumerated below.

\begin{equation*}
\begin{aligned}
&\tilde{n}_{L}^{*}=\frac{3}{5}(1-\frac{I_{F}^{*}}{I_{L}^{*}})+\varphi_{L}(p_{L}, I_{L})+\frac{2b}{5}I_{F}^{*}-\frac{b}{5}I_{L}^{*}\\	
&\tilde{n}_{F}^{*}=1-\frac{3}{5}(1-\frac{I_{F}^{*}}{I_{L}^{*}})+\varphi_{F}(p_{F}, I_{F})-\frac{2b}{5}I_{F}^{*}+\frac{b}{5}I_{L}^{*}
\end{aligned}
\end{equation*}
In both equations, intuitively, the first term, $\frac{3}{5}(1-\frac{I_{F}^{*}}{I_{L}^{*}}),  1-\frac{3}{5}(1-\frac{I_{F}^{*}}{I_{L}^{*}})$, represents the subscription from the common pool, if there had been no attrition to an outside option.
 The second and third terms represent the impacts of the attritions as also the additions from the exclusive customer bases. The first term depends on the degree of cooperation similar to the the base case specified in part (4) of Theorem~\ref{thm: Un-conclusion-sectionA}.   In the special case that $b=0$, i.e., when the demand functions depend only on the access fees,  the third term is $0$ and the demand functions capture the impact of attrition and additions in the SPNE expression for the subscriptions. For $b > 0$, the second and the third term together become $k - p_L^* + \frac{b}{5} I_L^* (4-3 I_F^*/I_L^*)$ in the expression for $\tilde{n}_{L}^{*}$, and
$k - p_F^* + \frac{b}{5} I_L^* (1+ 3 I_F^*/I_L^*)$ in that for $\tilde{n}_{F}^{*}$. Thus, higher degree of cooperation decreases (increases, respectively)  the subscription for SP$_L$ (SP$_F$, respectively) even in these terms, and therefore, overall, like in the base case. Note that the subscriptions represent the efficacy in competition. However, as in the base case, the decrease in subscription  does not directly lead to reduction in overall payoffs of SP$_L$, as the deficit may be compensated through income generated by leasing spectrum to SP$_F.$


\subsection{Numerical results}\label{sec: out-4-numerical}


\begin{figure}
\centering
\includegraphics[width=1.74in]{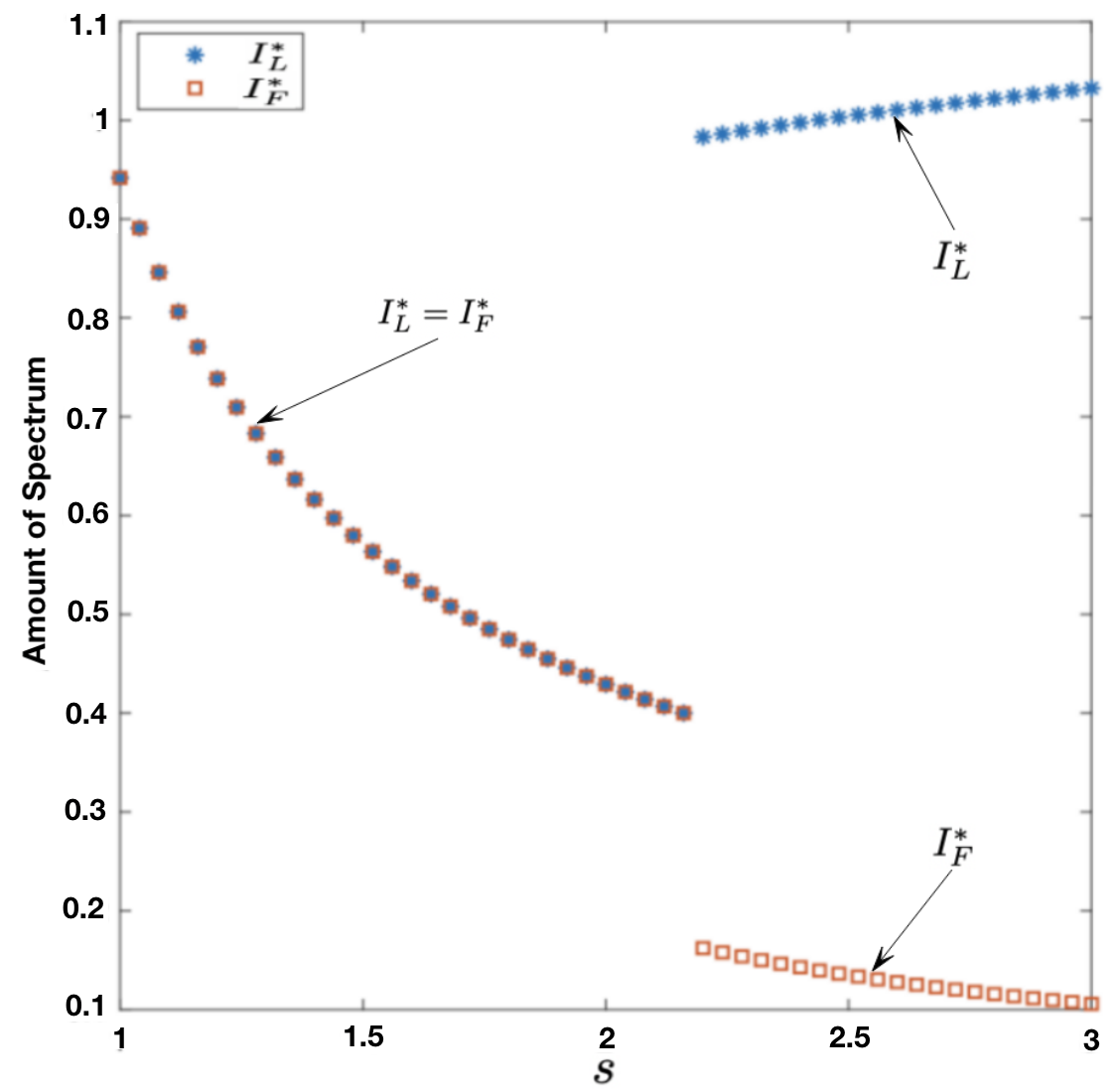}
\includegraphics[width=1.7in]{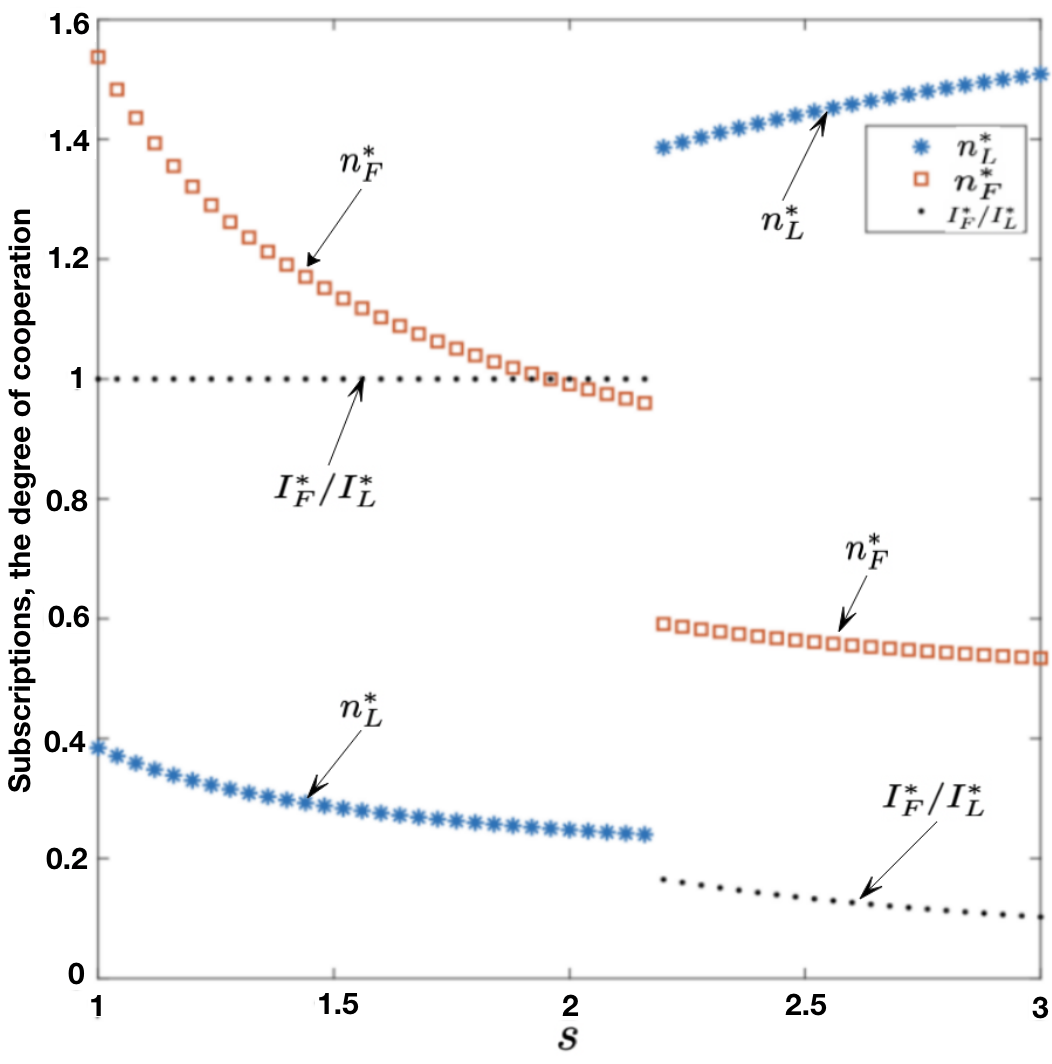}
\caption{Spectrum (left), degree of cooperation and subscriptions (right) vs. $s$ Here, $\gamma=0.8$, $c=k=1$, $b=2$.  }
\label{figOut2}
\end{figure}

 Figure~\ref{figOut2} show that now, both $n_L^*, n_F^*$ can decrease (eg, with changes in $s$)  because of   attrition to the outside option  possibly due to decrease of $I_L^*, I_F^*.$ We note this when $s$ is below a threshold. Otherwise, the trends resemble  Figures~\ref{figBM1} and \ref{figBM2} (the base case).

\begin{figure}
\centering
\includegraphics[width=1.7in,height=1.75in]{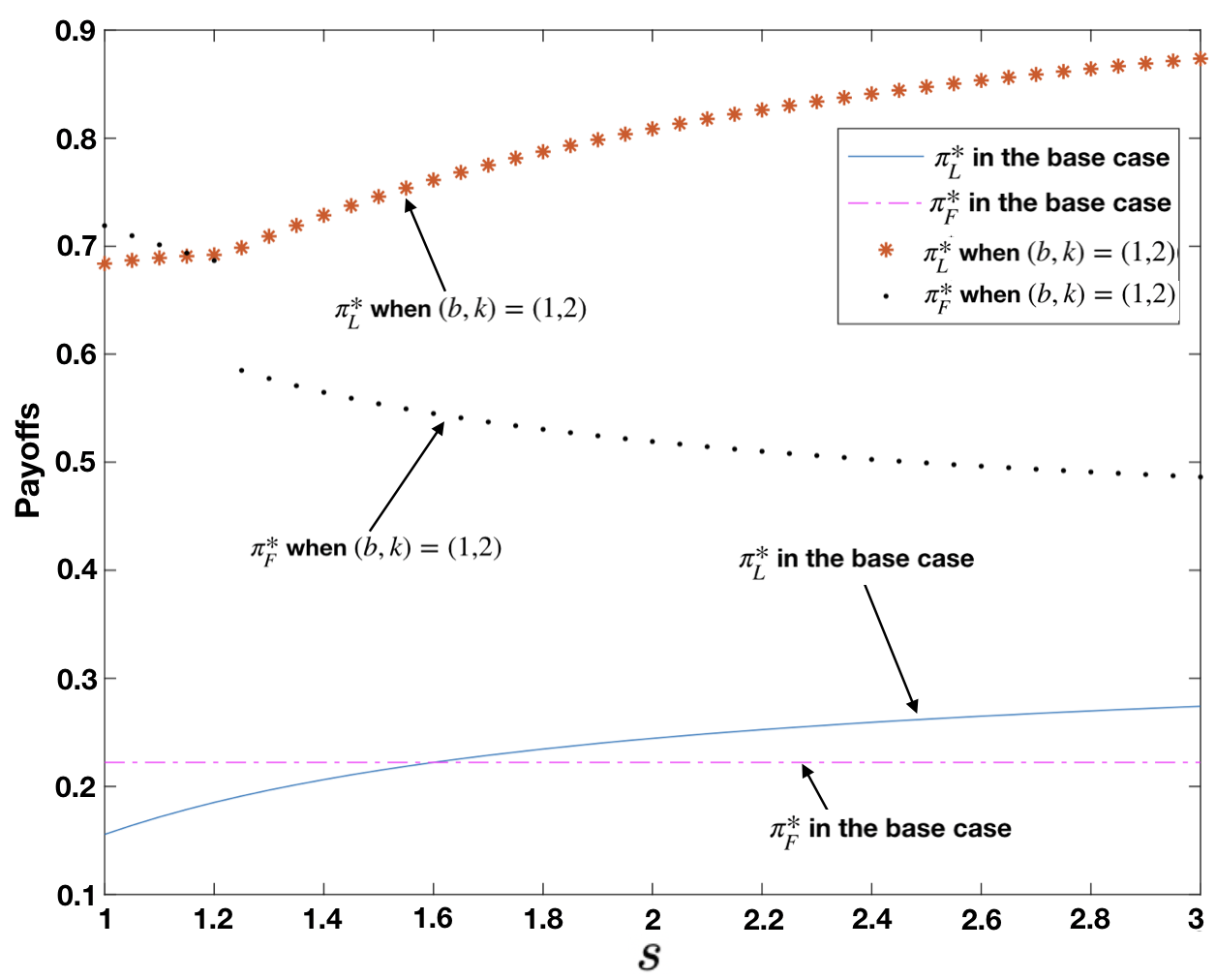}
\includegraphics[width=1.7in,height=1.75in]{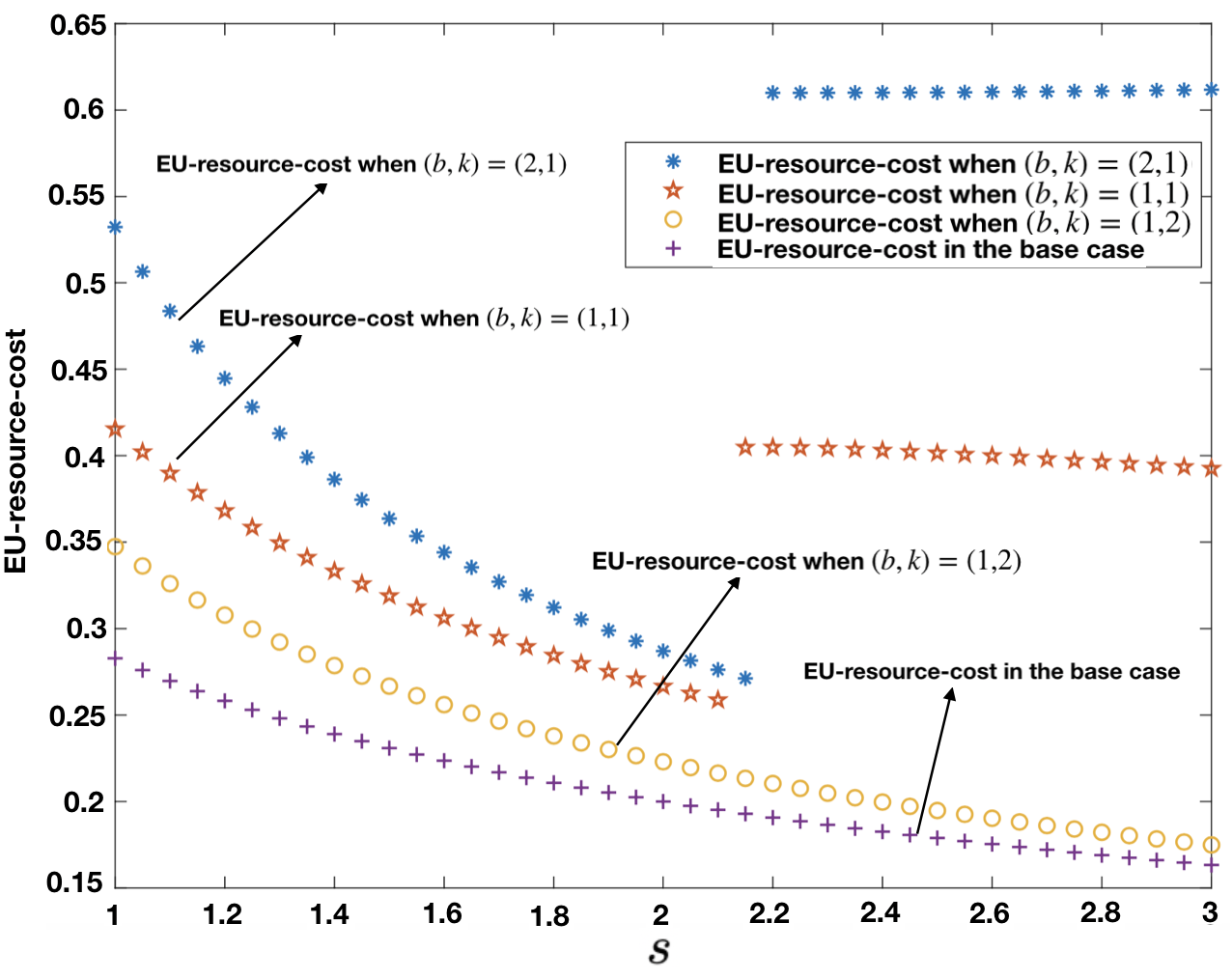}
\caption{Payoffs (left), EU-resource-cost (right) vs. $s$ Here, $\gamma=0.8$, $c=1$.}
\label{figOut3}
\end{figure}
 Figure \ref{figOut3} (left) shows the payoffs under different  $s$. The trends of payoffs are similar with Figure~\ref{figBM1}~(left). 
 The SPs earn higher payoffs than in the base case, as they have additional exclusive customers bases to draw additional EUs from.

 Figure~\ref{figOut3} (right) shows that for different values of the parameters $b, k$,  the EU-resource-cost exceeds that for the base case shown in Figure~\ref{figBM_a1}. This is because the SPs provide better resource-cost tradeoff to the EUs so as not to loose them to the outside option, and also to draw more EUs from their exclusive additional bases.

\section{The 3-player model}\label{sec: 3-player model}
We now generalize our framework to consider competition between MNOs, rather than that only between
 an MNO and an MVNO.
 In a 3-player model, we consider two MNOs and one MVNO competing for a common pool of EUs in a covered market (i.e., each EU needs to opt for exactly one  SP).
  We present the model in Section \ref{sec: 3p-1-model}, and characterize the SPNE  in Section \ref{sec: 3p-2-outcome}.
   We show that the competition among multiple SPs reduces their payoffs, but benefits the EUs:   the SPs acquire higher amounts of spectrum (hence provide higher service quality), and charge the EUs less. The competition also reduces the payoffs of SPs. We prove the results  in Appendix~\ref{Appendix 3p} (Theorems~8, 9) and in Section~\ref{Appendix: Corollary} (Corollary~1).

\subsection{Model}\label{sec: 3p-1-model}
We consider a symmetric model and seek a symmetric equilibrium i.e.,  the strategies of the MNOs  are the same, and the MVNO leases the same amount of spectrum from each MNO. Thus, in the SPNE,   $I_{L}=I_{L_{1}}=I_{L_{2}}$,  $I_{F}=I_{F_{1}}=I_{F_{2}}$, $p_{L}=p_{L_{1}}=p_{L_{2}}$, and $n_{L}=n_{L_{1}}=n_{L_{2}}$.
The total amount spectrum of SPs is $2I_{L}$. Thus, each MNO retains $I_{L}-I_{F}$ spectrum. We define the payoffs of MVNO and MNOs as
\begin{align}
 &\pi_{F}=n_{F}(p_{F}-c)-2sI_{F}^{2}\label{equ: 3p-F-payoff}\\
 &\pi_{L}=n_{L}(p_{L}-c)+sI_{F}^{2}-\gamma I_{L}^{2}\label{equ: 3p-L-payoff}
 \end{align}
\begin{figure}[hpt]
  \centering
  \includegraphics[width=0.275\textwidth]{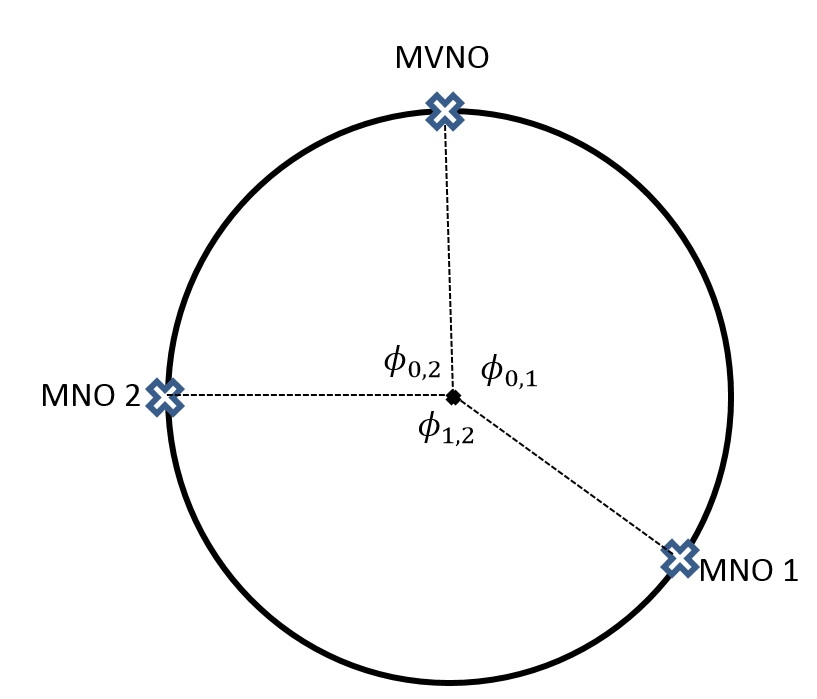}
  \caption{The hoteling model for the three players case}
  \begin{center}
  \label{pic: Modification of the hotelling model for the three players case}
\end{center}
\end{figure}
To accommodate the three SPs, we modify the hotelling model. The EUs are uniformly distributed along a circle of radius 1 on which the SPs are virtually located  (Figure~\ref{pic: Modification of the hotelling model for the three players case}). Since the radius is $1$,  each arc length equals the corresponding angle.
 Thus, the number of EUs located 1) between the MVNO and MNO$_{i}$ is $\phi_{0,i}$ and 2) between the MNOs is $\phi_{1,2}$.

 We consider that $\phi_{0,1}$, $\phi_{0,2}$ and $\phi_{1,2}$ reflect the natural preferences of EUs for SPs (intuitively, for example, those in the arc $\phi_{0,1}$ would have stronger preference for the MVNO and MNO$_{1}$, and so on). We allow the preferences to depend on spectrum investments by defining these arcs as:   $\phi_{0,1}=\phi_{0,2}=h_{1}(I_{L},I_{F})$ and $\phi_{1,2}=h_{2}(I_{L}, I_{F})$ for some functions $h_{1}$ and $h_{2}$ (considering that the model is symmetric). We can now consider the transport cost as a parameter $t>0$ rather than a function of $I_L, I_F$, unlike  in Section \ref{sec: basic case}. We focus on the special case that $v^{L}=v^{F}=v$.

  Similar to (\ref{equ: BM-utility EUs}),
if an EU is located in the arc of $\phi_{0,1}$, at a distance of $x$ from  the MVNO,
\begin{equation}\label{equ: 3p-UEs utility-1}
\begin{aligned}
u_{MVNO}=&v-tx-p_{F}\\
u_{MNO_{1}}=&v-t(\phi_{0,1}-x)-p_{L}\\
u_{MNO_{2}}=&v-t\cdot\min(x+\phi_{0,2}, \phi_{0,1}-x+\phi_{1,2})-p_{L}
\end{aligned}
\end{equation}
By calculation, if $x\leq\phi_{0,1}/2$, then $u_{MNO_{1}}\leq u_{MVNO}$, and $u_{MNO_{2}}=v-t(x+\phi_{0,2})-p_{L}<u_{MVNO}$. Then, EUs choose MVNO. If $x>\phi_{0,1}/2$, then $u_{MVNO}<u_{MNO_{1}}$, and $u_{MNO_{2}}=v-t(\phi_{0,1}-x+\phi_{1,2})-p_{L}<u_{MNO_{1}}$. Then, EUs choose MNO$_{1}$ instead of MNO$_{2}$.

Similarly, due to symmetry, if an EU is located in the arc of $\phi_{0,2}$, he does not choose MNO$_{1}$, and suppose the distance from the EU to the MVNO is $x$, thus
\begin{equation}\label{equ: 3p-UEs Ulitiy-2}
\begin{aligned}
u_{MVNO}=&v-tx-p_{F}\\
u_{MNO_{2}}=&v-t(\phi_{0,2}-x)-p_{L}
\end{aligned}
\end{equation}

If an EU is located in the arc of $\phi_{1,2}$, at a distance of $x$ to the MNO$_{1}$, then his utility is;
\begin{equation}\label{equ: 3p-EUs utility-3}
 \begin{aligned}
u_{MNO_{1}}=&v-tx-p_{L},\\
u_{MNO_{2}}=&v-t(\phi_{1,2}-x)-p_{L}\\
u_{MVNO}=&v-t\cdot\min(x+\phi_{0,1}, \phi_{1,2}-x+\phi_{0,2})-p_{F}
\end{aligned}
\end{equation}
Now we have the following lemma,
\begin{lemma}\label{lem: 3p-utility of EU/MNO/MVNO}
If $p_{L}-p_{F}\geq t\phi_{0,1}$, then all  EUs  choose the MVNO; if $p_{L}-p_{F}<t\phi_{0,1}$, then EUs located in the arc of $\phi_{1,2}$ do not choose the MVNO.
\end{lemma}
Henceforth, we only consider $p_{L}-p_{F}<t\phi_{0,1}$, as:
\begin{theorem}\label{thm: 3p-no NE strategies}
No SPNE strategy exists if $p_{L}-p_{F}\geq t\phi_{0,1}$.
\end{theorem}
Now, from Lemma~\ref{lem: 3p-utility of EU/MNO/MVNO} and the discussion above,  the MVNO and MNO$_{i}$ (MNO$_{1}$ and MNO$_{2}$, respectively) compete to attract the EUs located only on the arc of $\phi_{0,i}$ ($\phi_{1,2}$, respectively). Thus, we define the number of EUs of any two SPs depends only on their total investment levels, i.e., for a constant $\zeta$,
\begin{equation*}
\begin{aligned}
&\phi_{01}=\phi_{02}=\zeta\frac{2I_{F}+I_{L}-I_{F}}{2I_{L}}=\zeta\frac{I_{F}+I_{L}}{2I_{L}},\\
&\phi_{12}=\zeta\frac{2(I_{L}-I_{F})}{I_{L}}=\zeta\frac{I_{L}-I_{F}}{I_{L}}.	
\end{aligned}
\end{equation*}

 \subsection{The SPNE outcome}\label{sec: 3p-2-outcome}

 With $\delta<\frac{\pi}{2}\sqrt{\frac{t}{3s}}$, we prove in  Appendix \ref{Appendix 3p}:

\begin{theorem}\label{cor: 3-p model}
The unique symmetric SPNE strategy, with $I_L^*, p_L^*$ representing the choices of, and $n_L^*$ subscription to,  each MNO, and $I_F^*, p_F^*, n_F^*$ the corresponding quantities for the MVNO, is:
\[I_L^*=I_F^*=\frac{\pi}{2}\sqrt{\frac{t}{3s}}, p_{L}^{*}=p_{F}^{*}=t\pi+c, \ \ n_{F}^{*}=2n_L^* = \pi. \]
\end{theorem}

\begin{remark}
The MVNO leases the entire new spectrum from each MNO.
The degree of cooperation, $I_F^*/I_L^*$ is $1.$ The characterization of the SPNE is easy to compute.
\end{remark}

 We compare the outcome of the 3-player model with the 2-player model, to understand the impact of the competition between the MNOs.  To ensure consistency of comparison,  we  modify the 2-player model of the base case in Section \ref{sec: basic case} as follows:
(1) The transport cost is $t$ instead of $t_{L}=I_{F}/I_{L}$ and $t_{F}=1-t_{L}$.
(2) EUs are distributed uniformly along the interval $[0, 2\pi]$ instead of $[0,1]$, since in the 3-player model, the total amount of EUs is $2\pi$ (3) $v^L = v^F = v$.
By the same analysis method in Section \ref{sec: basic case}, we prove in Appendix~\ref{Appendix: Corollary}:
\begin{corollary}\label{cor: 2-p model}
In the 2-player game formulation, the unique SPNE strategies
 are:
 \[I_{L}^{*}=\delta, \ I_{F}^{*}=0, \ p_{L}^{*}=p_{F}^{*}=2t\pi+c, \ n_{F}^{*}= n_{L}^{*}=\pi .\]
\end{corollary}

Comparing Theorem~\ref{cor: 3-p model} and Corollary~\ref{cor: 2-p model}, we note that due to the competition by an additional MNO, SPs acquire higher amounts of spectrum in the 3-player model, i.e., the two MNOs order additional spectrum, and the MVNO leases the entire new spectrum from each MNO.
The SPs charge the EUs less too:  $t\pi +c$, as opposed to $2t\pi+c$ in the 2-player model.
In both models, the MNO(s) and the MVNO divide the EUs equally: in the 2-player model, each SP has half of the EUs ($\pi$), while in the 3-player model, the MVNO has half of the EUs ($\pi$), and each MNO has a quarter of the EUs ($\pi/2$).

From (\ref{equ: 3p-F-payoff}) and (\ref{equ: 3p-L-payoff}), for 3 players, the payoffs are:  (1) $\frac{5t\pi^{2}}{6}$ for each MNO, and (2)  $\frac{t\pi^{2}}{12}(7-\frac{\gamma}{s})$ for the MVNO.  For 2 players,  the payoffs are  $2t\pi^{2}-\delta^{2}$ and  $2t\pi^{2}$ for the MNO and the MVNO respectively. Thus, clearly (each) MNO secures a higher payoff than the MVNO for both the $3$-player and the $2-$player cases. Also, the SPs earn more in the 2-player model, since fewer SPs compete for the same number of EUs.

Since there are $2$ MNOs and $1$ MVNO now, and the MVNO leases $I_F^*$ amount of spectrum from each MNO, the EU-resource-cost becomes $2I_F^*/p_F^*+2(I_L^*-I_F^*)/p_L^*.$ 

\subsection{Numerical results}\label{sec: 3p-numerical results}
\begin{figure}
\centering
\includegraphics[width=1.55in]{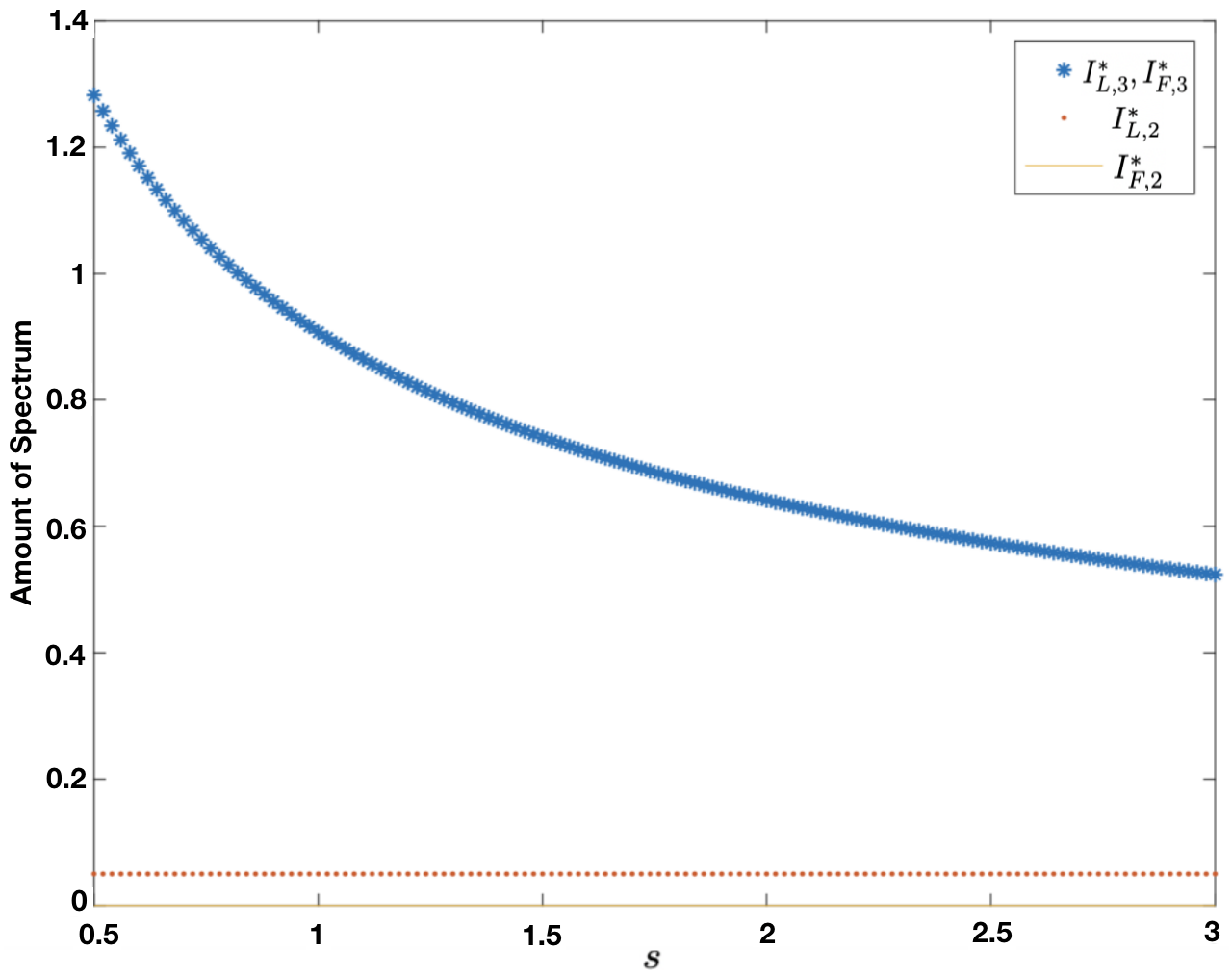}
\includegraphics[width=1.55in]{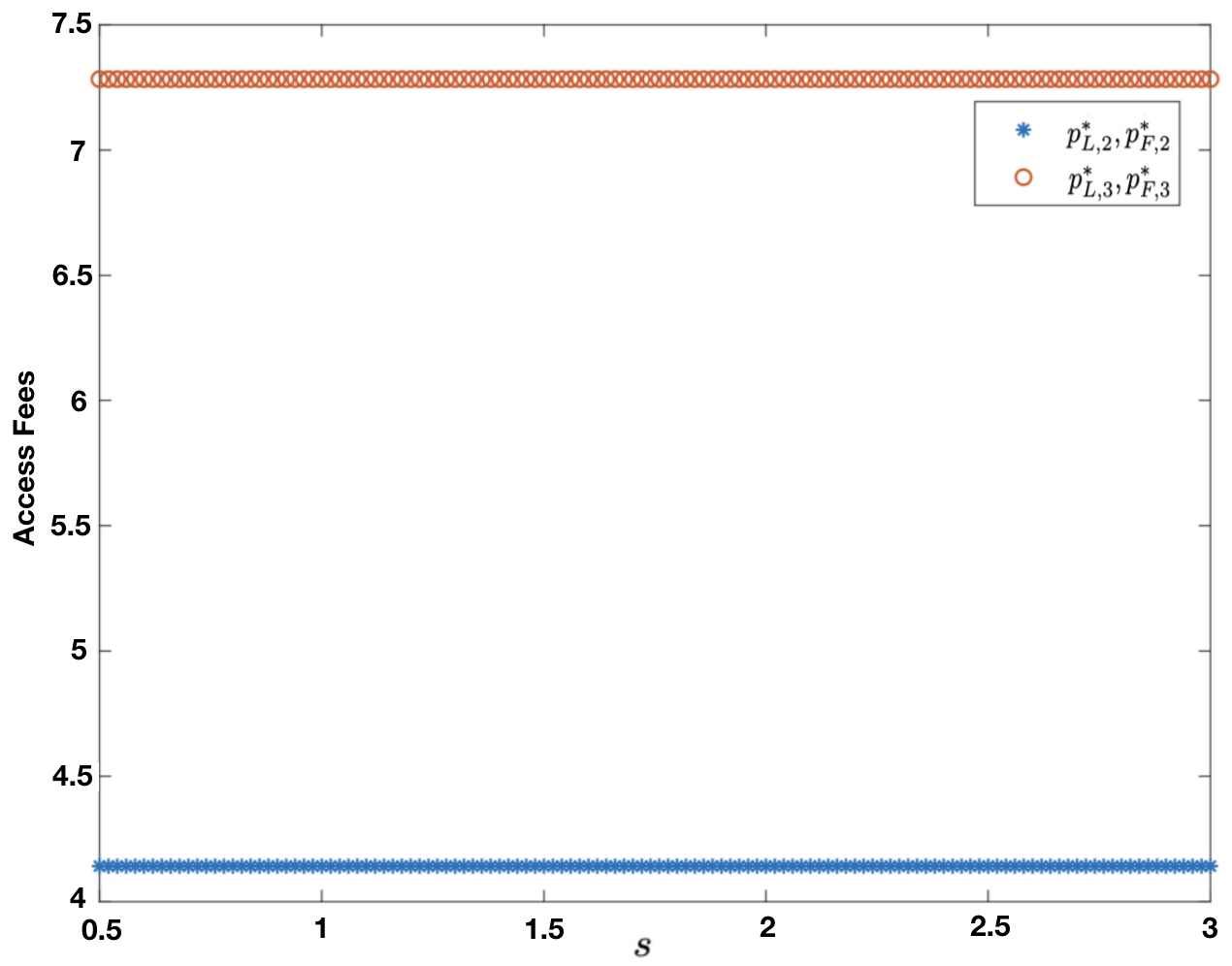}
\caption{Spectrum (left), access fees (right) vs. $s$   }
\label{fig-3player1}
\end{figure}

\begin{figure}
\centering
\includegraphics[width=1.55in]{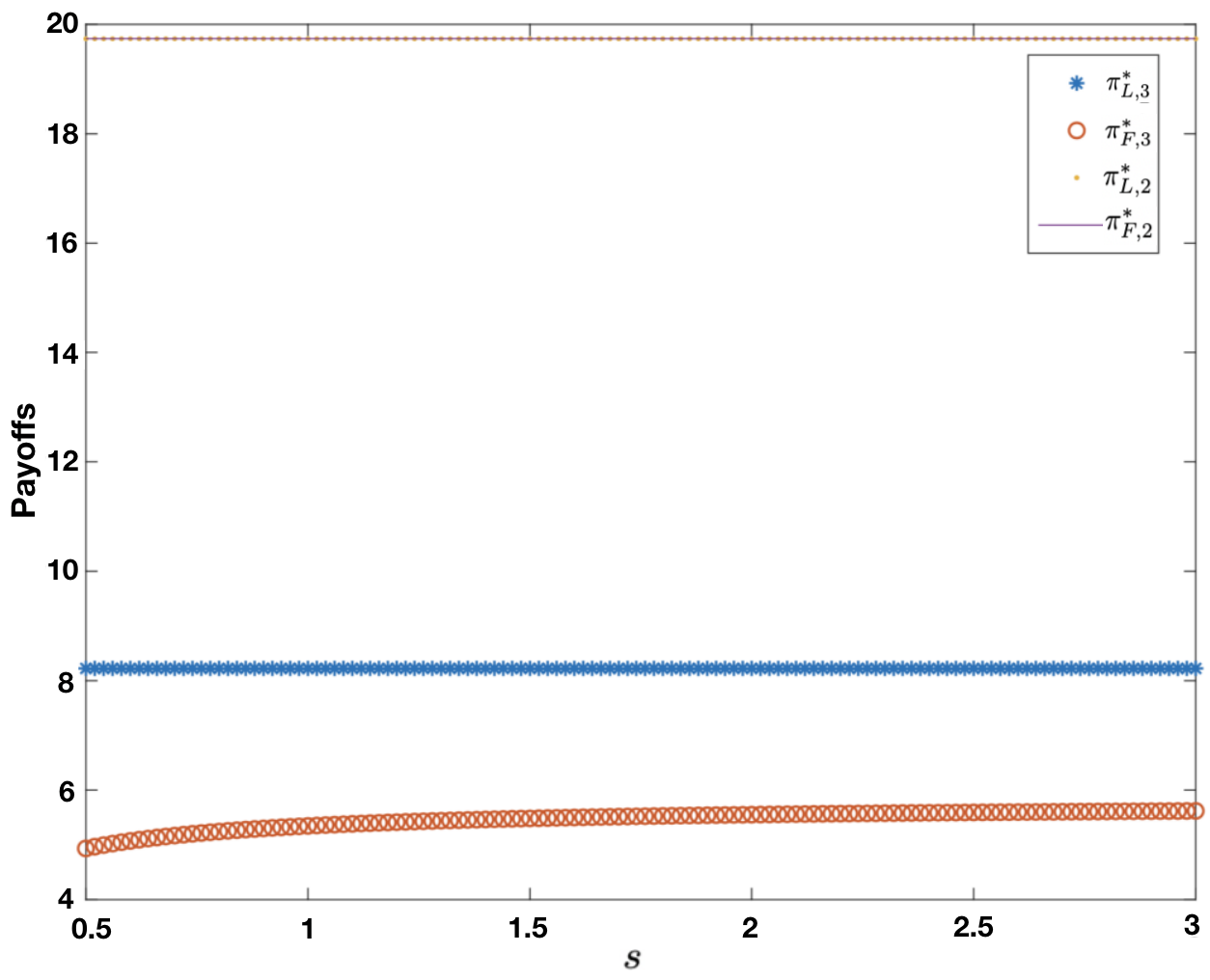}
\includegraphics[width=1.55in]{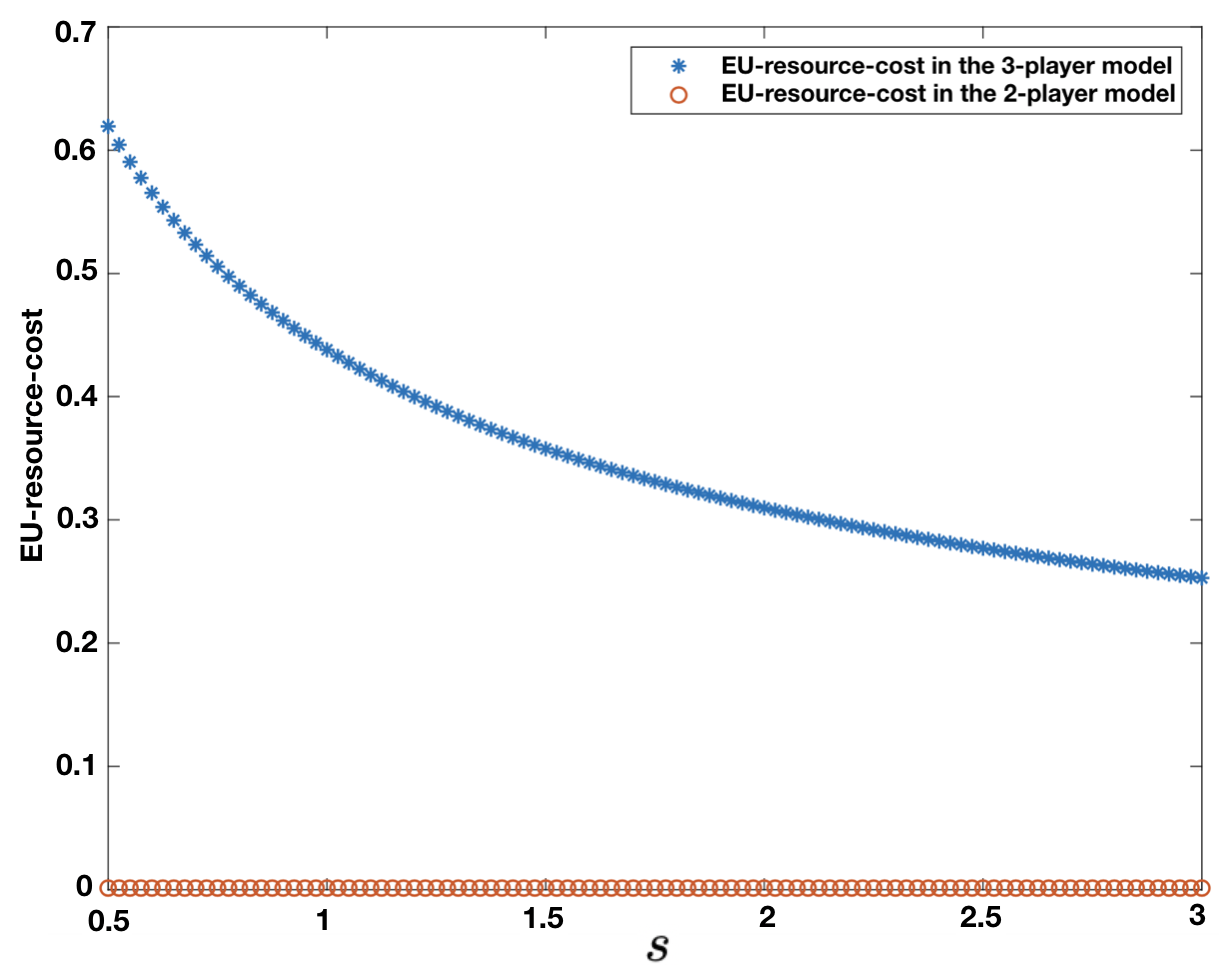}
\caption{Relative payoffs (left), the overall resource per unit price of all subscribers (right) vs. $s$.}
\label{fig-3player2}
\end{figure}

In Figure~\ref{fig-3player1} (left), $I_{L,3}^*, I_{F,3}^*$ (respectively, $I_{L,2}^*, I_{F,2}^*$) are investment levels of SPs in 3-player (respectively, 2-player) model, comparing Theorem \ref{cor: 3-p model} and Corollary \ref{cor: 2-p model}, we note that due to the competition by an additional MNO, SPs acquire higher amounts of spectrum in the 3-player model, i.e., the two MNOs order additional spectrum, and the MVNO leases the entire new spectrum from each MNO. From Figure~\ref{fig-3player1} (right), $p_{L,3}^*, p_{F,3}^*$ (respectively, $p_{L,2}^*, p_{F,2}^*$) are access fees of SPs in 3-player (respectively, 2-player) model, the SPs charge the EUs less too:  $t\pi +c$, as opposed to $2t\pi+c$ in the 2-player model. 

Figure~\ref{fig-3player2} (left) shows that SPs can gain less if an additional MNO enters the system due to the additional competition. 
Figure~\ref{fig-3player2} (right) shows that the EU-resource-cost in the $3$-player model exceeds that in the  base case for 2 SPs shown in Figure~\ref{figBM_a1}. This follows because as noted earlier EUs pay lower access fees and the SPs acquire higher spectrum overall. Thus, like in Section~\ref{sec: out-4-numerical}, the additional competition among the SPs is beneficial for the EUs.

\section{Generalization: Limited Spectrum from the Central Regulator}\label{sec: Limited Spectrum from Central Regulator}
Since we have assumed the spectrum available to the central regulator is limited. A natural assumption is that to set an upper bound to the investment level of SP$_L$, $I_L$. In this section, we assume $\delta\leq I_L\leq M$. Similar with the assumption of $\delta$, $M$ is parameter of choice. After considering the new condition of $I_L$, we characterize the SPNE of the three cases  above as follows. The proofs of Theorems~\ref{thm: Un-conclusion-sectionA, M}, are given in Appendix~\ref{Appendix: Limited Spectrum}.

\subsection{The Base Case} 
\begin{theorem}\label{thm: Un-conclusion-sectionA, M}
Let $|\Delta| < 1$. The SPNE  is: 
\begin{description}
  \item[(1)] If $M\leq\frac{2-\Delta}{9s}$, $I_L^*=I_F^*=M$
$p_L^*-c=n_L^*=\frac{1+\Delta}{3}$, $p_F^*-c=n_F^*=\frac{2-\Delta}{3}$.
\item[(2)] If $M>\frac{2-\Delta}{9s}$, the SPNE are the same as that in Theorem~\ref{thm: Un-conclusion-sectionA}.
\end{description}
\end{theorem}

 \begin{theorem}\label{thm: Un-corner-sectionA, M}
{\bf (1) } $\Delta\geq 1$: The  SPNE  is
the same as that in Theorem~\ref{thm: Un-corner-sectionA}~(1).  \newline
{\bf (2) } $\Delta = 1:$ The following interior strategy constitute an additional  SPNE, if $M\leq\frac{1}{3\sqrt{s}}$, 
\[I_L^*=I_F^*=M, p_L^* - c = n_L^* = 2/3, p_F^* - c = n_F^*= 1/3.\]
If $M>\frac{1}{3\sqrt{s}}$, the  SPNE  is
the same as that in Theorem~\ref{thm: Un-corner-sectionA}~(2).

{\bf (3) } $\Delta\leq -1:$ The SPNE strategy is: If $\delta\leq M\leq\frac{1}{\sqrt{2s}}$, then
 \[I_L^*=I_F^*=M, p_{L}^{*}=p_{F}^{*}+\Delta-1, n_{L}^{*}=0,\, n_{F}^{*}=1.\] 
  If  $ M>\frac{1}{\sqrt{2s}}$, 
the  SPNE  is
the same as that in Theorem~\ref{thm: Un-corner-sectionA}~(3). 
\end{theorem}

From Theorems~\ref{thm: Un-conclusion-sectionA, M} and \ref{thm: Un-corner-sectionA, M}, we can find that if the upper bound $M$ is relative small, the MNO acquires the maximum amount of spectrum from the regulator, and the MVNO leases all spectrum from the MNO.

\subsection{EUs with Outside Options}
For simplicity, we consider only interior SPNE strategies, that is, $0 < n_L^*, n_F^* < 1$.
 We define sets $\mathbb{L}_{1,M}$, $\mathbb{L}_{2,M}$ as follows:
\begin{equation*}
\footnotesize
\begin{aligned}
\mathbb{L}_{1,M}=&\{s>2\alpha f^{2}(I_{L})+2\alpha f(I_{L})g(I_{L})/I_{L},\, g(I_{L})\geq0,\\
&\delta\leq I_{L}\leq M, I_L<4/b\},	
\end{aligned}
\end{equation*}
\begin{equation*}
\footnotesize
\begin{aligned}
\mathbb{L}_{2,M}=&\{0\leq I_{L}\leq M, I_L<4/b\}\cap\Big(\{g(I_{L})\geq0,\\
&\,2\alpha f^{2}(I_{L})\leq s\leq2\alpha f^{2}(I_{L})+2\alpha f(I_{L})g(I_{L})/I_{L}\}\\
\cup&\{2\alpha f^{2}(I_{L})+4\alpha f(I_{L})g(I_{L})/I_{L}\geq s,\,2\alpha f^{2}(I_{L})>s\}\Big).	
\end{aligned}
\end{equation*}
With $\delta<4/b$, we have the following SPNE:
\begin{theorem}\label{thm: EUs with outside option, M}
The interior SPNE strategies are:
\begin{itemize}
  \item [(1)]  $I_{L}^{*}$ is characterized in
 \begin{equation*}
\small
\begin{aligned}
I_{L}^{*}=\argmax_{I_{L}}\Big(\max_{I_{L}\in\mathbb{L}_{1,M}}\theta(\frac{-2\alpha f(I_{L})g(I_{L})}{2\alpha f^{2}(I_{L})-s}),\max_{I_{L}\in\mathbb{L}_{2,M}}\theta(I_{L})\Big)
\end{aligned}
\end{equation*}

  \item [(2)]  $I_{F}^{*}$ is characterized in
\begin{equation*}
 \begin{aligned}
  I_{F}^{*}=\left\{\begin{aligned}
  &\frac{-2\alpha f(I_{L})g(I_{L})}{2\alpha f^{2}(I_{L})-s}\,\,& \text{if}\,\, I_{L}\in\mathbb{L}_{1,M}\\
  &I_{L}\,\, & \text{if}\,\, I_{L}\in\mathbb{L}_{2,M}\end{aligned}\right.
\end{aligned}
\end{equation*}
 \item [(3)]
$p_{L}^{*}=\frac{1}{15}+\frac{2c}{3}+\frac{k}{3}+\frac{I_{L}^{*}-I_{F}^{*}}{5I_{L}^{*}}-\frac{b}{5}I_{F}^{*}+\frac{4b}{15}I_{L}^{*}$,
$p_{F}^{*}=\frac{1}{15}+\frac{2c}{3}+\frac{k}{3}+\frac{I_{F}^{*}}{5I_{L}^{*}}+\frac{b}{15}I_{L}^{*}+\frac{b}{5}I_{F}^{*}$.

\item [(4)] $\tilde{n}_{L}^{*}=\frac{I_{L}^{*}-I_{F}^{*}}{I_{L}^{*}}+p_{F}^{*}-2p_{L}^{*}+k+bI_{L}^{*}-bI_{F}^{*}$,
  $\tilde{n}_{F}^{*}=\frac{I_{F}^{*}}{I_{L}^{*}}+p_{L}^{*}-2p_{F}^{*}+k+bI_{F}^{*}$
  \end{itemize}
\end{theorem}
The proof of Theorem~\ref{thm: EUs with outside option, M} is the same as the proof of Theorem~\ref{outsideoptiontheorem}. Comparing Theorems~\ref{thm: EUs with outside option, M} and \ref{outsideoptiontheorem}, after adding the new condition $\delta\leq I_L\leq M$ on $I_L$, the only change is that the region of $I_L$ is shrinked by the upper bound.

\subsection{The 3-player model}

 With $\delta<\frac{\pi}{2}\sqrt{\frac{t}{3s}}$, we have:
 \begin{theorem}\label{cor: 3-p model, M}
The unique symmetric SPNE strategy, with $I_L^*, p_L^*$ representing the choices of, and $n_L^*$ subscription to,  each MNO, and $I_F^*, p_F^*, n_F^*$ the corresponding quantities for the MVNO, is: 
\begin{description}
  \item[(1)] If $M\leq\frac{\pi}{2}\sqrt{\frac{t}{3s}}$, then
\[I_L^*=I_F^*=M,\ \ p_{L}^{*}=p_{F}^{*}=t\pi+c, \ \ n_{F}^{*}=2n_L^* = \pi.\]
\item[(2)] If $M>\frac{\pi}{2}\sqrt{\frac{t}{3s}}$, the SPNE is the same as that in Theorem~\ref{cor: 3-p model}.   
\end{description}
\end{theorem}
Similar with Theorem~\ref{thm: Un-conclusion-sectionA, M},  if the upper bound $M$ is relative small, the MNO acquires the maximum amount of spectrum from the regulator, and the MVNO leases all spectrum from the MNO.

\section{Conclusions and Future Research Directions}
\label{sec: conclusion}

This paper investigates the incentives of mobile network operators (MNOs) for acquiring
additional  spectrum to offer mobile virtual network operators (MVNOs) and thereby inviting competition
for a common pool of end users (EUs). We consider a base case and two generalizations: (\romannumeral1) one MNO and one MVNO,  (\romannumeral2) one MNO, one MVNO and an outside option, and  (\romannumeral3) two MNOs and one MVNO. We identify metrics \big($I_F^*/I_L^*$ for cooperation between SPs,  $(n_L^*, n_F^*)$ for competition between SPs, $I_F^*/p_F^*+(I_L^*-I_F^*)/p_L^*$ for resource-cost tradeoff of the EUs\big) to quantify the interplay between cooperation and competition.  Four-stage noncooperative sequential games are formulated and SPNE are obtained analytically.

Analytical and numerical results show that higher degree of cooperation can enhance the payoff of both SPs, and increase (respectively, decrease) the competition efficacy of SP$_F$ (respectively, SP$_L$). In addition,  high degree of cooperation coincides with high EU-resource-cost, and provides low access fee options to the EUs. 
Increased competition due to the presence of additional MNOs is beneficial to EUs but reduces the payoffs of the SPs.

All results extend, with some modifications,  when we consider that $I_L$ is upper bounded by $M$. Such bounds may apply when   the central regulator has limited spectrum to offer. In this case, if the upper bound $M$ is relatively small (less than some threshold), in the SPNE, $I_L^* = I_F^* = M$, but otherwise $I_L^*, I_F^*$ characterized in various Theorems apply. The thresholds will in general be different for different cases and have been quantified. The SPNE values of the other decisions variables, namely $p_L^*, p_F^*, n_L^*, n_F^*$ remain as in various Theorems.   Refer to Section~V of the technical report \cite{technicalreport} for the deductions.

Future research includes generalization  to accommodate: 1) non-uniform distribution of EUs between the two SPs in the hotelling model,  2) distinct transaction costs $c_L$ and $c_F$,  3) potentially non-convex  spectrum reservation fee functions that the SP$_F$ pays the SP$_L$ and the SP$_L$ pays the regulator,   4) arbitrary number of MNOs and MVNOs,   5) arbitrary transport cost $t_L, t_F$ functions of the spectrum acquired by the SPs, $I_L, I_F$. We next provide research directions in each.
\begin{enumerate}
 \item If the EUs are non-uniformly distributed in $[0, 1]$, one can start with a cumulative distribution function $F(x)$ which gives the fraction of EUs in $(0, x)$. Starting with the base case and $v^L=v^F$,  in \eqref{equ: BM-demand}, for $x_0 \in (0, 1)$,  $n_L$ will now be $F(x_0)$, where $x_0$ is given by \eqref{equ: BM-indifferent location}, $n_F=1-n_L$ as before. Following the analytical progression in Section~\ref{sec: BM-SPNE Analysis}, the results  must now  be derived using specific expressions for $F(\cdot)$ (eg, Lemma~\ref{lem: BM-stage-3-payoffs}, Theorems~\ref{thm: BM-prices}, \ref{thm: BM-I_F}, \ref{thm: stage 1}). This will in turn help determine how the characteristics of the distribution function $F(\cdot)$   affect the equilibrium closed forms, which currently remains open.
 \item The EUs may incur different amounts of  transaction costs for the SPs, namely $c_F, c_L$ respectively for SP$_F$, SP$_L.$ Starting with the base case, \eqref{equ: BM-indifferent location}, \eqref{equ: BM-demand} continue to hold. But, $c$ need to be replaced by $c_L, c_F$ respectively in the expressions for the payoffs $\pi_L, \pi_F$ in  Lemma~\ref{lem: BM-stage-3-payoffs}. Also, $c$ need to be replaced by  $\frac{2c_F+c_L}{3}, \frac{2c_L+c_F}{3}$ respectively in the expressions for the access fees $ p_L^*, p_F^*$ in Theorem~\ref{thm: BM-prices}. The expressions in Theorems~\ref{thm: BM-I_F}, \ref{thm: stage 1} must now be derived and modified, building on the above modifications. This derivation remains open.
\item Following Remark~\ref{re: convex function}, 
    the SPNE of investment levels ($I_{L}^{*}$, $I_{F}^{*}$) remain open for an arbitrary spectrum reservation fee function that the SP$_F$ pays the SP$_L$ and the SP$_L$ pays the regulator. The analytical methodology used in Theorems~\ref{thm: BM-I_F}, \ref{thm: stage 1} should however apply, though the expressions would depend on the specific function in question.
 \item To obtain the SPNE for arbitrary number of MNOs and MVNOs, one may distribute them on a circle as for 3 SPs (refer to Section~\ref{sec: 3p-1-model} and Figure~\ref{pic: Modification of the hotelling model for the three players case}), and follow the analytical approach presented in Sections~\ref{sec: 3p-1-model}, \ref{sec: 3p-2-outcome}. The limitation of this distribution of SPs on a circle is that a SP can compete for EUs with only $2$ other SPs,  as a SP can have only $2$ adjacent SPs and effectively only a pair of SPs compete for the EUs in the segment of the circumference between them. For $3$ SPs, this is not restrictive, as  each SP anyway has no more than $2$ SPs to compete with, but  it is restrictive for $n$ SPs when $n > 3$ as there in general each SP competes with $n-1$ other SPs. Nonetheless, our circular distribution method provides a foundation for this general problem, by allowing SPNE computation for arbitrary number of SPs when each SP competes for EUs with $2$ predetermined SPs. More innovative topology of placements of SPs involving distributions in potentially higher dimensions may be able to relax this restriction, which remains open.
\item For arbitrary transport cost $t_L, t_F$ functions, the analytical methodologies (eg, Section~\ref{sec: BM-SPNE Analysis} for the base case) would apply. But the derivation of the results remain open.
\end{enumerate}

\appendices

\section{On quadratic function maximization}

\begin{lemma}\label{lem: quadratic function}
Define a quadratic function $f(x)=ax^{2}+bx+c$ with $a\neq0$. The maximum of $f(x)$ in an interval $[d, e] (d<e)$ can be obtained by the following rules:
 \begin{itemize}
 	\item[(1)] If $a>0$, and define the midpoint of the interval $M=\frac{d+e}{2}$, then
 $f_{\max}(x)=f(d)$ if $M<-\frac{b}{2a}$;
 $f_{\max}(x)=f(e)$ if $M\geq-\frac{b}{2a}$.
   \item[(2)] If $a<0$, i.e., $f(x)$ is concave, then
 $f_{\max}(x)=f(d)$ if $d\geq-\frac{b}{2a}$;
 $f_{\max}(x)=f(e)$ if $e\leq-\frac{b}{2a}$;
 $f_{\max}(x)=f(-\frac{b}{2a})$ if $d<-\frac{b}{2a}<e$.
 \end{itemize}
\end{lemma}

\begin{proof}
\noindent{\textbf{(1)}}. Since $a>0$, then $f(x)$ is convex, thus the maximum point can only be obtained at the boundary points, i.e., $x=d$ or $x=e$. Thus,  
\begin{equation}\label{equ: quadratic function-0}
f(d)-f(e)=(a(d+e)+b)(d-e).
\end{equation}

Let $M< -\frac{b}{2a}$. Since $a>0$,
$M< -\frac{b}{2a}\Leftrightarrow \frac{d+e}{2}< -\frac{b}{2a}\Leftrightarrow (d+e)a+b<0$.
Note $d-e<0$, from (\ref{equ: quadratic function-0}), $f(d)-f(e)=(a(d+e)+b)(d-e)>0$, which implies $f_{\max}(x)=f(d)$.
Similarly, if $M\geq-\frac{b}{2a}$, note $a>0$, then $M\geq-\frac{b}{2a}\Leftrightarrow \frac{d+e}{2}\geq-\frac{b}{2a}\Leftrightarrow (d+e)a+b\geq0$.
Since $d-e<0$, then from (\ref{equ: quadratic function-0}),
$f(d)-f(e)=(a(d+e)+b)(d-e)\leq0$, which implies $f_{\max}(x)=f(e)$.

\noindent{\textbf{(2)}.} If $a<0$, then $f(x)$ is concave. Since $f'(x)=2ax+b$, then 1) $f'(x)<0$
and $f(x)$ is decreasing if $x>-\frac{b}{2a}$, 2)  $f'(x)\geq0$ and $f(x)$ is increasing if $x\leq-\frac{b}{2a}$.
\textbf{(\romannumeral1)} If $d\geq-\frac{b}{2a}$, then $f(x)$ is decreasing if $x\in[d, e]$, hence $f_{\max}(x)=f(d)$.
\textbf{(\romannumeral2)} If $e\leq-\frac{b}{2a}$, then $f(x)$ is increasing if $x\in[d, e]$, hence $f_{\max}(x)=f(e)$.
\textbf{(\romannumeral3)} Let $d<-\frac{b}{2a}< e$. Since $f(x)$ is concave, 
thus $f(x)$ has a unique maximum point (stationary point) $x=-\frac{b}{2a}$, i.e., $f(-\frac{b}{2a})\geq f(x)$ for all $x\in\mathbb{R}$. If $[d, f]$ contains  $-\frac{b}{2a}$, i.e., $d\leq-\frac{b}{2a}\leq f$, then $f(-\frac{b}{2a})\geq f(x)$ for all $x\in[d, f]$, hence $f_{\max}(x)=f(-\frac{b}{2a})$.
\end{proof}

\section{Proofs in the Base case when $v^L = v^F$}\label{Appendix-BM}
\noindent{{\bf Proof of Theorem \ref{thm: Un-conclusion-3sp} when $v^L=v^F$}.}
\begin{proof}
Let $(p_{L}^{*},p_{F}^{*},I_{L}^{*},I_{F}^{*})$ be a corner SPNE strategy. Thus, 1) $x_{0} \geq1$, or 2)  $x_{0}\leq0$. We arrive at a contradiction for  1)  {\bf Step 1}  and 2) in {\bf Step 2}
respectively. 
 \begin{lemma}\label{lem1}
 $\pi_{F}^{*}\geq0$. If $n_F^* > 0,$  $p_{F}^{*} \geq c.$
  \end{lemma}
\begin{proof}
Let $\pi_{F}^*<0$. Consider a unilateral deviation in which $I_F = 0, p_F \geq c.$ From \eqref{equ: 3p-F-payoff}, $\pi_F \geq 0$, leading to a contradiction. Now, let $n_F^* > 0$ and  $p_{F}^{*} < c$. Thus, $\pi_{F}^*<0$ which is a contradiction.
\end{proof}

\noindent{\bf Step 1}. Let $x_{0}^*\geq1$. Clearly,  $n_{F}^*=0$ and $n_{L}^*=1$.
From \eqref{equ: BM-F-payoff}, $\pi_{F}^{*}=-sI_{F}^{*2}.$

 From Lemma~\ref{lem1}, $I_F^* = 0.$ Thus, $\pi_F^*=0, t_F^* = 1.$
From \eqref{equ: BM-indifferent location}, $1 \leq x_{0}^*=t_{F}^*+p_{F}^*-p_{L}^* = 1+ p_{F}^*-p_{L}^*$. Thus,  $p_{F}^*\geq p_{L}^*$. 

From \eqref{equ: BM-L-payoff}, $ \pi_{L}^{*}=p_{L}^{*}-c-\gamma I_{L}^{*2}.$
If $p_{L}^{*}<c$, then $\pi_{L}^*<-\gamma \delta^{2} < 0$ since $I_{L}^{*}\geq\delta$. Consider a unilateral deviation by which  $I_{L}=\delta, p_{L}=c$, then
$\pi_{L} = -\gamma\delta^{2}$, which is beneficial for SP$_L$.   Thus, $p_{L}^{*}\geq c$.

Now, let $p_L^* > c.$  Thus, $p_{F}^{*} \geq p_{L}^{*} > c$. Recall that $x_0^* = 1+p_F^*-p_L^*.$
 Consider a unilateral deviation by which $p_F = p_L^*-\epsilon > c$. Now, by \eqref{equ: BM-indifferent location}, $x_0 < 1$, and hence $n_F > 0.$  Now, from \eqref{equ: BM-F-payoff}, $\pi_F > 0 = \pi_F^*$. Thus, $(I_F^*, p_F^*)$ is not  SP$_F$'s best response  to SP$_L$'s choices $(I_L^*, p_L^*)$,    which is a contradiction. Hence, $p_L^* = c.$

Now consider another unilateral deviation of SP$_{L}$, $p_{L}'=p_{F}^{*}+\epsilon$, where $0 < \epsilon < 1$,  with all the rest the same. Since $p_L^* \leq p_F^*$, $p_L' > p_L^* = c.$
\begin{align*}
&n_{L}'=x_{0}'=t_{F}^{*}+p_{F}^*-p_{L}'=1-\epsilon.
\end{align*}

Then
\[\pi_{L}'-\pi_{L}^{*}=n_{L}'(p_{L}'-c) - (p_{L}^{*}-c) =
(1-\epsilon)(p_L' - c)>0.\]
The last inequality follows because $p_L' > c$ and $\epsilon < 1.$ Thus, we again arrive at a contradiction.

\noindent{\bf Step 2}. Let $x_0^* \leq 0.$ Clearly, $n_{F}^*=1, n_{L}^*=0$. Since $n_F^* > 0 $, by Lemma~\ref{lem1},  $p_{F}^{*}\geq c$.
From \eqref{equ: BM-indifferent location}, $x_{0}^*=t_{F}^*+p_{F}^*-p_{L}^*\leq0$. Thus,
  $p_L^* \geq p_F^* + t_F^*.$ Now, from  \eqref{equ: BM-L-payoff},
\begin{align} \label{fallback}
 \pi_{L}^*=sI_{F}^{*2}-\gamma I_{L}^{*2}.
 \end{align}
 Consider a unilateral deviation by SP$_{L}$, by which  $p_{L}'=t_F^* + p_F^* -\epsilon$, $0 < \epsilon < 1$. Then 
\begin{align*}
&n_{L}'=x_{0}'=t_{F}^{*}+p_{F}^*-p_{L}'=\epsilon>0
\end{align*}
Therefore, by \eqref{fallback},
\begin{align*}
\pi_{L}'-\pi_{L}^{*}=n_{L}'(p_{L}'-c)=\epsilon(p_{F}^{*}-\epsilon+t_{F}^{*}-c)
\end{align*}
Since $p_{F}^*\geq c$, either $p_{F}^{*}=c$ or $p_{F}^{*}>c$.
If $p_{F}^{*}>c$, then let $\epsilon<p_{F}^{*}-c$. Then,  $\pi_{L}'-\pi_{L}^{*}>0$.
If $p_{F}^{*}=c$, then $I_{F}^{*}=0$ (otherwise $\pi_{F}^{*}<0$, which by Lemma~\ref{lem1} implies that $p_{F}^{*}$ is not a NE), then $t_{F}^{*}=1$. Thus, $\pi_{L}'-\pi_{L}^{*}>0$. We again arrive at a contradiction. 
\end{proof}

By Theorem~\ref{thm: Un-conclusion-3sp} proved above henceforth we only consider interior SPNE in which $0 < x_0^* < 1.$

\noindent{{\bf Proof of Theorem \ref{thm: BM-I_F} when $v^L = v^F$}.}

\begin{proof}
Substituting $p_{F}$ and $p_{L}$ from (\ref{equ: BM-prices}) into (\ref{equ: BM-interior-payoffs}), using $t_{L}=I_{F}/I_{L}$ and $t_{F}=1-t_{L}$, $\text{SP}_{F}$'s payoff  becomes,
\begin{equation}\label{equ: BM-stage-2-interior-F-payoff}
\pi_{F}(I_{F})=(\frac{1}{9I_{L}^{2}}-s)I_{F}^{2}+\frac{2}{9I_{L}}I_{F}+\frac{1}{9}	
 \end{equation}

Thus, the following maximization yields $I_F^*$: 
\begin{equation}\label{equ: BM-best-inteior-I_F-1}
\begin{aligned}
\max\,\,&\pi_{F}(I_{F})=(\frac{1}{9I_{L}^{2}}-s)I_{F}^{2}+\frac{2}{9I_{L}}I_{F}+\frac{1}{9}\\
s.t\,\,&0\leq I_{F}\leq I_{L}.
\end{aligned}
\end{equation}

\noindent{\textbf{(A)}.} If $I_{L}=\frac{1}{\sqrt{9s}}$, i.e., $\frac{1}{9I_{L}^{2}}-s=0$, $\pi_{F}(I_{F}; I_{L})$ is increasing in $I_{F}.$
  Thus,
$I_{F}^{*}=I_{L}$.

\noindent{\textbf{(B)}.} Let $I_{L}\neq\frac{1}{\sqrt{9s}}$. Referring to
 the terminology of Lemma \ref{lem: quadratic function},  $-b/2a = \frac{I_{L}}{9I_{L}^{2}s-1}$,
 which we denote as $F_1$. 

{\bf (B-1).} Let $I_{L}<\frac{1}{\sqrt{9s}}$, i.e., $1-9I_{L}^{2}s>0$. Then $\pi_{F}$ is a convex function. Note that $I_F \in [0,  I_{L}]$, and the midpoint of the interval is $I_{L}/2$. From Lemma~\ref{lem: quadratic function}, since $1-9I_{L}^{2}s>0$, then $F_{1}<0<I_{L}/2$, $\Rightarrow$ the maximum is obtained at $I_{F}^{*}=I_{L}$.

{\bf (B-2).} Let $I_{L}>\frac{1}{\sqrt{9s}}$, i.e., $1-9I_{L}^{2}s<0$. Then $\pi_{F}$ is a concave function. Note that  $F_{1}=\frac{I_{L}}{9I_{L}^{2}s-1} > 0$. From Lemma \ref{lem: quadratic function}, $0<F_{1}<I_{L}\Leftrightarrow\sqrt{\frac{2}{9s}}<I_{L}$ and $F_{1}\geq I_{L}\Leftrightarrow\frac{1}{\sqrt{9s}}<I_{L}\leq\sqrt{\frac{2}{9s}}$,
thus
\begin{equation*}
I_F^* = \left\{
\begin{aligned}
&F_{1}& \,\,&\text{if}\quad \sqrt{\frac{2}{9s}}<I_{L}\\
&I_{L}& \,\,&\text{if}\quad \frac{1}{\sqrt{9s}}<I_{L}\leq\sqrt{\frac{2}{9s}}
\end{aligned}
\right..
\end{equation*}
Combining \textbf{(A)} and \textbf{(B)}, we obtain (\ref{equ: BM-I_F}).
\end{proof}

\noindent{{\bf Proof of Theorem \ref{thm: stage 1}.}

\begin{proof}
 Substituting $p_{L}$ and $p_{F}$ from (\ref{equ: BM-prices}) into $\pi_{L}$ from (\ref{equ: BM-interior-payoffs}), using $t_{L}=I_{F}/I_{L}$ and $t_{F}=\frac{I_{L}-I_{F}}{I_{L}}$, $\text{SP}_{L}$'s payoff becomes:
\begin{equation}\label{equ: BM-L-payoff1}
\pi_{L}(I_{L})=(\frac{2}{3}-\frac{I_{F}^{*}}{3I_{L}})^{2}+sI_{F}^{*2}-\gamma I_{L}^{2}.
 \end{equation}

Now, the following optimization yields $I_L^*$:
\begin{equation*}
\begin{aligned}
\max_{I_{L}}\quad&\pi_{L}(I_{L})=(\frac{2}{3}-\frac{I_{F}^{*}}{3I_{L}})^{2}+s(I_{F}^{*})^{2}-\gamma I_{L}^{2}\\
s.t\quad&\delta\leq I_L.
\end{aligned}
\end{equation*}
Then, we have the following two sub-cases.

{\textbf{(A)}.} From (\ref{equ: BM-I_F}), if $\delta\leq I_{L}\leq\sqrt{\frac{2}{9s}}$, then $I_{F}^{*}=I_{L}$, thus for $I_L$ in this range, the objective function
of the optimization is $ \frac{1}{9}+(s-\gamma)I_{L}^{2}. $ 
 This is an increasing function of $I_{L}$,
since $s>\gamma.$ Thus
the optimum solution for $I_L \in [\delta, \sqrt{\frac{2}{9s}}]$ is   $\sqrt{\frac{2}{9s}}$. 

{\textbf{(B)}.} Next, if $\sqrt{\frac{2}{9s}}<I_L$, then $I_{F}^{*}=\frac{I_{L}}{9I_{L}^{2}s-1}$.
Since $I_{L}=\frac{I_{L}}{(9I_{L}^{2}s-1)}$ when $I_{L}=\sqrt{\frac{2}{9s}}$, then $I_{F}^{*}$ is continuous at $I_{L}=\sqrt{\frac{2}{9s}}$. So $\pi_{L}(I_{L}; I_{F}^{*})\rightarrow\pi_{L}|_{I_{L}=I_{F}^{*}=\sqrt{\frac{2}{9s}}}$ as $I_{L}\downarrow\sqrt{\frac{2}{9s}}$. Therefore, this case also includes the optimum solution of previous case. Thus substituting $I_{F}^{*}=\frac{I_{L}}{9I_{L}^{2}s-1}$ to (\ref{equ: BM-L-payoff1}),  (\ref{equ: BM-I_L}) is obtained.
\end{proof}	

\section{The proofs in the 3-player model}\label{Appendix 3p}

\noindent{{\bf Proof of Lemma \ref{lem: 3p-utility of EU/MNO/MVNO}}.}

\begin{proof}
First, let $p_{L}-p_{F} \geq t\phi_{0,1}$. Consider EUs in the arc of $\phi_{1,2}$.
Consider an EU  at distance $x$ from MNO$_{1}$. From the symmetry of  MNO$_{1}$ and MNO$_{2}$, 1) if $x \leq \frac{\phi_{1,2}}{2}$, $u_{MNO_{1}} \geq u_{MNO_{2}}$, and 2)  if $x>\frac{\phi_{1,2}}{2}$, $u_{MNO_{2}} \geq u_{MNO_{1}}.$
Since $p_{L}-p_{F}\geq t\phi_{0,1}$, 1) if $x < \frac{\phi_{1,2}}{2}$, then $u_{MNO_{1}}=v-tx-p_{L} < v-t(x+\phi_{0,1})-p_{F}=u_{MVNO}$, and 2) if $x>\frac{\phi_{1,2}}{2}$, then $u_{MNO_{2}}=v-tx-p_{L} <  v-t(x+\phi_{0,1})-p_{F}=u_{MVNO}$. Thus, all the EUs in  arc $\phi_{1,2}$ will choose the MVNO.

Note that $\phi_{0,1}=\phi_{0,2}$. Now consider the EUs in  arc  $\phi_{0,1}$ ($\phi_{0,2}$), at a distance of $x$ from MNO$_1$ (MNO$_2$, respectively). From (\ref{equ: 3p-UEs utility-1}) and (\ref{equ: 3p-UEs Ulitiy-2}), $u_{MNO_{i}}-u_{MVNO}=t\phi_{0,i}-p_{L}+p_{F}-2tx < 0$
since $p_{L}-p_{F}\geq t\phi_{0,1}, x > 0$. Thus all these EUs opt for the MVNO.

Let $p_{L}-p_{F}< t\phi_{0,1}$. One can similarly show that the EUs in arc  $\phi_{1,2}$ choose either MNO$_{1}$ or  MNO$_{2}$.
\end{proof}

\noindent{{\bf Proof of Theorem \ref{thm: 3p-no NE strategies}}.}
\begin{proof}
 Since $I_L^* \geq \delta > 0,$ $\phi_{0,1}^* = \phi_{0,2}^* > 0.$
From  Lemma \ref{lem: 3p-utility of EU/MNO/MVNO},  $n_F^*=2\pi$, and $n_L^*=0$.
 Thus, 
\[ \pi_F^*=2\pi(p_{F}^*-c)-2s(I_{F}^*)^{2}, \pi_L^*=sI_{F}^{*2}-\gamma I_{L}^{*2}. \]

 Let $p_{F}^{*}<c$, then $\pi_F^*<0$. Consider a unilateral deviation of the MVNO,
  by which  $p_{F}=c, I_F=0$. Thus,  $\pi_{F}=0$, and the unilateral deviation is profitable, which is a contradiction. Thus, $p_{F}^{*}=c.$

Thus, since $\phi_{0,1}^*>0$, and  from the condition of the theorem, $
p_{L}^{*} \geq p_{F}^{*}+t\phi_{0,1}^*> c.
$
Consider a unilateral deviation of MNO$_1$, by which $p_{L}' = p_{F}^{*}+t\phi_{0,1} - \epsilon > c, $ with $\epsilon>0$. Now consider the utilities of the EUs in  arc  $\phi_{0,1}$,  at a distance of $x$ from MNO$_1$. From (\ref{equ: 3p-UEs utility-1}),
 $$u_{MNO_{1}}'-u_{MVNO}=t\phi_{0,1}^*-p_{L}'+p_{F}^*-2tx = \epsilon - 2tx.$$
So for $x \in (0, \epsilon/2t)$,  $u_{MNO_{1}}>u_{MVNO}$. Thus  $n_{MNO_{1}}'>0$.

Since $I_{F}^{*}$ and $I_{L}^{*}$ are the same as before, then
$\pi_{MNO_{1}}'=n_{MNO_{1}}'(p_{L}'-c)+sI_{F}^{*2}-\gamma I_{L}^{*2}.$ Thus,

\begin{align*}
\pi_{MNO_{1}}'-\pi_{MNO_{1}}^{*}=n_{MNO_{1}}'(p_{L}'-c) > 0.
\end{align*}
The last inequality follows since $p_{L}'>c$ and $n_{MNO_{1}}'>0$.
  Thus, the unilateral deviation is profitable which leads to a contradiction. 
\end{proof}

\noindent{{\bf Proof of Theorem \ref{cor: 3-p model}}.}

\begin{proof}
\noindent Due to Theorem~\ref{thm: 3p-no NE strategies}, we consider that $p_{L}-p_{F}<t\phi_{0,1}$ henceforth. We sequentially  progress from {\bf Stage~4} to {\bf Stage~1}.

\noindent{\bf{Stage 4:}}
First, we determine the constant $\zeta$.
\begin{lemma}\label{lem: 3S-length of arcs}
$\zeta=\pi$, and $\phi_{0,1}=\phi_{0,2}=\pi\frac{I_{F}+I_{L}}{2I_{L}}$, $\phi_{1,2}=\pi\frac{I_{L}-I_{F}}{I_{L}}$.
\end{lemma}
\begin{proof}
$\phi_{01}+\phi_{02}+\phi_{12}=2\pi$, then $\zeta=\pi$. The rest follows from the definition of $\phi_{01}$, $\phi_{02}$, and $\phi_{12}$. 
\end{proof}

By symmetry, we only consider the split of the EUs between the MNO$_{1}$ and the MVNO.

\begin{theorem}\label{thm: 3p-demand-division}
\begin{equation}\label{equ: 3p-demand of MVNO}
\begin{aligned}
n_{MVNO}=\left\{\begin{aligned}
&0\,&x_{0}\leq0\\
&\pi\frac{I_{F}+I_{L}}{2I_{L}}+\frac{p_{L}-p_{F}}{t}\,&0< x_{0}<\phi_{0,1}\\
&\pi\frac{I_{L}+I_{F}}{I_{L}}\, &x_{0}\geq\phi_{0,1}
\end{aligned}\right.
\end{aligned}
\end{equation}
\begin{equation}\label{equ: 3p-demand of MNOs}
\begin{aligned}
n_{MNO_{1}}=\left\{\begin{aligned}
&\pi\,&x_{0}\leq0\\
&\pi\frac{3I_{L}-I_{F}}{4I_{L}}+\frac{p_{F}-p_{L}}{2t}\,&0< x_{0}<\phi_{0,1}\\
&\pi\frac{I_{L}-I_{F}}{2I_{L}}\, &x_{0}\geq\phi_{0,1}\end{aligned}\right.
\end{aligned}
\end{equation}
where $x_{0}=\frac{\phi_{0,1}}{2}+\frac{p_{L}-p_{F}}{2t}$.
\end{theorem}
\begin{proof}
 Suppose $x_{0}$ is the indifferent location of joining MVNO and MNO1, then:
\begin{equation}\label{equ: 3p-x0}
\begin{aligned}
&v-tx_{0}-p_{F}=v-t(\phi_{0,1}-x_{0})-p_{L}\\
\Rightarrow&x_{0}=\frac{\phi_{0,1}}{2}+\frac{p_{L}-p_{F}}{2t}.
\end{aligned}
\end{equation}
Let $x_{MVNO,MNO2}, x_{MNO1, MNO2}$ be the indifferent locations between 1) MVNO and MNO$_2$, and 2) MNO$_1$  and MNO$_2$ respectively. Then,
$x_{MVNO,MNO2}=\frac{\phi_{0,2}}{2}+\frac{p_{L}-p_{F}}{2t}$, and $x_{MNO1, MNO2}=\frac{\phi_{1,2}}{2}$.
The number of EUs per unit length  to be normalized to one,  $n_{MVNO}$ equals $x_{0}+x_{MVNO,MNO_2}$ if $0< x_{0}<\phi_{0,1}$, $0$ if $x_{0}\leq0$, and $\phi_{0,1}+\phi_{0,2}$ if $x_{0}\geq\phi_{0,1}$. From the symmetry of the game, $x_{MVNO,MNO_2}=x_0$. Now,  (\ref{equ: 3p-demand of MVNO}) follows from Lemma~\ref{lem: 3S-length of arcs}.

 Next, $n_{MNO_1}$ and $n_{MNO_2}$ equal $(\phi_{0,1}-x_{0})+x_{MNO_1,MNO_2}$ if $0<x_{0}<\phi_{0,1}$, $\phi_{0,1}+x_{MNO_1,MNO_2}$ if $x_{0}\leq0$, and $x_{MNO_1,MNO_2}$ if $x_{0}\geq\phi_{0,1}$. Similarly,
  (\ref{equ: 3p-demand of MNOs}) follows.
\end{proof}

\noindent{{\bf Stage 3:}}
Now we characterize the SPNE  access fees.

\begin{theorem}\label{Thm: 3p-interior-prices}
The SPNE access fees of EUs of SPs, $(p_{F}^{*}, p_{L}^{*})$ by which $0<x_{0}<\phi_{0,1}$, is:
\begin{align}\label{equ: 3p-interior-prices}
p_{F}^{*}=\frac{t\pi}{3}\frac{I_{F}+5I_{L}}{2I_{L}}+c,\, p_{L}^{*}=\frac{t\pi}{3}\frac{7I_{L}-I_{F}}{2I_{L}}+c.
\end{align}
\end{theorem}
\begin{proof}
Substituting (\ref{equ: 3p-demand of MVNO}) and (\ref{equ: 3p-demand of MNOs}) into (\ref{equ: 3p-F-payoff}) and (\ref{equ: 3p-L-payoff}),
\begin{equation}\label{equ: 3p-interior-F-payoff}
\pi_{F}=(\pi\frac{I_{F}+I_{L}}{2I_{L}}+\frac{p_{L}-p_{F}}{t})(p_{F}-c)-2sI_{F}^{2}
\end{equation}
\begin{equation}\label{equ: 3p-interior-L-payoff}
\pi_{L}=(\pi\frac{3I_{L}-I_{F}}{4I_{L}}+\frac{p_{F}-p_{L}}{2t})(p_{L}-c)+sI_{F}^{2}-\gamma I_{L}^{2}	
\end{equation}
 $p_{F}^{*}$ and $p_{L}^{*}$ should be determined to satisfy the first order condition, i.e., $\frac{\pi_{F}}{dp_{F}}|_{p_{F}^{*}}=0$ and $\frac{\pi_{L}}{dp_{L}}|_{p_{L}^{*}}=0$, thus
$p_{F}^{*}=\frac{t\pi}{3}\frac{I_{F}+5I_{L}}{2I_{L}}+c,\, p_{L}^{*}=\frac{t\pi}{3}\frac{7I_{L}-I_{F}}{2I_{L}}+c$.
Therefore, $p_{F}^{*}$ and $p_{L}^{*}$ are the unique interior SPNE strategies if 1) they yield $0<x_{0}<\phi_{0,1}$ and $p_{L}-p_{F}\leq t\phi_{0,1}$, and 2) no unilateral deviation is profitable for SPs.
We establish these in  Parts A and B respectively.

\noindent{\textbf{Part A}.} Substituting $p_{L}^{*}$ and $p_{F}^{*}$ into (\ref{equ: 3p-x0}),
$x_{0}=\frac{\phi_{0,1}}{2}+\frac{p_{L}-p_{F}}{2t}=\pi(\frac{5}{12}+\frac{I_{F}}{12I_{L}})\in(0, \phi_{0,1})$,
since $0\leq I_{F}\leq I_{L}$ $I_{L}>0$.
Also,
$p_{L}-p_{F}=\frac{t\pi}{3}\frac{I_{L}-I_{F}}{I_{L}}<\frac{t\pi}{2}\frac{I_{L}+I_{F}}{I_{L}}=t\phi_{0,1}$.

\noindent{\textbf{Part B}.} Since
$\frac{d^{2}\pi_{F}}{d(p_{F}^{*})^{2}}=-\frac{2}{t}<0,\, \frac{d^{2}\pi_{L}}{d(p_{L}^{*})^{2}}=-\frac{1}{t}<0$,
then $p_{L}^{*}$ and $p_{F}^{*}$ are the unique maximal solutions of $\pi_{L}$ and $\pi_{F}$, respectively for $0 < x_0 < \phi_{0,1}$. Similar to the proof of Theorem \ref{thm: BM-prices},  any deviation by SPs such that  $x_{0}\leq0$ or $x_{0}\geq\phi_{0,1}$ (which yields $n_L=1, n_F=0$ and $n_L=0, n_F=1$, respectively) is not profitable.
 \end{proof}

\noindent{\bf Stage 2:}
We characterize the spectrum SP$_{F}$ acquires from SP$_{L}$ in the SPNE.

\begin{theorem}\label{thm: 3p-interior-I_F}
 $I_{F}^{*}$ is given by:
\begin{equation}\label{equ: 3p-I_F_star}
\begin{aligned}
I_F^* = \left\{\begin{aligned}
&\frac{5t\pi^{2}I_{L}}{72I_{L}^{2}s-t\pi^{2}}\quad&\text{if}&\quad I_{L}\geq\frac{\pi}{2}\sqrt{\frac{t}{3s}}\\
&I_{L}\quad&\text{if}&\quad\delta\leq I_{L}<\frac{\pi}{2}\sqrt{\frac{t}{3s}}
\end{aligned}\right.
\end{aligned}	
\end{equation}
\end{theorem}
\begin{proof}
 $I_{F}^{*}$ is obtained as the optimum solution of
\begin{equation}\label{equ: 3p-interior-pi_F}
\begin{aligned}
\max_{I_{F}}\,\pi_{F}
=&(\frac{t\pi^{2}}{36I_{L}^{2}}-2s)I_{F}^{2}+\frac{5t\pi^{2}}{18I_{L}}I_{F}+\frac{25t\pi^{2}}{36}\\
s.t\quad&0\leq I_{F}\leq I_{L}
\end{aligned}
\end{equation}

The objective function follows from substituting (\ref{equ: 3p-interior-prices}) into (\ref{equ: 3p-interior-F-payoff}). The constraints come from the model assumptions directly.

\noindent{\textbf{(A)}.} Let $I_{L}=\frac{\pi}{6}\sqrt{\frac{t}{2s}}$. Then $\pi_{F}$ is increasing in $I_{F}$, as
$\pi_{F}=\frac{5t\pi^{2}}{18I_{L}}I_{F}+\frac{25t\pi^{2}}{36}.$ 
Thus $I_{F}^{*}=I_{L}$.

\noindent{\textbf{(B)}.} Let $I_{L} \neq  \frac{\pi}{6}\sqrt{\frac{t}{2s}}$.Referring to the terminology of Lemma \ref{lem: quadratic function},   $(-b/2a) = -\frac{\frac{5t\pi^{2}}{18I_{L}}}{2(\frac{t\pi^{2}}{36I_{L}^{2}}-2s)}=\frac{5t\pi^{2}I_{L}}{72I_{L}^{2}s-t\pi^{2}}$.
We denote this quantity as $F_1$.

\noindent{\textbf{(B-1)}.} Let $I_{L}<\frac{\pi}{6}\sqrt{\frac{t}{2s}}$. Then $\pi_{F}$ is convex.   $I_F \in [0, I_{L}]$. 
Since $\frac{t\pi^{2}}{36I_{L}^{2}}-2s>0$, then $72sI_{L}^{2}-t\pi^{2}<0$, thus $F_{1}<0<\frac{I_{L}}{2}$. From Lemma \ref{lem: quadratic function},
$I_{F}^{*}=I_{L}$.

\noindent{\textbf{(B-2)}.} Let $I_{L}>\frac{\pi}{6}\sqrt{\frac{t}{2s}}$, i.e., $\frac{t\pi^{2}}{36I_{L}^{2}}-2s<0$, then $\pi_{F}$ is concave, and   $F_{1}=\frac{5t\pi^{2}I_{L}}{72I_{L}^{2}s-t\pi^{2}}>0$. From Lemma \ref{lem: quadratic function},
\begin{equation*}
I_F^* = \begin{aligned}
\left\{\begin{aligned}
&\frac{5t\pi^{2}I_{L}}{72I_{L}^{2}s-t\pi^{2}}&\quad\text{if}&\quad I_{L}\geq\frac{\pi}{2}\sqrt{\frac{t}{3s}}\\
&I_{L}&\quad\text{if}&\quad \frac{\pi}{6}\sqrt{\frac{t}{2s}}<I_{L}<\frac{\pi}{2}\sqrt{\frac{t}{3s}}
\end{aligned}\right.
\end{aligned}
\end{equation*}
The desired results come from \textbf{(A)}, \textbf{(B)} and \textbf{(C)}.
\end{proof}

\noindent{\bf{Stage 1:}} We characterize the spectrum SP$_{L}$ acquires from the regulator in the SPNE.

\begin{theorem}\label{thm: 3p-payoff-L-sectionA}
Any solution to the following maximization problem constitutes $I_L^*$,
\begin{equation}\label{equ: 3p-L-star}
\footnotesize
\begin{aligned}
\max_{I_{L}}\quad&\pi_{L}=\frac{t\pi^{2}}{18}(\frac{7I_{L}-\frac{5t\pi^{2}I_{L}}{72I_{L}^{2}s-t\pi^{2}}}{2I_{L}})^{2}+s(\frac{5t\pi^{2}I_{L}}{72I_{L}^{2}s-t\pi^{2}})^{2}-\gamma I_{L}^{2}\\
s.t\quad&\frac{\pi}{2}\sqrt{\frac{t}{3s}}\leq I_L.
\end{aligned}
\end{equation}
\end{theorem}
\begin{proof}
 Each MNO chooses its $I_L$ as the solution of the following maximization: \begin{equation}\label{equ: 3p-pi_L-proof1}
\begin{aligned}
\max_{I_{L}}\quad&\pi_{L}(I_{L})=\frac{t\pi^{2}}{18}(\frac{7I_{L}-I_{F}^{*}}{2I_{L}})^{2}+sI_{F}^{*2}-\gamma I_{L}^{2}\\
s.t\quad&\delta\leq I_{L}.
\end{aligned}
\end{equation}
The objective function follows by substituting (\ref{equ: 3p-interior-prices}) into (\ref{equ: 3p-interior-L-payoff}). The constraint follows from the modeling assumption. 

We consider two cases separately: A) $\delta\leq I_{L}\leq\frac{\pi}{2}\sqrt{\frac{t}{3s}}$ and B)
$I_{L}>\frac{\pi}{2}\sqrt{\frac{t}{3s}}$.

{\bf (A).} From (\ref{equ: 3p-I_F_star}), if $\delta\leq I_{L}\leq\frac{\pi}{2}\sqrt{\frac{t}{3s}}$, then $I_{F}^{*}=I_{L}$, thus the objective function of  (\ref{equ: 3p-pi_L-proof1}) is $\frac{t\pi^{2}}{2}+(s-\gamma)I_{L}^{2}.$
This is an increasing function of $I_{L}$ since $s>\gamma.$ Thus the optimum solution in this range is $\frac{\pi}{2}\sqrt{\frac{t}{3s}}$.

{\bf (B).} Next, if $I_{L}>\frac{\pi}{2}\sqrt{\frac{t}{3s}}$, then $I_{F}^{*}=\frac{5t\pi^{2}I_{L}}{72I_{L}^{2}s-t\pi^{2}}$, thus $\pi_{L}(I_{L}, I_{F}^{*})=\pi_{L}(I_{L}, \frac{5t\pi^{2}I_{L}}{72I_{L}^{2}s-t\pi^{2}})$.
Note that $I_{L}=\frac{5t\pi^{2}I_{L}}{72I_{L}^{2}s-t\pi^{2}}$ when $I_{L}=\frac{\pi}{2}\sqrt{\frac{t}{3s}}$, then $I_{F}^{*}$ is continuous at $I_{L}=\frac{\pi}{2}\sqrt{\frac{t}{3s}}$. So $\pi_{L}(I_{L}; I_{F}^{*})\rightarrow\pi_{L}|_{I_{F}^{*}=\frac{\pi}{2}\sqrt{\frac{t}{3s}}}$ as $I_{L}\rightarrow\frac{\pi}{2}\sqrt{\frac{t}{3s}}$. Therefore, this case also includes the optimum solution of previous case. Substituting $I_{F}^{*}=\frac{5t\pi^{2}I_{L}}{72I_{L}^{2}s-t\pi^{2}}$ into  \eqref{equ: 3p-pi_L-proof1}, we get (\ref{equ: 3p-L-star}).
\end{proof}

\begin{theorem}\label{thm: 3p-payoff-L-sectionA2}
$I_{L}^{*}=I_{F}^{*}=\frac{\pi}{2}\sqrt{\frac{t}{3s}}$.
\end{theorem}
\begin{proof}
From (\ref{equ: 3p-L-star}), we have $\pi_{L}(I_{L})=\frac{t\pi^{2}}{18}(\frac{7I_{L}-\frac{5t\pi^{2}I_{L}}{72I_{L}^{2}s-t\pi^{2}}}{2I_{L}})^{2}+s(\frac{5t\pi^{2}I_{L}}{72I_{L}^{2}s-t\pi^{2}})^{2}-\gamma I_{L}^{2}\triangleq f_{1}(I_{L})+f_{2}(I_{L})+f_{3}(I_{L})$, where
$f_{1}(I_{L})=\frac{t\pi^{2}}{18}(\frac{7}{2}-\frac{5t\pi^{2}}{144I_{L}^{2}s-2t\pi^{2}})^{2}$,
$f_{2}(I_{L})=s(\frac{5t\pi^{2}I_{L}}{72I_{L}^{2}s-t\pi^{2}})^{2}$, and
$f_{3}(I_{L})=-\gamma I_{L}^{2}$. Now we take the derivatives of $f_{1}$, $f_{2}$, and $f_{3}$ with respect to $I_{L}$,
$\frac{d\pi_{L}}{d I_{L}}=f_{1}'(I_{L})+f_{2}'(I_{L})+f_{3}'(I_{L})=\frac{10t^{2}\pi^{4}sI_{L}^{2}}{(72I_{L}^{2}s-t\pi^{2})^{3}}\times19\cdot(t\pi^{2}-144I_{L}^{2}s)-2\gamma I_{L}$.
Since $I_{L}\geq \frac{\pi}{2}\sqrt{\frac{t}{3s}}$, then $t\pi^{2}\leq 12I_{L}^{2}s$, thus $72I_{L}^{2}s-t\pi^{2}\geq0$ and $t\pi^{2}-144I_{L}^{2}s\leq0$, which implies $\frac{d f_{1}}{d I_{L}}+\frac{d f_{2}}{d I_{L}}\leq0$. $\frac{d f_{3}}{d I_{L}}=-2\gamma I_{L}<0$, therefore $\frac{d\pi_{L}}{d I_{L}}<0$
so $\pi_{L}$ is a decreasing functions of $I_{L}$, so $I_{L}^{*}=\frac{\pi}{2}\sqrt{\frac{t}{3s}}$. In addition, $\pi_{L}^{*}=\frac{t\pi^{2}}{2}+(s-\gamma)I_{L}^{*}>0$, and
$I_{F}^{*}=\frac{5t\pi^{2}I_{L}^{*}}{72I_{L}^{*2}s-t\pi^{2}}=\frac{\pi}{2}\sqrt{\frac{t}{3s}}=I_{L}^{*}$.
\end{proof}

Theorem~\ref{cor: 3-p model} follows from Theorems~\ref{thm: 3p-demand-division}, \ref{Thm: 3p-interior-prices}, \ref{thm: 3p-payoff-L-sectionA2}.

\end{proof}

\section*{
\bf Supplementary Proofs}

\section{SPNE Analysis of Basic Case}\label{Appendix: Unequal}

If SP$_{L}$ invests in the minimum new spectrum, i.e., $I_{L}=\delta$, and set $p_{L}=c$, then 
$$\pi_{L}=s I_{F}^{2}-\gamma\delta^{2}.$$ Thus for any Nash equilibrium (NE) strategy $(I_{L}^{*}, p_{L}^{*})$, we have
\begin{align*}
\pi_{L}^*|_{p_{L}^{*},I_{L}^{*}}\geq-\gamma\delta^{2}.
\end{align*}

 If SP$_{F}$ leases no new spectrum from SP$_{L}$, then $\pi_{F}=0$. So for any NE strategy $(I_{F}^*,p_{F}^*)$, we have 
\begin{align*}
\pi_{F}^*|_{p_{F}^{*},I_{F}^{*}}\geq0.
\end{align*}

\noindent{\bf Stage 4:}
We first characterize the equilibrium division of EUs between SPs, i.e., $n_L^*$ and $n_F^*$, using the knowledge of the strategies chosen by the  SPs in Stages 1$\sim$3.

\begin{theorem}\label{thm: Un-A-indiffenrent location}
The indifferent location between the two service providers is
\begin{align}\label{equ: Un-A-indifferent location}
x_{0}=\Delta+t_{F}+p_{F}-p_{L}.
\end{align}
\end{theorem}
\begin{proof}
From Definition \ref{def: BM-indifferent location},
\begin{align*}
&u_{F}(x_{0})=v^{F}-t_{F}(1-x_{0})-p_{F}\\
&=v^{L}-t_{L}x_{0}-p_{L}=u_{L}(x_{0}).
\end{align*}

Note $t_{L}+t_{F}=1$, then
\begin{align*}
x_{0}=&\frac{\Delta+t_{F}+p_{F}-p_{L}}{t_{L}+t_{F}}\\
=&\Delta+t_{F}+p_{F}-p_{L}.
\end{align*}
\end{proof}
The fraction of EUs with each SP ($n_{L}$ and $n_{F}$) is:
\begin{equation}\label{equ: Un-A-demand}
\begin{aligned}
&n_{L}=\left\{\begin{aligned}
&0,&\,\text{if}\quad&x_{0}\leq0\\
&x_{0},&\,\text{if}\quad&0<x_{0}<1\\
&1,&\,\text{if}\quad&x_{0}\geq1\\
\end{aligned}\right.\\
&n_{F}=1-n_{L},
\end{aligned}
\end{equation}
where $x_{0}$ is defined in (\ref{equ: Un-A-indifferent location}).

\subsection{The interior SPNE}\label{Appendix: the interior SPNE}
In this section, we consider the interior SPNE ($0<n_F, n_L<1$), and the corner SPNE \big($(n_L, n_F)=(1,0)$ or $(0,1)$\big) are considered in Appendix~\ref{sec: corner SPNE}.

\noindent{\bf Stage 3:}
SP$_L$ and SP$_F$ determine their prices for EUs, $p_L$ and $p_F$, respectively, to maximize their payoffs.

\begin{lemma}\label{lem: Un-A-stage-3-payoffs}
The utility functions of SPs are
\begin{equation}\label{equ: Un-A-interior-payoffs}
\begin{aligned}
\pi_{L}=&(\Delta+t_{F}+p_{F}-p_{L})(p_{L}-c)+sI_{F}^{2}-\gamma I_{L}^{2}\\
\pi_{F}=&(-\Delta+t_{L}+p_{L}-p_{F})(p_{F}-c)-sI_{F}^{2}.
\end{aligned}
\end{equation}

\end{lemma}
\proof
From (\ref{equ: Un-A-demand}), substituting $(n_{L}, n_{F})=(\Delta+t_{F}+p_{F}-p_{L}, 1-n_{L})$ into (\ref{equ: BM-F-payoff}) and (\ref{equ: BM-L-payoff}), we get (\ref{equ: Un-A-interior-payoffs}).
\qed

In the following theorem, we characterize the SPNE access fees of SPs. 
\begin{theorem}\label{thm: Un-A-prices}
 The interior SPNE access fees $p_{L}^{*}$, $p_{F}^{*}$ are
\begin{equation}\label{equ: Un-A-interior-prices}
\begin{aligned}
p_{L}^{*}=&c+\frac{2}{3}-\frac{I_{F}}{3I_{L}}+\frac{\Delta}{3}\\
p_{F}^{*}=&c+\frac{1}{3}+\frac{I_{F}}{3I_{L}}-\frac{\Delta}{3},
\end{aligned}
\end{equation}
and $(p_{L}^{*}, p_{F}^{*})$ are unique if and only if
\begin{align}\label{equ: Un-A-interior-condition}
\Delta-1<\frac{I_{F}}{I_{L}}<\Delta+2.
\end{align}
\end{theorem}

\begin{proof}
We complete the proof in two steps: we first obtain equilibrium access fees $(p_{L}^{*}, p_{F}^{*})$ (\textbf{Step~1}); then we get the condition (\ref{equ: Un-A-interior-condition}) and prove that $p_{L}^{*}$ and $p_{F}^{*})$ are the unique Nash equilibrium access fees of SP$_{L}$ and SP$_{F}$, respectively (\textbf{Step~2}).

\noindent{\textbf{Step 1}.} Consider a SPNE, every Nash equilibrium $(p_{L}^{*},p_{F}^{*})$ should satisfy the first order condition. Get $\pi_{F}$ and $\pi_{L}$ from (\ref{equ: Un-A-interior-payoffs}), then $p_{L}^{*}$ and $p_{F}^{*}$ should be solved by
$$\frac{d\pi_{L}}{dp_{L}}|_{p_{L}^{*}}=0,\, \frac{d\pi_{F}}{dp_{F}}|_{p_{F}^{*}}=0.$$ Note that $t_{L}+t_{F}=1$, then
\begin{align*}
p_{L}^{*}=&c+\frac{2}{3}-\frac{I_{F}}{3I_{L}}+\frac{\Delta}{3}\\
p_{F}^{*}=&c+\frac{1}{3}+\frac{I_{F}}{3I_{L}}-\frac{\Delta}{3}.
\end{align*}

\noindent{\textbf{Step 2}.} In this step, we prove that the $p_{F}^{*}$ and $p_{L}^{*}$ are the unique maximum solutions (in (\textbf{A})). Then, we prove that the condition  (\ref{equ: Un-A-interior-condition}) is sufficient and necessary (in (\textbf{B})). Finally, we show that $p_{F}^{*}$ and $p_{L}^{*}$ are Nash equilibrium by proving that no unilateral is profitable for SPs (in (\textbf{C})).

\noindent{\textbf{(A)}.} Taking the second derivative of $\pi_{L}$ $(\pi_{F})$ with respect to $p_{L}^{*}$ $(p_{F}^{*})$,
$$\frac{d^{2}\pi_{L}}{d(p_{L}^{*})^{2}}=\frac{d^{2}\pi_{F}}{d(p_{F}^{*})^{2}}=-2<0,$$
then $p_{L}^{*}$ and $p_{F}^{*}$ are the unique maximal solutions of $\pi_{L}$ and $\pi_{F}$, respectively.

\noindent{\textbf{(B)}.} Substituting (\ref{equ: Un-A-interior-prices}) into (\ref{equ: Un-A-demand}), we have 
\begin{align*}
x_{0}=&\frac{\Delta}{3}+\frac{2I_{L}-I_{F}}{3I_{L}}=\frac{\Delta}{3}+\frac{2}{3}-\frac{I_{F}}{3I_{L}},
\end{align*}
thus 
\begin{equation}\label{equ: Un-condition}
\begin{aligned}
&0<x_{0}=\frac{\Delta}{3}+\frac{2}{3}-\frac{I_{F}}{3I_{L}}<1\\
\Leftrightarrow&\Delta-1<\frac{I_{F}}{I_{L}}<\Delta+2.
\end{aligned}
\end{equation}
From (\ref{equ: Un-condition}), $0<x_{0}<1$ if and only if (\ref{equ: Un-A-interior-condition}) holds. Therefore if (\ref{equ: Un-A-interior-condition}) does not hold,  then $x_{0}\leq0$ or $x_{0}\geq1$, which implies $n_{L}=0, n_{F}=1$ or $n_{L}=1, n_{F}=0$.

\noindent{\textbf{(C)}.} Since $\frac{d^2 \pi_F}{dp^2_F}<0, \frac{d^2 \pi_L}{dp^2_L}<0$,  a local maxima is also a global maximum, and any solution to the first order conditions  maximize the payoffs   when $0<x_0<1$,  and no unilateral deviation by which $0<x_0<1$ would be profitable for the SPs.
Now, we show that unilateral deviations of the SPs leading to  $n_L=0, n_F=1$ and $n_L=1, n_F=0$  is not profitable. Note that the payoffs of the SPs, \eqref{equ: BM-L-payoff} and \eqref{equ: BM-F-payoff}, are continuous  as $n_L\downarrow 0$, and $n_L\uparrow 1$ (which subsequently yields  $n_F\uparrow 1$ and $n_F\downarrow 0$, respectively). Thus, the payoffs of both SPs when selecting $p_L$ and $p_F$ as the  solutions of the first order conditions are greater than or equal to the payoffs when $n_L=0$ and $n_L=1$. Thus, the unilateral deviations under consideration are not profitable for the SPs.
\end{proof}

\begin{corollary}\label{remark: Un-region}
No corner SPNE access fees exist if $(I_{F}, I_{L})\in R$, where
\begin{equation}\label{equ: R1}
\begin{aligned}
R=&\{\delta\leq I_{L}, 0\leq I_{F}\leq I_{L}\}\\
\cap&\{\Delta-1<I_{F}/I_{L}<\Delta+2\}.
\end{aligned}
\end{equation}
\end{corollary}
\begin{proof}
From Theorem~\ref{thm: Un-A-prices}, if (\ref{equ: Un-A-interior-condition}) holds, then no corner SPNE access fees $(p_{L}^{*},p_{F}^{*})$ exist. Note that $\delta\leq I_{L}\leq M$ and $0\leq I_{F}\leq I_{L}$, combining with (\ref{equ: Un-A-interior-condition}), we obtain the desired results.
\end{proof}

Based on the results in Theorem \ref{thm: Un-A-prices}, we can obtain the payoffs of SPs as follows,
\begin{lemma}\label{lem: Un-A-stage-2-interior-F-payoff}
The payoff of $\text{SP}_{F}$ is
\begin{equation}\label{equ: Fpayoff10}
\begin{aligned}
\pi_{F}(I_{F})=&(\frac{1}{9I_{L}^{2}}-s)I_{F}^{2}+\frac{2(1-\Delta)}{9I_{L}}I_{F}+\frac{(1-\Delta)^{2}}{9}.
\end{aligned}
\end{equation}
\end{lemma}
\begin{proof}
First, we consider interior equilibrium strategies, 
from (\ref{equ: Un-A-interior-payoffs}) in Lemma~\ref{lem: Un-A-stage-3-payoffs} , we have 
\begin{align*}
\pi_{F}=(t_{L}+p_{L}-p_{F}-\Delta)(p_{F}-c)-sI_{F}^{2}.
\end{align*}
Note that $t_{L}=I_{F}/I_{L}$ and $t_{F}=1-t_{L}$.

\noindent{\textbf{(\romannumeral1)}.} Calculate $t_{L}+p_{L}-p_{F}-\Delta$. Substituting $p_{F}$ and $p_{L}$ in (\ref{equ: Un-A-interior-prices}) into $t_{L}+p_{L}-p_{F}-\Delta$, we have
\begin{align*}
&t_{L}+p_{L}-p_{F}-\Delta\\
=&-\Delta+t_{L}+\frac{I_{L}-I_{F}}{3I_{L}}+\frac{\Delta}{3}-\frac{I_{F}}{3I_{L}}-\frac{-\Delta}{3}\\
=&\frac{1-\Delta}{3}+\frac{I_{F}}{3I_{L}}.
\end{align*}

\noindent{\textbf{(\romannumeral2)}.} Calculate $p_{F}-c$. Substituting $p_{F}$ in (\ref{equ: Un-A-interior-prices}) into $p_{F}-c$, we have
\begin{align*}
p_{F}-c=&c+\frac{1}{3}+\frac{I_{F}}{3I_{L}}+\frac{-\Delta}{3}-c=\frac{1-\Delta}{3}+\frac{I_{F}}{3I_{L}}.
\end{align*}
From \textbf{(\romannumeral1)} and \textbf{(\romannumeral2)}, we can obtain (\ref{equ: Fpayoff10}).
\end{proof}

\begin{lemma}\label{lem: Un-A-stage-1-interior-L-payoff}
The payoff of $\text{SP}_{L}$ is
  \begin{equation}\label{equ: Un-L-payoff1}
 \begin{aligned}
\pi_{L}(I_{L})=(\frac{\Delta+2}{3}-\frac{I_{F}}{3I_{L}})^{2}+s(I_{F})^{2}-\gamma I_{L}^{2}.
\end{aligned}
\end{equation}
\end{lemma}
\begin{proof}
From (\ref{equ: Un-A-interior-payoffs}), we have
\begin{align*}
\pi_{L}(I_{L})=&(\Delta+t_{F}+p_{F}-p_{L})(p_{L}-c)+sI_{F}^{2}-\gamma I_{L}^{2}.
\end{align*}

\noindent{\textbf{(\romannumeral1)}.} Calculate $\Delta+t_{F}+p_{F}-p_{L}$. Note that $t_{L}=I_{F}/I_{L}$ and $t_{F}=\frac{I_{L}-I_{F}}{I_{L}}$. From (\ref{equ: Un-A-interior-prices}), then
\begin{align*}
&\Delta+t_{F}+p_{F}-p_{L}\\
=&\Delta+t_{F}
+(c+\frac{1}{3}+\frac{I_{F}}{3I_{L}}+\frac{-\Delta}{3})\\
-&(c+\frac{1}{3}+\frac{I_{L}-I_{F}}{3I_{L}}+\frac{\Delta}{3})\\
=&\frac{\Delta}{3}+t_{F}+\frac{2I_{F}-I_{L}}{3I_{L}}=\frac{\Delta+2}{3}-\frac{I_{F}}{3I_{L}}.
\end{align*}

\noindent{\textbf{(\romannumeral2)}.} Calculate $p_{L}-c$. From (\ref{equ: Un-A-interior-prices}), 
\begin{align*}
&p_{L}-c=c+\frac{1}{3}+\frac{I_{L}-I_{F}}{3I_{L}}+\frac{\Delta}{3}-c\\
=&\frac{1}{3}+\frac{I_{L}-I_{F}}{3I_{L}}+\frac{\Delta}{3}=\frac{\Delta+2}{3}-\frac{I_{F}}{3I_{L}}.
\end{align*}
From \textbf{(\romannumeral1)} and \textbf{(\romannumeral2)}, we get (\ref{equ: Un-L-payoff1}).
\end{proof}

Based on the proof of Theorem~\ref{thm: Un-A-prices}, the existence of equilibria are showed in the following statement:

In Stage 2 and Stage 1, we characterize the optimum investment levels $I_{L}^{*}$ and $I_{F}^{*}$ of SPs. 
 To analyze easily, we consider 4 cases: $-1<\Delta<1$ ({\bf Case~A}), $1\leq \Delta<2$ ({\bf Case~B}), $-2<\Delta\leq-1$ ({\bf Case~C}), and $|\Delta|\geq2$ ({\bf Case~D}).

\noindent{\bf Case~A: $-1<\Delta<1$}

In this section, we consider $-1<\Delta<1$. First, we show that if a SPNE exists when $-1<\Delta<1$, then it must be an interior SPNE (in Proposition~\ref{pro: pro-sectionA}). Then, we characterize the unique optimum  $I_{F}^{*}$ (in Theorem~\ref{thm: Un-A-I_F-sectionA}) and an optimum $I_{L}^{*}$ (in Theorem \ref{thm: Un-payoff-L-sectionA}), respectively. Finally, we collect the optimum strategies in \textbf{Stages 1$\sim$4}, and prove that this strategiy $(p_{L}^{*}, p_{F}^{*}, I_{L}^{*}, I_{F}^{*})$ is an interior Nash equilibrium strategy.

\begin{proposition}\label{pro: pro-sectionA}
If a SPNE exists when $-1<\Delta<1$, then it is an interior SPNE.
\end{proposition}
\proof
From Corollary~\ref{remark: Un-region}, no corner SPNE access fees exist if $(I_{L},I_{F})\in R$.
Note that $-1<\Delta<1$,  then 
$$\Delta-1<0\leq I_{F}/I_{L}\leq1<\Delta+2.$$ Thus from (\ref{equ: R1}), 
$$R=\left\{\delta\leq I_{L}\leq M, 0\leq I_{F}\leq I_{L}\right\}.$$ So (\ref{equ: Un-A-interior-condition}) holds for any $\delta\leq I_{L}\leq M$ and $0\leq I_{F}\leq I_{L}$ when $-1<\Delta<1.$
\qed

\noindent{\bf Stage 2:}
SP$_F$ decides on  the amount of spectrum to be leased from SP$_L$ ($I_F$), with the condition that $0\leq I_F\leq I_L$,  to maximize $\pi_F$. From the model assumptions, $\delta$ is small, then let $\delta<\min(\sqrt{\frac{2-\Delta}{9s}},\frac{1}{\sqrt{9s}})$. 

\begin{theorem}\label{thm: Un-A-I_F-sectionA}
If  $-1<\Delta<1$, then the optimum investment level of $\text{SP}_{F}$, $I_{F}^{*}$, is
\begin{align}\label{equ: Un-I_F1}
  I_{F}^{*}=\left\{\begin{aligned}&\frac{(1-\Delta)I_{L}}{9I_{L}^{2}s-1}\,& I_{L}>\sqrt{\frac{2-\Delta}{9s}}\\&I_{L}\,& \delta\leq I_{L}\leq\sqrt{\frac{2-\Delta}{9s}}\end{aligned}\right..
\end{align}
\end{theorem}
\proof

From (\ref{equ: Fpayoff10}) and Proposition~\ref{pro: pro-sectionA}, the optimal investment level of $\text{SP}_{F}$, $I_{F}^{*}$, is a solution of the following optimization problem,
\begin{equation}\label{equ: 1S-best-inteior-I_F-1}
\small
\begin{aligned}
\max\quad&\pi_{F}(I_{F})=(\frac{1}{9I_{L}^{2}}-s)I_{F}^{2}+\frac{2(1-\Delta)}{9I_{L}}I_{F}+\frac{(1-\Delta)^{2}}{9}\\
s.t\quad&0\leq I_{F}\leq I_{L}
\end{aligned}
\end{equation}

\noindent{\textbf{(A)}.} If $I_{L}=\frac{1}{\sqrt{9s}}$,  then $\pi_{F}(I_{F}; I_{L})$ is a linear function of $I_{F}$, i.e.,
$$\pi_{F}(I_{F})=\frac{2(1-\Delta)}{9I_{L}}I_{F}+\frac{(1-\Delta)^{2}}{9}.$$ Since $-1<\Delta<1$, then $\frac{2(1-\Delta)}{9I_{L}}>0,$
$\pi_{F}(I_{F}; I_{L})$ is an increasing function of $I_{F}$, so
$I_{F}^{*}=I_{L}$.

\noindent{\textbf{(B)}.} If $I_{L}\neq\frac{1}{\sqrt{9s}}$ and $\pi_{F}$ is a quadratic function. We discuss the optimal solutions in two cases: (\romannumeral1) $\delta\leq I_{L}<\frac{1}{\sqrt{9s}}$, and (\romannumeral2) $I_{L}>\frac{1}{\sqrt{9s}}$. We denote $F_{1}$ as
\begin{align}\label{equ: Un-symmetric1}
\frac{d\pi_{F}}{dI_{F}}|_{I_{F}=F_{1}}=0\Rightarrow F_{1}=\frac{(1-\Delta)I_{L}}{9I_{L}^{2}s-1}.
\end{align}

\noindent{\textbf{(B-1)}.} If $\delta\leq I_{L}<\frac{1}{\sqrt{9s}}$, then $\pi_{F}$ is a convex function. Since $I_{F}\in[0,I_{L}]$, then the midpoint is $I_{L}/2$. Note that $-1<\Delta<1$ and $1-9I_{L}^{2}s>0$, thus $$F_{1}=\frac{(1-\Delta)I_{L}}{9I_{L}^{2}s-1}<0<I_{L}/2.$$ From Lemma~\ref{lem: quadratic function}, the maximum is obtained at $I_{F}^{*}=I_{L}$.

\noindent{\textbf{(B-2)}.} If $I_{L}>\frac{1}{\sqrt{9s}}$, then $\pi_{F}$ is a concave function. Note that $-1<\Delta<1$ and $1-9I_{L}^{2}s<0$, then $$F_{1}=\frac{(1-\Delta)I_{L}}{9I_{L}^{2}s-1}>0.$$ From Lemma~\ref{lem: quadratic function},
\begin{align*}
\left\{
\begin{aligned}
&I_{F}^{*}=F_{1}& \,\,&\text{if}\quad 0<F_{1}<I_{L}\\
&I_{F}^{*}=I_{L}& \,\,&\text{if}\quad F_{1}\geq I_{L}
\end{aligned}
\right..
\end{align*}
By simple calculation, 
\begin{align*}
&0<F_{1}<I_{L}\Leftrightarrow\sqrt{\frac{2-\Delta}{9s}}<I_{L}\\
&F_{1}\geq I_{L}\Leftrightarrow\frac{1}{\sqrt{9s}}<I_{L}\leq\sqrt{\frac{2-\Delta}{9s}},
\end{align*}
thus
\begin{align*}
\left\{
\begin{aligned}
&I_{F}^{*}=F_{1}& \,\,&\text{if}\quad \sqrt{\frac{2-\Delta}{9s}}<I_{L}\\
&I_{F}^{*}=I_{L}& \,\,&\text{if}\quad \frac{1}{\sqrt{9s}}<I_{L}\leq\sqrt{\frac{2-\Delta}{9s}}
\end{aligned}
\right..
\end{align*}
From \textbf{(A)} and \textbf{(B)}, we obtain (\ref{equ: Un-I_F1}). Given $v^{L}$, $v^{F}$, $s$ and $I_{L}$, $I_{F}^{*}$ is the unique maximum of $\pi_{F}$, so no unilateral deviation is beneficial for SP$_{F}$. 
\qed

\noindent{\bf Stage~1:}
 SP$_L$ decides on the amount of spectrum $I_L$ acquired from the regulator to maximize $\pi_L$.

\begin{theorem}\label{thm: Un-payoff-L-sectionA}
If $-1<\Delta<1$, then the optimal investment of $\text{SP}_{L}$, $I_{L}^{*}$ is 
a solution of the following optimization problem:
\begin{equation}\label{equ: Un-A-interior-payoff-L-2}
\begin{aligned}
\max_{I_{L}}\quad&\pi_{L}(I_{L})=(\frac{2+\Delta}{3}-\frac{1-\Delta}{27sI_{L}^{2}-3})^{2}\\
&+s(\frac{(1-\Delta)I_{L}}{9sI_{L}^{2}-1})^{2}-\gamma I_{L}^{2}\\
s.t\quad&\sqrt{\frac{2-\Delta}{9s}}\leq I_L.
\end{aligned}
\end{equation}
\end{theorem}
\begin{proof}
Substituting $I_{F}^{*}$ in \eqref{equ: Un-I_F1} into (\ref{equ: Un-L-payoff1}), the optimal investment level of $\text{SP}_{L}$, $I_{L}^{*}$, is a solution of the following optimization problem,
\begin{equation}\label{equ: Un-interior-pi_L-proof1}
\small
\begin{aligned}
\max_{I_{L}}\quad&\pi_{L}(I_{L})=(\frac{2+\Delta}{3}-\frac{I_{F}^{*}}{3I_{L}})^{2}+s(I_{F}^{*})^{2}-\gamma I_{L}^{2}\\
s.t\quad&\delta\leq I_{L}.
\end{aligned}
\end{equation}

\noindent{\bf Case 2.} If $M> \sqrt{\frac{2-\Delta}{9s}}$,
 then we have to consider the following sub-cases. 
 
\noindent{\textbf{(A)}.} From (\ref{equ: Un-I_F1}), if $\delta\leq I_{L}\leq\sqrt{\frac{2-\Delta}{9s}}$, then $I_{F}^{*}=I_{L}$, thus (\ref{equ: Un-interior-pi_L-proof1}) is equivalent to
\begin{align*}
\max_{I_{L}}\quad&\pi_{L}(I_{L})=\frac{(1+\Delta)^{2}}{9}+(s-\gamma)I_{L}^{2}\\
&\delta\leq I_{L}\leq\sqrt{\frac{2-\Delta}{9s}}
\end{align*}
Since $s>\gamma$, then $\pi_{L}(I_{L})$ is an increasing function of $I_{L}$, thus $I_{L}^{*}=\sqrt{\frac{2-\Delta}{9s}}$. This case can be considered as part of the next part.

\noindent{\textbf{(B)}.} If $\sqrt{\frac{2-\Delta}{9s}}<I_{L}\leq M$, then $I_{F}^{*}=\frac{(1-\Delta)I_{L}}{9I_{L}^{2}s-1}$. Note that $I_{F}^{*}=I_{L}$ when $I_{L}=\sqrt{\frac{2-\Delta}{9s}}$, then $I_{F}^{*}$ is continuous at $I_{L}=\sqrt{\frac{2-\Delta}{9s}}$.
Thus  $$\pi_{L}(I_{L})|_{I_F^*}\rightarrow\pi_{L}|_{I_{L}=I_{F}^{*}=\sqrt{\frac{2-\Delta}{9s}}}$$ as $$I_{L}\downarrow\sqrt{\frac{2-\Delta}{9s}}.$$ Therefore, this case also includes the optimum solution of previous case. Thus in this case (\ref{equ: Un-interior-pi_L-proof1}) is equivalent to
\begin{align*}
\max_{I_{L}}\quad&\pi_{L}(I_{L})=(\frac{2+\Delta}{3}-\frac{1-\Delta}{27sI_{L}^{2}-3})^{2}\\
&+s(\frac{(1-\Delta)I_{L}}{9sI_{L}^{2}-1})^{2}-\gamma I_{L}^{2}\\
s.t\quad&\sqrt{\frac{2-\Delta}{9s}}\leq I_L.
\end{align*}
Given $v^{L}$, $v^{F}$ and $s$, $I_{L}^{*}$ is a maximum of $\pi_{L}$, then no unilateral deviation is beneficial for SP$_{L}$.

\end{proof}

Collect all interior equilibria of $p_{F}^{*}, p_{L}^{*}$, and $I_{F}^{*}, I_{L}^{*}$, we have 
\begin{corollary}\label{cor-NE} 
If $-1<\Delta<1$, then the unique  SPNE strategy is: 

\noindent{\it Stage 1:} $I_{L}^{*}$  is characterized by
 \begin{equation*}
\begin{aligned}
\max_{I_{L}}\,\,&\pi_{L}(I_{L})=(\frac{2+\Delta}{3}-\frac{1-\Delta}{27sI_{L}^{2}-3})^{2}\\
&+s(\frac{(1-\Delta)I_{L}}{9sI_{L}^{2}-1})^{2}-\gamma I_{L}^{2}\\
s.t.\,\,&\sqrt{\frac{2-\Delta}{9s}}\leq I_L.
\end{aligned}
\end{equation*}

\noindent{\it Stage 2:}  $I_{F}^{*}$ is characterized in
 \begin{align*}
  I_{F}^{*}=\left\{\begin{aligned}
  &\frac{(1-\Delta)I_{L}}{9I_{L}^{2}s-1}\,\,\text{if}\,\, I_{L}>\sqrt{\frac{2-\Delta}{9s}}\\
  &I_{L}\,\,\text{if}\,\, I_{L}=\sqrt{\frac{2-\Delta}{9s}}\end{aligned}\right.
\end{align*}

\noindent{\it Stage 3:}
$p_{L}^{*}=c+\frac{2}{3}-\frac{I_{F}^{*}}{3I_{L}^{*}}+\frac{\Delta}{3},\quad
p_{F}^{*}=c+\frac{1}{3}+\frac{I_{F}^{*}}{3I_{L}^{*}}-\frac{\Delta}{3}$.

\noindent{\it Stage 4:}
 $n_{L}^{*}=\frac{\Delta}{3}+\frac{2}{3}-\frac{I_{F}^{*}}{3I_{L}^{*}},\, n_{F}^{*}=\frac{I_{F}^{*}}{3I_{L}^{*}}+\frac{1}{3}-\frac{\Delta}{3}$.
\end{corollary}

\noindent{\bf Section B: $1\leq \Delta<2$}

In this section, we consider $1\leq \Delta<2$. First, give the conditions under which the interior SPNE may exist (Proposition~\ref{pro: sectionB}). Then, We obtain an optimum  $I_{F}^{*}$ (in Theorem~\ref{thm: Un-A-I_F-sectionB}) and an optimum $I_{L}^{*}$ (in Theorem~\ref{thm: Un-interior-payoff-L-sectionB}), respectively. Finally, we find the interior SPNE $I_{F}^{*}$ and $I_{L}^{*}$. Note that $\delta$ is small, let 
\begin{equation*}
\footnotesize
\begin{aligned}
\delta<\min\big(\sqrt{\frac{2}{9s\Delta}-\frac{1}{9s}}, \frac{1}{\sqrt{2s(\Delta-1)}}\big).
\end{aligned}
\end{equation*}

\begin{proposition}\label{pro: sectionB}
If $1\leq \Delta<2$, then no corner SPNE strategies exist when 
\begin{align*}
(I_{F},I_{L})\in\left\{\delta\leq I_L, (\Delta-1)I_{L}< I_{F}\leq I_{L}\right\}.
\end{align*} 
\end{proposition}

\proof
From Corollary~\ref{remark: Un-region}, no corner equilibrium strategies exist if $(I_{L},I_{F})\in R$. Since 
\begin{align}
&0\leq \Delta-1<1\label{equ: SectionB-proposition-1}\\
&3\leq 2-\Delta<4\label{equ: SectionB-proposition-2},
\end{align}
then from \eqref{equ: R1}, \eqref{equ: SectionB-proposition-1} and \eqref{equ: SectionB-proposition-2},
$$R=\left\{\delta\leq I_L, (\Delta-1)I_{L}< I_{F}\leq I_{L}\right\}.$$
\qed

\noindent{\bf Stage 2:}
SP$_F$ decides on  the amount of spectrum to be leased from SP$_L$ ($I_F$), with the condition that $0\leq I_F\leq I_L$,  to maximize $\pi_F$.

\begin{theorem}\label{thm: Un-A-I_F-sectionB}
If $1\leq \Delta<2$, then the optimum investment level of $\text{SP}_{F}$, $I_{F}^{*}$, is  obtained by the following rules:
  \begin{itemize}
    \item[(1)] if $\Delta=1$, then $I_{F}^{*}\in[0,\frac{1}{\sqrt{9s}}]$ when $I_{L}=\frac{1}{\sqrt{9s}}$; $I_{F}^{*}=I_{L}$ when $0\leq I_{L}<\frac{1}{\sqrt
    {9s}}$; and no optimum $I_F^*$ when $I_L>\frac{1}{\sqrt
    {9s}}$;
    \item[(2)] if $1<\Delta<2$, then
$I_{F}^{*}=I_{L}$ when $\delta\leq I_{L}\leq\sqrt{\frac{2}{9s\Delta}-\frac{1}{9s}}$; no interior equilibria $I_{F}^{*}$ exist when 
$I_{L}>\sqrt{\frac{2}{9s\Delta}-\frac{1}{9s}}$.
\end{itemize}
\end{theorem}
\begin{proof}

From (\ref{equ: Fpayoff10}) and Proposition~\ref{pro: sectionB}, $I_{F}^{*}$ is obtained by the following optimization problems,
\begin{equation}\label{equ: 1S-best-inteior-I_F-2}
\footnotesize
\begin{aligned}
\max_{I_{F}}\quad&\pi_{F}(I_{F})=(\frac{1}{9I_{L}^{2}}-s)I_{F}^{2}+\frac{2(1-\Delta)}{9I_{L}}I_{F}+\frac{(1-\Delta)^{2}}{9}\\
s.t\quad&(\Delta-1)I_{L}<I_{F}\leq I_{L}
\end{aligned}
\end{equation}

\noindent{\textbf{(A)}.} First, we consider $I_{L}=\frac{1}{\sqrt{9s}}$, then
\begin{equation*}
\footnotesize
\begin{aligned}
\pi_{F}(I_{F})=\frac{2(1-\Delta)}{9I_{L}}I_{F}+\frac{(1-\Delta)^{2}}{9}		
\end{aligned}
\end{equation*}
 is a linear function of $I_{F}$. Since $1\leq \Delta<2$, then $\frac{2(1-\Delta)}{9I_{L}}\leq0$. 
 
\noindent{{\bf (A-1)}.} If $1<\Delta<2$, then $\pi_{F}(I_{F})$ is a strictly decreasing function of $I_{F}$, then 
$$I_{F}^{*}\downarrow(\Delta-1)I_{L},$$
which means
$$\pi_{F}(I_{F}^{*})\rightarrow\pi_{F}\big((\Delta-1)I_{L}\big),$$
which means $SP_{F}$ always wants to make a deviation to get a higher payoff by decreasing the investment level ($I_{F}\downarrow(\Delta-1)I_{L}$). There exists no optimum $I_{F}^{*}$ in this case.
  
\noindent{{\bf (A-2)}.} If $\Delta=1$, then $\pi_{F}(I_{F})=0$, so $I_{F}^{*}$ can be any number in the interval $(0, \frac{1}{\sqrt{9s}}]$ since $I_{L}=\frac{1}{\sqrt{9s}}$.

\noindent{\textbf{(B)}.} Then, we consider $I_{L}\neq\frac{1}{\sqrt{9s}}$. $\pi_{F}$ is a quadratic function. Note that
$$F_{1}=\frac{(1-\Delta)I_{L}}{9I_{L}^{2}s-1}.$$

\noindent{\textbf{(B-1)}.} If $\delta\leq I_{L}<\frac{1}{\sqrt{9s}}$, then $\pi_{F}$ is a convex function. Since $I_{L}\in\big((\Delta-1)I_{L},I_{L}]$, the midpoint of the interval is $\Delta I_{L}/2$. Note that $1\leq \Delta<2$ and $1-9sI_{L}^{2}>0$, then $$F_{1}=\frac{(1-\Delta)I_{L}}{9I_{L}^{2}s-1}\geq0,$$
From Lemma~\ref{lem: quadratic function},
\begin{align*}
\left\{\begin{aligned}
&I_{F}^{*}\rightarrow(\Delta-1)I_{L}\,& \frac{\Delta}{2}I_{L}<F_{1}\\
&I_{F}^{*}=I_{L}\,&\frac{\Delta}{2}I_{L}\geq F_{1}\end{aligned}\right..
\end{align*}
By simple calculation 
\begin{equation*}
\footnotesize
\begin{aligned}
&\frac{\Delta}{2}I_{L}<F_{1}\Leftrightarrow\frac{1}{\sqrt{9s}}>I_{L}>\sqrt{\frac{2}{9s\Delta}-\frac{1}{9s}}\\
&\frac{\Delta}{2}I_{L}\geq F_{1}\Leftrightarrow \delta\leq I_{L}\leq\sqrt{\frac{2}{9s\Delta}-\frac{1}{9s}}.
\end{aligned}	
\end{equation*}
thus
\begin{itemize}
  \item [(\romannumeral1)]$I_{F}^{*}\downarrow(\Delta-1)I_{L}$ when $\sqrt{\frac{2}{9s\Delta}-\frac{1}{9s}}<I_{L}<\frac{1}{\sqrt{9s}}$;
  \item [(\romannumeral2)]$I_{F}^{*}=I_{L}$ when $ \delta\leq I_{L}\leq\sqrt{\frac{2}{9s\Delta}-\frac{1}{9s}}.$
\end{itemize}
If $\sqrt{\frac{2}{9s\Delta}-\frac{1}{9s}}<I_{L}<\frac{1}{\sqrt{9s}}$,
then $I_{F}^{*}\downarrow(\Delta-1)I_{L}$,
which means
$$\pi_{F}(I_{F}^{*})\rightarrow\pi_{F}((\Delta-1)I_{L}),$$
which means $SP_{F}$ always wants to make a deviation to get a higher payoff by decreasing the investment level ($I_{F}\downarrow(\Delta-1)I_{L}$). There exists no optimum $I_{F}^{*}$ when $\sqrt{\frac{2}{9s\Delta}-\frac{1}{9s}}<I_{L}<\frac{1}{\sqrt{9s}}$. 
So the optimum investment level, $I_{F}^{*}$, is
$$I_{F}^{*}=I_{L}\,\,\text{when}\,\, \delta\leq I_{L}\leq\sqrt{\frac{2}{9s\Delta}-\frac{1}{9s}}.$$

\noindent{\textbf{(B-2)}.} If $I_{L}>\frac{1}{\sqrt{9s}}$, then $\pi_{F}$ is a concave function.  Since $1\leq \Delta<2$ and $1-9I_{L}^{2}s<0$, then $$F_{1}=\frac{(1-\Delta)I_{L}}{9I_{L}^{2}s-1}\leq0\leq(\Delta-1)I_{L}.$$ 
From Lemma~\ref{lem: quadratic function}, we have 
$$I_{F}^{*}\downarrow(\Delta-1)I_{L},$$
which means
$$\pi_{F}(I_{F})\rightarrow\pi_{F}((\Delta-1)I_{L}),$$
which means $SP_{F}$ always wants to make a deviation to get a higher payoff by decreasing the investment level ($I_{F}\downarrow(\Delta-1)I_{L}$). There exists no optimum $I_{F}^{*}$ in this case.

From \textbf{(A)} and \textbf{(B)}, we obtain the desired results. Given $v^{L}$, $v^{F}$, $s$ and $I_{L}$, if $I_{F}^{*}$ exists, then $I_F^*$ is the unique maximum of $\pi_{F}$, so no unilateral deviation is beneficial for SP$_{F}$.
\end{proof}
\noindent{\bf Stage 1:}
In this stage, MNO decides on the level of investment $I_{L}$ with the condition that $\delta\leq I_L\leq M$, to maximize his payoff $\pi_{L}$. 

\begin{theorem}\label{thm: Un-interior-payoff-L-sectionB}
If $1\leq \Delta<2$, the unique optimum investment level of $\text{SP}_{L}$, $I_{L}^{*}$, is
$I_{L}^{*}=I_{F}^{*}=\frac{1}{\sqrt{9s}}$ when $\Delta=1$, otherwise no interior SPNE $I_{L}^{*}$ exists.
\end{theorem}

\begin{proof}
Substituting $I_{F}^{*}$ in Theorem~\ref{thm: Un-A-I_F-sectionB} into (\ref{equ: Un-L-payoff1}), then the optimal investment level of $\text{SP}_{L}$, $I_{L}^{*}$, is a solution of the following optimization problem,
\begin{equation}\label{equ: Un-L-payoff1-sectionB}
 \begin{aligned}
\max_{I_{L}}\quad&\pi_{L}(I_{L})=(\frac{\Delta+2}{3}-\frac{I_{F}^{*}}{3I_{L}})^{2}+s(I_{F}^{*})^{2}-\gamma I_{L}^{2}\\
s.t.\quad& \delta\leq I_L.
\end{aligned}
\end{equation}

\noindent{\textbf{(A)}.} Consider $1<\Delta<2$. 
From Theorem~\ref{thm: Un-A-I_F-sectionB}~(2), $I_{F}^{*}=I_{L}$ when 
$\delta\leq I_{L}\leq\sqrt{\frac{2}{9s\Delta}-\frac{1}{9s}},$
thus the optimization (\ref{equ: Un-L-payoff1-sectionB}) is equivalent to
\begin{equation*}
\begin{aligned}
\max_{I_{L}}\quad&\pi_{L}(I_{L})=\frac{(1+\Delta)^{2}}{9}+(s-\gamma) I_{L}^{2}\\
s.t\quad&\delta\leq I_{L}\leq\sqrt{\frac{2}{9s\Delta}-\frac{1}{9s}}
\end{aligned}
\end{equation*}
Since $s>\gamma$, then $\pi_{L}(I_{L})>0$ for all $\delta\leq I_{L}\leq\sqrt{\frac{2}{9s\Delta}-\frac{1}{9s}}$, and   $\pi_{L}$ is an increasing function of $I_{L}$, thus $I_{L}^{*}=\sqrt{\frac{2}{9s\Delta}-\frac{1}{9s}}$.

\noindent{\bf{(B)}.} Consider $\Delta=1$, then we have the following sub-cases. 

\noindent{\bf Sub-case~1:} If $I_{L}=\frac{1}{\sqrt{9s}}$, from Theorem~\ref{thm: Un-A-I_F-sectionB}~(1),  (\ref{equ: Un-L-payoff1-sectionB}) is equivalent to
\begin{align*}
\pi_{L}(I_{L})=&\frac{1}{9}(1-\sqrt{9s}I_{F}^{*})^{2}+s(I_{F}^{*})^{2}- \frac{\gamma}{9s}\\
=&2sI_{F}^{*2}-2\sqrt{\frac{s}{9}}I_{F}^{*}+\frac{1}{9}(1-\frac{\gamma}{s}).
\end{align*}
Since $I_{F}^{*}\leq \frac{1}{\sqrt{9s}}$, then $\sqrt{9s}I_{F}^{*}\leq1$, thus $\pi_{L}$ is an increasing function of $I_{F}^*$, and $\pi_{L}(I_{L}; I_{F}^{*})\leq \frac{s-\gamma}{9s}$. Note that if $I_{F}^{*}<\frac{1}{\sqrt{9s}}$
\begin{align*}
\pi_L(I_L)=&2sI_{F}^{*2}-2\sqrt{\frac{s}{9}}I_{F}^{*}+\frac{1}{9}(1-\frac{\gamma}{s})\\
<&\lim_{I_{L}\rightarrow\frac{1}{\sqrt{9s}}}(s-\gamma) I_{L}^{2}=\frac{s-\gamma}{9s};
\end{align*}
and if $I_{F}^{*}=\frac{1}{\sqrt{9s}}$, 
$\pi_{L}(I_{L})=\frac{s-\gamma}{9s}$.

\noindent{\bf Sub-case~2:}
From Theorem~\ref{thm: Un-A-I_F-sectionB}~(1), if $0\leq I_{L}<\frac{1}{\sqrt{9s}}$, (\ref{equ: Un-L-payoff1-sectionB}) is equivalent to 
$\pi_{L}(I_{L}; I_{F}^{*})=(s-\gamma) I_{L}^{2} < \frac{s-\gamma}{9s}.$

From two sub-cases above, $I_L^*=I_F^*=\frac{1}{\sqrt{9s}}$ when $\Delta=1$. Note that
$$\sqrt{\frac{2}{9s\Delta}-\frac{1}{9s}}=\frac{1}{\sqrt{9s}},$$ 
when $\Delta=1$.
Thus this case can be considered as part of the above part. Therefore $I_{L}^{*}=\sqrt{\frac{2}{9s\Delta}-\frac{1}{9s}}$ for any $1\leq \Delta<2$. 

\noindent{\bf (C)}.
Now we compute $\pi_{F}$, from Theorem~\ref{thm: Un-A-I_F-sectionB}, 
\begin{align*}
\pi_{F}=&n_{F}(p_{F}-c)-sI_{F}^{*}=(\frac{2-\Delta}{3})^{2}-\frac{2}{9\Delta}+\frac{1}{9}\\
=&\frac{1}{9}(\Delta^{2}-4\Delta-\frac{2}{\Delta}+5)\triangleq f(\Delta)
\end{align*}
Taking the derivative with respect to $\Delta$, 
\begin{align*}
f'(\Delta)=&\frac{1}{9}(2\Delta-4+\frac{2}{\Delta^{2}})=\frac{2}{9\Delta^{2}}(\Delta^{3}-2\Delta^{2}+1)\\
=&\frac{2}{9\Delta^{2}}(\Delta-1)(\Delta^{2}-\Delta-1)
\end{align*}
Therefore, $f'(\Delta)>0$ when $\Delta\in[\frac{1+\sqrt{5}}{2},2)$, and $f'(\Delta)\leq0$ when $\Delta\in[1,\frac{1+\sqrt{5}}{2})$. Thus, $f_{\max}(\Delta)=f(1)=0$, which implies the possible interior equilibria exist when $\Delta=1$.
Then,  $I_{F}^{*}=I_{L}^{*}=\sqrt{\frac{1}{9s}}$, and
\begin{align*}
&p_{L}^{*}=c+\frac{2}{3},\quad p_{F}^{*}=c+\frac{1}{3}\\
&n_{L}^{*}=\frac{2}{3},\quad n_{F}^{*}=\frac{1}{3}.
\end{align*}

It is easy to check that if $\Delta=1$, then $(I_{L}^{*}, I_{F}^{*}, p_{L}^{*}, p_{F}^{*}, n_{L}^{*}, n_{F}^{*})$ satisfies Corollary~\ref{cor-NE}.
\end{proof}

\begin{corollary}\label{cor-NE2} 
If $\Delta=1$, then the unique  SPNE strategy is: $I_L^*=I_F^*=\sqrt{\frac{1}{9s}}$ and
$n_L^*=p_{L}^{*}-c=\frac{2}{3}$ and $n_F^*=p_{F}^{*}-c=\frac{1}{3}$.
\end{corollary}

\subsection*{ Section C: $-2< \Delta\leq-1$}
In this section, we consider $-2<\Delta\leq-1$. First, give the conditions under which the interior SPNE may exist (Proposition~\ref{pro: sectionC}). Then, we prove that no interior SPNE exists (Theorem~\ref{thm: Un-interior-payoff-L-sectionC}). Note that $\delta$ is small, let $\delta<\frac{1}{\sqrt{9s}}$.

\begin{proposition}\label{pro: sectionC}
If $-2<\Delta\leq-1$, then no corner SPNE exist when $\delta\leq I_{L}$ and $0\leq I_{F}< (\Delta+2)I_{L}$.
\end{proposition}
\begin{proof}
From Corollary~\ref{remark: Un-region}, no corner Nash equilibria exist if $(I_{L},I_{F})\in R$.
 If $-2<\Delta\leq-1$, then 
 \begin{align*}
 0<&\Delta+2\leq1\\
 -3<&\Delta-1\leq-2.
 \end{align*}
Thus from \eqref{equ: R1},
$$R=\left\{\delta\leq I_{L}, 0\leq I_{F}< (\Delta+2)I_{L}\right\}.$$
\end{proof}

\noindent{\bf Stage 2:}
SP$_F$ decides on  the amount of spectrum to be leased from SP$_L$ ($I_F$), with the condition that $0\leq I_F\leq I_L$,  to maximize $\pi_F$.

\begin{theorem}\label{thm: Un-A-I_F-sectionC}
If $-2<\Delta\leq-1$ and $\pi_{F}(I_{F}^{*}; I_{L})\geq0$, the optimum investment level of $\text{SP}_{F}$, $I_{F}^{*}$, is 
$I_{F}^{*}=\frac{(1-\Delta)I_{L}}{9I_{L}^{2}s-1}$ when $I_{L}>\frac{1}{\sqrt{3s(\Delta+2)}}$;
  and no interior SPNE $I_{F}^{*}$ exist when $\delta\leq I_{L}\leq\frac{1}{\sqrt{3s(\Delta+2)}}$.

\end{theorem}

\begin{proof}

From (\ref{equ: Fpayoff10}),  the optimal investment level of $\text{SP}_{F}$, $I_{F}^{*}$, is the solution of the following optimization problem, 
\begin{equation*}
\begin{aligned}
\max\quad&\pi_{F}(I_{F})=(\frac{1}{9I_{L}^{2}}-s)I_{F}^{2}+\frac{2(1-\Delta)}{9I_{L}}I_{F}+\frac{(1-\Delta)^{2}}{9}\\
s.t\quad&0\leq I_{F}< (\Delta+2)I_{L}.
\end{aligned}
\end{equation*}

\noindent{\textbf{(A)}.} If $I_{L}=\frac{1}{\sqrt{9s}}$, then
\begin{equation*}
\begin{aligned}
\pi_{F}(I_{F})=\frac{2(1-\Delta)}{9I_{L}}I_{F}+\frac{(1-\Delta)^{2}}{9}
\end{aligned}
\end{equation*}
is a linear function of $I_{F}$. Since $-2<\Delta\leq-1$, then $\frac{2(1-\Delta)}{9I_{L}}>0$, thus $\pi_{F}(I_{F}; I_{L})$ is a strictly increasing function of $I_{F}$. Therefore the optimum investment $I_{F}^{*}$,
$I_{F}^{*}\uparrow(\Delta+2)I_{L},$
which implies
$$\pi_{F}(I_{F})\rightarrow\pi_{F}((\Delta+2)I_{L}),$$
which means $SP_{F}$ always wants to make a deviation to get a higher payoff by increasing the investment level ($I_{F}\uparrow(\Delta+2)I_{L}$). No interior equilibria $I_{F}^{*}$ exists in this case.

\noindent{\textbf{(B)}.} If $I_{L}\neq\frac{1}{\sqrt{9s}}$, then $\pi_{F}$ is a quadratic function, and $F_{1}=\frac{(1-\Delta)I_{L}}{9I_{L}^{2}s-1}$.

{\bf (B-1).} If $\delta\leq I_{L}<\frac{1}{\sqrt{9s}}$,  then $\pi_{F}$ is a convex function. Since $I_{F}\in[0,(\Delta+2)I_{L})$, the midpoint of the interval is 
$(\Delta+2)I_{L}/2.$ Since $-2<\Delta\leq-1$ and $1-9I_{L}^{2}s>0$, then 
$$F_{1}=\frac{(1-\Delta)I_{L}}{9I_{L}^{2}s-1}<0<(\Delta+2)I_{L}/2,$$ from Lemma~\ref{lem: quadratic function}~(1), $I_{F}^{*}\uparrow(\Delta+2)I_{L},$ which means
$$\pi_{F}(I_{F}; I_{L})\rightarrow\pi_{F}((\Delta+2)I_{L};I_{L}),$$ 
which means $SP_{F}$ always wants to make a deviation to get a higher payoff by increasing the investment level. No interior equilibria $I_{F}^{*}$ exists in this case.

{\bf (B-2) .} If $I_{L}>\frac{1}{\sqrt{9s}}$, then $\pi_{F}$ is a  concave function.  Since $-2< \Delta\leq-1$ and $1-9I_{L}^{2}s<0$, then 
$$F_{1}=\frac{(1-\Delta)I_{L}}{9I_{L}^{2}s-1}>0.$$ From Lemma~\ref{lem: quadratic function}~(2),
\begin{align*}
I_F^*=\left\{\begin{aligned}
&I_{F}^{*}=F_{1}&\text{when}\quad&0<F_{1}<(\Delta+2)I_{L}\\
&I_{F}^{*}\rightarrow(\Delta+2)I_{L}&\text{when}\quad&F_{1}\geq(\Delta+2)I_{L}
\end{aligned}\right..
\end{align*} 
By simple calculation,
\begin{equation*}
\begin{aligned}
&0<F_{1}<(\Delta+2)I_{L}\Leftrightarrow I_{L}>\frac{1}{\sqrt{3s(\Delta+2)}}\\
&F_{1}\geq(\Delta+2)I_{L}\Leftrightarrow\frac{1}{\sqrt{9s}}<I_{L}\leq\frac{1}{\sqrt{3s(\Delta+2)}},
\end{aligned}
\end{equation*}
then
\begin{align*}
I_F^*\left\{\begin{aligned}
&=\frac{(1-\Delta)I_{L}}{9I_{L}^{2}s-1}&\text{when}\quad&I_{L}>\frac{1}{\sqrt{3s(\Delta+2)}}\\
&\rightarrow(\Delta+2)I_{L}&\text{when}\quad&\frac{1}{\sqrt{9s}}<I_{L}\leq\frac{1}{\sqrt{3s(\Delta+2)}}
\end{aligned}\right..
\end{align*} 
Note that $I_{F}^{*}\uparrow(\Delta+2)I_{L},$
which means 
$$\pi_{F}(I_{F}; I_{L})\rightarrow\pi_{F}((\Delta+2)I_{L};I_{L}),$$ 
which means $SP_{F}$ always wants to make a deviation to get a higher payoff by increasing the investment level ($I_{F}\uparrow(\Delta+2)I_{L}$). Thus
\begin{equation*}
\begin{aligned}
I_{F}^{*}=\frac{(1-\Delta)I_{L}}{9I_{L}^{2}s-1}\,\,\text{when}\,\, I_{L}>\frac{1}{\sqrt{3s(\Delta+2)}}
\end{aligned}
\end{equation*}
Since $I_{F}^{*}$ is the unique maximum of $\pi_{F}$, so no unilateral deviation is beneficial for SP$_{F}$.
From \textbf{(A)} and \textbf{(B)}, we obtain the desired results.

\end{proof}

\noindent{\bf Stage 1:}
In this stage, MNO decides on the level of investment $I_{L}$ with the condition that $\delta\leq I_L\leq M$, to maximize his payoff $\pi_{L}$.

\begin{theorem}\label{thm: Un-interior-payoff-L-sectionC}
If $-2<\Delta\leq-1$, then no interior SPNE $I_{L}^{*}$ exists.
\end{theorem}

\begin{proof}
Substituting $I_{F}^{*}$ in Theorem~\ref{thm: Un-A-I_F-sectionC} into (\ref{equ: Un-L-payoff1}), the optimum investment level of SP$_{L}$, $I_{L}^{*}$, is a solution of the following optimization problem,
\begin{equation*}
\small
\begin{aligned}
\max_{I_{L}}\,\,&\pi_{L}=(\frac{2+\Delta}{3}-\frac{1-\Delta}{27I_{L}^{2}s-3})^{2}+s(\frac{3(1-\Delta)I_{L}}{9I_{L}^{2}s-1})^{2}-\gamma I_{L}^{2}\\
s.t\,\,&\frac{1}{\sqrt{3s(\Delta+2)}}<I_L
\end{aligned}
\end{equation*}
Denote
\begin{equation*}
\footnotesize
\begin{aligned}
f(I_{L})=&(\frac{2+\Delta}{3}-\frac{1-\Delta}{27I_{L}^{2}s-3})^{2}+s(\frac{3(1-\Delta)I_{L}}{9I_{L}^{2}s-1})^{2},
\end{aligned}	
\end{equation*}
we prove that $f(I_{L})$ is a strictly decreasing function of $I_{L}$. Denote 
 $$f_{1}(I_{L})=(\frac{1-\Delta}{27I_{L}^{2}s-3}-\frac{2+\Delta}{3})^{2}$$
and 
$$f_{2}(I_{L})=s\Big(\frac{3(1-\Delta)I_{L}}{9I_{L}^{2}s-1}\Big)^{2},$$
then $f(I_{L})=f_{1}(I_{L})+f_{2}(I_{L})$. In fact,
\begin{equation*}
\small
\begin{aligned}
f_{1}'(I_{L})=&2(\frac{1-\Delta}{27I_{L}^{2}s-3}-\frac{2+\Delta}{3})\cdot\frac{(\Delta-1)}{(27I_{L}^{2}s-3)^{2}}54I_{L}s\\
=&\frac{4(\Delta-1)I_{L}s}{(9I_{L}^{2}s-1)^{2}}\cdot(\frac{1-\Delta}{9I_{L}^{2}s-1}-(2+\Delta)),
\end{aligned}
\end{equation*}
and
\begin{equation*}
\small
\begin{aligned}
f_{2}'(I_{L})=&\frac{6s(1-\Delta)I_{L}}{(9I_{L}^{2}s-1)^{3}}[3(1-\Delta)(9I_{L}^{2}s-1)-54(1-\Delta)I_{L}^2s]\\
=&\frac{-18I_{L}s(1-\Delta)^{2}(9I_{L}^{2}s+1)}{(9I_{L}^{2}s-1)^{3}}
\end{aligned}
\end{equation*}
Therefore
\begin{equation*}
\small
\begin{aligned}
f'(I_{L})=&f_{1}'(I_{L})+f_{2}'(I_{L})=\frac{2(1-\Delta)I_{L}s}{(9I_{L}s-1)^{2}}[\frac{-2(1-\Delta)}{9I_{L}^{2}s-1}+2(2+\Delta)\\
-&\frac{9(1-\Delta)(9I_{L}^{2}s+1)}{9I_{L}^{2}s-1}]\\
=&\frac{2(1-\Delta)I_{L}s}{(9I_{L}s-1)^{2}}[\frac{-20(1-\Delta)}{9I_{L}^{2}s-1}+2(2+\Delta)-9(1-\Delta)]\\
=&\frac{2(1-\Delta)I_{L}s}{(9I_{L}s-1)^{2}}[\frac{-20(1-\Delta)}{9I_{L}^{2}s-1}+11\Delta-5]
\end{aligned}
\end{equation*}
Note that $-2<\Delta\leq-1$, then $2\leq1-\Delta<3$ and $0<\Delta+2\leq1$. Since 
$$I_{L}>\frac{1}{\sqrt{3s(\Delta+2)}},$$ then 
$$9I_{L}^{2}s-1>\frac{1-\Delta}{\Delta+2}>0.$$ Thus 
\begin{align*}
\frac{-20(1-\Delta)}{9I_{L}^{2}s-1}<0,\quad 11\Delta-5<0,
\end{align*}
so $f'(I_{L})<0$, and $f(I_{L})$ is a strictly decreasing function. Therefore $\pi_{L}(I_{L})=f(I_{L})-\gamma I_{L}^{2}$ is a strictly decreasing function of $I_{L}$ when $I_{L}>\frac{1}{\sqrt{3s(\Delta+2)}}$.
Hence 
$$I_{L}^{*}\downarrow\frac{1}{\sqrt{3s(\Delta+2)}},$$
which implies
$$\pi_{L}(I_{L})\uparrow\pi_{L}(\frac{1}{\sqrt{3s(\Delta+2)}}),$$ 
which implies $\text{SP}_{L}$ always wants to make a deviation to get a higher payoff by decreasing the investment level ($I_{L}\downarrow\frac{1}{\sqrt{3s(\Delta+2)}}$). No interior equilibria $I_{L}^{*}$ in this case.
\end{proof}

\subsection*{ Section D: $|\Delta|\geq2$}

\begin{theorem}\label{Pro: SectionD}
If $|\Delta|\geq2$, then no interior Nash equilibrium strategies exist.
\end{theorem}
\proof
We calculate $R$ in Corollary~\ref{remark: Un-region}:  
If $\Delta\geq2$, then 
$$\frac{I_{F}}{I_{L}}>\Delta-1\geq1\Rightarrow I_{F}>I_{L},$$ which is contradicted by $0\leq I_{F}\leq I_{L}$, thus $R=\varnothing$. 
Similarly, if $\Delta\leq-2$, then 
$$\frac{I_{F}}{I_{L}}<v^{F}-v^{L}+2\leq0\Rightarrow I_{F}<0,$$ which is contradicted by $0\leq I_{F}\leq I_{L}$, thus $R=\varnothing$. Therefore,  (\ref{equ: Un-A-interior-condition}) does not hold for any $\delta\leq I_{L}\leq M$ and $0\leq I_{F}\leq I_{L}$ when $|\Delta|\geq2$.

Thus no interior SPNE access fees exist, hence no interior Nash equiliberium strategies exist.
\qed

\subsection{Corner SPNE}\label{sec: corner SPNE}
Note that $\delta$ is small, let $\delta<\frac{1}{\sqrt{2s}}$ in this section.

\begin{lemma}\label{lem: no negative-corner SPNE}
Consider $x_{0}\leq0$, no corner SPNE strategies  exist when $\Delta>-1$. 
\end{lemma}

\begin{proof}

Let $x_{0}^*\leq0$. Clearly,  $n_{F}^*=1$ and $n_{L}^*=0$.
 From~\eqref{equ: Un-A-indifferent location},  
 \begin{align}\label{equ: positive-corner-p_L-p_F-1}
 p_{F}^*- p_{L}^*+\Delta+t_{F}^{*}\leq0. 	
 \end{align}

 \noindent{\bf Step 1.} We prove that $p_{F}^{*}-p_{L}^{*}+\Delta+t_{F}^{*}=0.$ 
 
Assume not, suppose $p_{F}^*- p_{L}^*+\Delta+t_{F}^{*}<0$.  Consider a unilateral deviation by which $p_F' = p_F^*+\epsilon$, such that $p_{F}'- p_{L}^*+\Delta+t_{F}^{*}<0$. From \eqref{equ: Un-A-indifferent location}, $x_0' =1$.  Now, from \eqref{equ: BM-F-payoff}, $\pi_F'-\pi_F^*=\epsilon > 0 $. Thus, $(I_F^*, p_F^*)$ is not  SP$_F$'s best response  to SP$_L$'s choices $(I_L^*, p_L^*)$,    which is a contradiction. Hence, $p_{F}^{*}-p_{L}^{*}+\Delta+t_{F}^{*}=0.$

\noindent{\bf Step 2.} We prove that $p_{F}^{*}\geq c.$

From \eqref{equ: BM-F-payoff}, $ \pi_{F}^{*}=p_{F}^{*}-c-s I_{F}^{*2}.$
If $p_{F}^{*}<c$, then $\pi_{F}^*<-s I_{F}^{*2} < 0$. Consider a unilateral deviation by which  $I_{F}=0, p_{F}=c$, then
$\pi_{F}=0$, which is beneficial for SP$_F$.   Thus, $p_{F}^{*}\geq c$.

\noindent{\bf Step 3.} If $\Delta>-1$, then $p_{F}^{*}<c+1$.

If $\Delta>-1$, then let $p_{F}^{*}\geq c+1$.  
Consider a unilateral deviation by which $p_{L}=p_L^*-\epsilon$, then $x_{0}=p_{F}^{*}-p_{L}+\Delta+t_{F}^{*}=\epsilon.$
In addition,
$p_{L}=p_L^*-\epsilon=p_{F}^{*}+\Delta+t_{F}^{*}\geq c+1+\Delta$, 
thus
\begin{align*}
\pi_{L}-\pi_{L}^{*}\geq \epsilon(1+\Delta-\epsilon).
\end{align*}
We can choose some $0<\epsilon<1$ such that $\pi_{L}-\pi_{L}^{*}>0$. Hence, $p_{F}< c+1$.

Now consider another unilateral deviation of SP$_{F}$, $p_{F}'=p_{F}^{*}+\epsilon$, where $0 < \epsilon < 1$,  with all the rest the same, then
\begin{align*}
&n_{L}'=x_{0}'=t_{F}^{*}+\Delta+p_{F}^*-p_{L}'=\epsilon\\
&n_{F}'=1-n_{L}'=1-\epsilon.
\end{align*}
Thus,
\begin{align*}
&\pi_{F}'-\pi_{F}^{*}=n_{F}'(p_{F}'-c) - (p_{F}^{*}-c)\\
=&
-\epsilon(p_F^* - c)+(1-\epsilon)\epsilon\\
=&\epsilon(-p_F^*+c+1-\epsilon)>0.	
\end{align*}
The last inequality follows because we can choose $0<\epsilon<1$ such that $p_F'=p_F^*+\epsilon < c+1.$ Thus, we  arrive at a contradiction.
\end{proof}

\begin{lemma}\label{lem: no positive-corner SPNE}
Consider $x_{0}\geq1$, no corner SPNE strategies exist when $\Delta<1$. 
\end{lemma}

\begin{proof}
Let $x_{0}^*\geq1$. Clearly,  $n_{F}^*=0$ and $n_{L}^*=1$.
 From \eqref{equ: Un-A-indifferent location}, $1 \leq x_{0}^*=\Delta+t_{F}^*+p_{F}^*-p_{L}^*$. Thus, 
 \begin{align}\label{equ: positive-corner-p_L-p_F-1}
 p_{F}^*- p_{L}^*+\Delta+t_{F}^{*}-1\geq0. 	
 \end{align}

 \noindent{\bf Step 1.} We prove that $p_{F}^{*}-p_{L}^{*}+\Delta=0.$ 

Assume not, suppose $p_{F}^*- p_{L}^*+\Delta>0$.  Consider a unilateral deviation by which $p_L' = p_L^*+\epsilon$, such that $p_{F}^*- p_{L}'+\Delta>0$. From \eqref{equ: Un-A-indifferent location}, $x_0' =1$.  Now, from \eqref{equ: BM-L-payoff}, $\pi_L'-\pi_L^*=\epsilon > 0 $. Thus, $(I_L^*, p_L^*)$ is not  SP$_L$'s best response  to SP$_F$'s choices $(I_F^*, p_F^*)$,    which is a contradiction. Hence, $p_{F}^{*}-p_{L}^{*}+\Delta=0.$

\noindent{\bf Step 2.}  We prove that $p_{L}^{*} \geq c.$

From \eqref{equ: BM-L-payoff}, $ \pi_{L}^{*}=p_{L}^{*}-c+s I_{F}^{*2}-\gamma I_{L}^{*2}.$
If $p_{L}^{*}<c$, then $\pi_{L}^*<s I_{F}^{*2} -\gamma I_{L}^{*2}\leq s I_{F}^{*2} -\gamma \delta^{2} $. Consider a unilateral deviation by which  $I_{L}=\delta, p_{L}=c$, then
$\pi_{L} =sI_{F}^{*2}-\gamma\delta^{2}$, which is beneficial for SP$_L$.   Thus, $p_{L}^{*}\geq c$.

\noindent{\bf Step 3.} We prove that $I_{F}^{*}=0$ and $\pi_{F}^{*}=0.$

For any SPNE $(I_{F}^{*}, p_{F}^{*})$, we have $\pi_{F}^{*}\geq0.$ Otherwise, assume $\pi_{F}<0$, we consider a unilateral deviation $I_{F}=0$ and $p_{F}=c$, then
$\pi_{F}=0$, which is beneficial for SP$_{F}$. If $n_{F}^{*}=0$, then $\pi_{F}^{*}=-sI_{F}^{*2}\geq0\Rightarrow I_{F}^{*}=0, \pi_{F}^{*}=0.$

Based on {\bf Step~3}, since $I_{F}^{*}=0,$ then $$t_{F}^{*}=\frac{I_{L}^{*}-I_{F}^{*}}{I_{L}^{*}}=1.$$

\noindent{\bf Step 4.} If $\Delta<1$,then $p_{L}<c+1$.

If $\Delta<1$, let $p_L^* \geq c+1.$  Thus, 
\begin{align}\label{equ: corner-proof1}
p_{F}^{*} = p_{L}^{*}-\Delta \geq c-\Delta. 
\end{align}
Recall that $x_0^* = 1+\Delta+p_F^*-p_L^*$ , then
 consider a unilateral deviation by which $p_F = p_L^*-\Delta-\epsilon > c+1-\Delta$. Now, by \eqref{equ: Un-A-indifferent location}, $x_0 < 1$, and hence $n_F > 0.$  Now, from \eqref{equ: BM-F-payoff}, $\pi_F > 0 = \pi_F^*$. Thus, $(I_F^*, p_F^*)$ is not  SP$_F$'s best response  to SP$_L$'s choices $(I_L^*, p_L^*)$,    which is a contradiction. Hence, $p_L^* < c+1.$

Now consider another unilateral deviation of SP$_{L}$, $p_{L}'=p_{L}^{*}+\epsilon$, where $0 < \epsilon < 1$,  with all the rest the same, then
\begin{align*}
&n_{L}'=x_{0}'=t_{F}^{*}+\Delta+p_{F}^*-p_{L}'=1-\epsilon.
\end{align*}
Then
\begin{align*}
&\pi_{L}'-\pi_{L}^{*}=n_{L}'(p_{L}'-c) - (p_{L}^{*}-c)\\
=&
-\epsilon(p_L^* - c)+(1-\epsilon)\epsilon\\
=&\epsilon(-p_L^*+c+1-\epsilon).	
\end{align*}
The last inequality follows because we can choose $0<\epsilon<1$ such that $p_L'=p_{L}^{*}+\epsilon < c+1.$ Thus, we  arrive at a contradiction.
\end{proof}

\begin{theorem}\label{thm: negative-corner} 
If $\Delta\leq-1$, then the unique corner SPNE strategy is:  $I_L^*=I_F^*=\frac{1}{\sqrt{2s}}$,  $p_{L}^{*}=p_{F}^{*}+\Delta-1,$ $c+1\leq p_{F}^{*}\leq c-\Delta-1$ and
 $n_{L}^{*}=0,\, n_{F}^{*}=1.$

\end{theorem}
\begin{proof}
\noindent{\bf Step 1.} We prove that $p_{F}^{*}\leq c-\Delta-t_{F}^{*}.$

Suppose $p_{F}^{*}>c-\Delta-t_{F}^{*}$, then from {\bf Step~1} in Lemma~\ref{lem: no negative-corner SPNE}, $p_{L}^{*}=p_{F}^{*}+\Delta+t_{F}^*-t_{F}^{*}>c.$ Now consider a unilateral deviation of SP$_{L}$, $p_{L}=p_{L}^{*}-\epsilon$, where $0 < \epsilon < 1$,  with the rest keeping original, then
\begin{align*}
n_{L}=x_{0}=t_{F}^{*}+\Delta+p_{F}^*-p_{L}=\epsilon.
\end{align*}
Thus,
\begin{align*}
\pi_{L}-\pi_{L}^{*}=n_{L}(p_{L}-c)=\epsilon(p_{L}-c)>0.	
\end{align*}
The last inequality holds because we can choose $0<\epsilon<1$ such that $p_L=p_{L}^{*} -\epsilon> c.$ So $p_{F}^{*}>c-\Delta$ can not be a SPNE.

\noindent{\bf Step 2.} We prove that $p_{F}^{*}\geq c+1.$

Suppose $p_{F}^{*}<c+1$, consider a unilateral deviation of SP$_{F}$, $p_{F}=p_{F}^{*}+\epsilon$, where $0 < \epsilon < 1$,  with the rest keeping original, then
\begin{align*}
&n_{L}=x_{0}=t_{F}^{*}+\Delta+p_{F}^*-p_{L}=\epsilon\\
&n_F=1-n_L=1-\epsilon.
\end{align*}
Thus,
\begin{align*}
&\pi_{F}-\pi_{F}^{*}=n_{F}(p_{F}-c)-p_F^*+c\\
=&\epsilon(1-\epsilon-p_F^*+c)>0.	
\end{align*}
The last inequality holds because we can choose $0<\epsilon<1$ such that $p_F^*+\epsilon<1+ c$. So $p_{F}^{*}<c+1$ can not be a SPNE.

Therefore from {\bf Steps~1, 2}, note that $t_F^*=1-I_F^*/I_L^*$, so $c+1>c-\Delta-t_{F}^{*}$ when $\frac{I_{F}^{*}}{I_{L}^{*}}< 2+\Delta$, thus no corner SPNE exists in this range. Then we consider $\frac{I_{F}^{*}}{I_{L}^{*}}\geq 2+\Delta.$

\noindent{\bf Step 3.} We prove that no unilateral deviation is beneficial for both SPs when $c+1\leq p_{F}^{*}\leq c-\Delta-t_{F}^{*}.$

Consider a unilateral deviation of SP$_{L}$, $p_{L}'=p_{L}^{*}-\epsilon$, where $0 < \epsilon < 1$,  with the rest keeping original, then
\begin{align*}
n_{L}'=x_{0}'=t_{F}^{*}+\Delta+p_{F}^*-p_{L}'=\epsilon.
\end{align*}
Since $p_{L}^{*}=p_{F}^{*}+\Delta+t_{F}^*$, then $p_{L}^{*}\in[c+1+\Delta+t_{F}^{*},c]$, then
\begin{align*}
\pi_{L}'-\pi_{L}^{*}=n_{L}'(p_{L}'-c)<0,
\end{align*}
which implies no unilateral deviation is beneficial for SP$_{L}$.

Consider another unilateral deviation of SP$_{F}$, $p_{F}'=p_{F}^{*}+\epsilon$, where $0 < \epsilon < 1$,  with the rest keeping original, then
\begin{align*}
&n_{L}'=x_{0}'=t_{F}^{*}+\Delta+p_{F}'-p_{L}^*=\epsilon\\
&n_{F}'=1-n_{L}'=1-\epsilon.
\end{align*}
Note that $c+1\leq p_{F}^{*}\leq c-\Delta-t_{F}^{*}$,
\begin{align*}
&\pi_{F}'-\pi_{F}^{*}=n_{F}'(p_{F}'-c)-p_{F}^{*}+c\\
=&\epsilon(-p_{F}^{*}+c+1-\epsilon)\leq -\epsilon^{2}<0.
\end{align*}
which implies no unilateral deviation is beneficial for SP$_{L}$.

\noindent{\bf Step 5.} Find $I_{F}^{*}$.

Note that $p_{L}^*$ is independent of $I_{F}^*$. Substituting $p_{F}^{*}=p_{L}^*-\Delta-t_{F}^{*}$ into \eqref{equ: BM-F-payoff},
$I_{F}^{*}$ is the solution of the following optimization problem, 
\begin{equation*}
\begin{aligned}
\max\quad&\pi_{F}(I_{F})=-sI_{F}^{2}+\frac{I_{F}}{I_{L}}-\Delta+p_{L}^{*}-c-1\\
s.t\quad&0\leq I_{F}\leq I_{L}
\end{aligned}
\end{equation*}
$\pi_{F}(I_{F})$ is a concave function, and the symmetric axis is $F_{2}=\frac{1}{2sI_{L}}>0$. From Lemma~\ref{lem: quadratic function}~(2),
\begin{align*}
I_F^*=\left\{\begin{aligned}
&F_{2}&\quad\text{when}\quad&F_{2}<I_{L}\\
&I_L&\quad\text{when}\quad&F_{2}\geq I_{L}
\end{aligned}\right.
\end{align*}
which is equivalent to
\begin{align*}
I_F^*=\left\{\begin{aligned}
&\frac{1}{2sI_{L}}&\quad\text{when}\quad&\frac{1}{\sqrt{2s}}<I_{L}\\
&I_L&\quad\text{when}\quad&I_L\leq\frac{1}{\sqrt{2s}}
\end{aligned}\right.
\end{align*}
Since $I_{F}^{*}$ is the unique maximum of $\pi_F$, thus no unilateral deviation is beneficial to SP$_{F}$.

\noindent{\bf Step 6.} Find $I_{L}^*$.

Substituting $I_{F}^{*}$ from {\bf Step~5} into (\ref{equ: BM-L-payoff}), the optimum investment level of SP$_{L}$, $I_{L}^{*}$, is a solution of the following optimization problem,
\begin{equation}\label{equ: Un-L-payoff2-sectionC}
\begin{aligned}
\max_{I_{L}}\quad&\pi_{L}(I_{L}; I_{F}^{*})=sI_{F}^{*2}-\gamma I_{L}^{2}\\
s.t\quad&\delta\leq I_{L}.
\end{aligned}
\end{equation} 

{\textbf{(A)}.} If $\delta\leq I_L\leq\frac{1}{\sqrt{2s}}$, then $I_{F}^{*}=I_{L}$, thus the optimization (\ref{equ: Un-L-payoff2-sectionC}) is equivalent to
\begin{align*}
\max_{I_{L}}\quad&\pi_{L,1}=(s-\gamma)I_{L}^{2}\\
s.t\quad&\delta\leq I_{L}\leq\frac{1}{\sqrt{2s}}.
\end{align*}
Note that $s>\gamma$, then $\pi_{L,1}$ is an increasing function of $I_{L}$, thus $I_{L}^{*}=I_{F}^{*}=\frac{1}{\sqrt{2s}}$. Denote 
$$\pi_{L,1}^{*}=\pi_{L,1}(\frac{1}{\sqrt{2s}})=\frac{s-\gamma}{2s}.$$ 

{\bf (B).} If  $\frac{1}{\sqrt{2s}}<I_{L}\leq M$, then $I_{F}^{*}=\frac{1}{2sI_{L}}$, thus the optimization (\ref{equ: Un-L-payoff2-sectionC}) is equivalent to
\begin{align*}
\max_{I_{L}}\quad&\pi_{L,2}=\frac{1}{4sI_{L}^{2}}-\gamma I_{L}^{2}\\
s.t\quad&\frac{1}{\sqrt{2s}}<I_{L}.
\end{align*}
$\pi_{L,2}$ is a decreasing function of $I_{L}$, note that $\gamma<s$, denote 
$$\pi_{L,2}^{*}=\pi_{L,2}(\frac{1}{\sqrt{2s}})=\frac{1}{2}(1-\frac{\gamma}{s})>0,$$ so $\pi_{L,2}\uparrow\pi_{L,2}^{*}$ as $I_{L}\downarrow\frac{1}{\sqrt{2s}}$,  which means $\text{SP}_{L}$ always wants to make a deviation to get a higher payoff by decreasing the investment level ($I_{L}\downarrow\frac{1}{\sqrt{2s}}$). No negative-corner equilibria $I_{L}^{*}$ in this case. From {\bf (A)} and {\bf (B)}, $\pi_{L,1}^*=\pi_{L,2}^*>\pi_{L,2}$, thus $I_L^*=I_F^*=\frac{1}{\sqrt{2s}}$.

From {\bf Sub-cases~1, 2}, we can obtain the desired results.

\end{proof}

\begin{theorem}\label{thm: positive-corner} 
If $\Delta\geq1$, then the unique negative-corner SPNE strategy is: $I_{L}^{*}=\delta$,  $I_{F}^{*}=0$, $p_{F}^{*}=p_{L}^{*}-\Delta$,  $c+1\leq p_{L}^{*}\leq c+\Delta$, $n_{L}^{*}=1$, $n_{F}^{*}=0$.
\end{theorem}
\begin{proof}
\noindent{\bf Step 1.} We prove that $c+1\leq p_{L}^{*}\leq c+\Delta.$ 

From {\bf Steps~1, 3} in Lemma~\ref{lem: no positive-corner SPNE}, $I_{F}^{*}=0$, $t_{F}^{*}=1$ and $p_{F}^{*}=p_{L}^{*}-\Delta.$

Suppose $p_{L}^{*}>c+\Delta$, $p_{F}^{*}=p_{L}^{*}-\Delta>c.$ Now consider a unilateral deviation of SP$_{F}$, $p_{F}=p_{F}^{*}-\epsilon$, where $0 < \epsilon < 1$,  with the rest keeping original, then
\begin{align*}
&n_{L}=x_{0}=t_{F}^{*}+\Delta+p_{F}-p_{L}^*=1-\epsilon\\
&n_{F}=1-n_{L}=\epsilon.
\end{align*}
Thus,
\begin{align*}
\pi_{F}-\pi_{F}^{*}=\epsilon(p_{F}-c)>0,	
\end{align*}
the last inequality holds because we can choose $0<\epsilon<1$ such that $p_F -\epsilon> c.$ Thus, $p_{L}^{*}>c+\Delta$ can not be a SPNE.

Suppose $p_{L}^{*}<c+1$, consider a unilateral deviation of SP$_{L}$, $p_{L}=p_{L}^{*}+\epsilon$, where $0 < \epsilon < 1$,  with the rest keeping original, then
\begin{align*}
n_{L}=x_{0}=t_{F}^{*}+\Delta+p_{F}-p_{L}^*=1-\epsilon.
\end{align*}
Thus,
\begin{align*}
\pi_{L}-\pi_{L}^{*}=\epsilon(-p_{L}^*+c+1-\epsilon)>0,	
\end{align*}
the last inequality follows because we can choose $0<\epsilon<1$ such that $p_{L}=p_L^* +\epsilon< c+1.$ Thus, $p_{L}^{*}<c+1$ can not be a SPNE.

In addition, we prove that no unilateral deviation is beneficial for both SPs when $c+1\leq p_{L}^{*}\leq c+\Delta.$ Consider another unilateral deviation of SP$_{F}$, $p_{F}'=p_{F}^{*}-\epsilon$, where $0 < \epsilon < 1$,  with the rest keeping original, then
\begin{align*}
&n_{L}'=x_{0}'=t_{F}^{*}+\Delta+p_{F}'-p_{L}^*=1-\epsilon\\
&n_{F}'=1-n_{L}'=\epsilon.
\end{align*}
Since $p_{F}^{*}=p_{L}^{*}-\Delta$, then $p_{F}^{*}\in[c-\Delta+1,c]$, then
\begin{align*}
\pi_{F}'-\pi_{F}^{*}=n_{F}'(p_{F}'-c)<0,
\end{align*}
which implies no unilateral deviation is beneficial for SP$_{F}$.

Consider another unilateral deviation of SP$_{L}$, $p_{L}'=p_{L}^{*}+\epsilon$, where $0 < \epsilon < 1$,  with the rest keeping original, then
\begin{align*}
n_{L}'=x_{0}'=t_{F}^{*}+\Delta+p_{F}^*-p_{L}'=1-\epsilon.
\end{align*}
Thus, note that $c+1\leq p_{L}^{*}\leq c+\Delta,$ 
\begin{align*}
&\pi_{L}'-\pi_{L}^{*}=n_{L}'(p_{L}'-c)-p_{L}^{*}+c\\
=&\epsilon(-p_{L}^{*}+c+1-\epsilon)\leq -\epsilon^{2}<0.
\end{align*}
which implies no unilateral deviation is beneficial for SP$_{L}$.

\noindent{\bf Step 2.} Find $I_{F}^{*}=0.$
From Lemma~\ref{lem: no positive-corner SPNE}, $\pi_F^*\geq0$, so $I_{F}^{*}=0$.

\noindent{\bf Step 3.} Find $I_{L}^{*}=\delta.$

Since $p_L^*$ is independent of $I_{L}^{*}$, then from \eqref{equ: BM-L-payoff}, $\pi_{L}=p_{L}^*-c-\gamma I_{L}^{*}$ is a decreasing function of $I_{L}$. Note that $I_{L}\geq\delta$, therefore $I_{L}^{*}=\delta.$ Since $I_{L}^{*}=\delta$ is the unique maximum of $\pi_L$, so no unilateral deviation is beneficial for SP$_{L}$.

\end{proof}

\section{EUs with Outside Option: SPNE Analysis}\label{Appendix: outside option}
Note that $\delta$ is small, so let $\delta<\frac{4}{b}$.

\noindent{\bf Stage 3:}
We consider interior NE strategies, i.e., $0<n_F, n_{L}<1$. Using Definition \ref{definition: new_demand}, \eqref{equ: BM-F-payoff}, \eqref{equ: BM-L-payoff} and \eqref{equ: BM-demand}, note that $v^{L}=v^{F}$, the payoffs of SPs are:
\begin{equation}\label{equ: out-stage 3}
\small
\begin{aligned}
\pi_{F}=&\alpha(t_{L}+k+p_{L}-2p_{F}+bI_{F})(p_{F}-c)-sI_{F}^{2}\\
\pi_{L}=&\alpha(t_{F}+k+p_{F}-2p_{L}+bI_{L}-bI_{F})(p_{L}-c)\\
+&sI_{F}^{2}-\gamma I_{L}^{2}
\end{aligned}
\end{equation}
We characterize the NE of access fees as follows,
\begin{theorem}\label{thm: out-pl pf}
For given $I_{F}$ and $I_{L}$, the  NE strategies of access fees are unique, and are:
\begin{equation}\label{equ: price}
\begin{aligned}
p_{L}^{*}=&\frac{1}{15}+\frac{2c}{3}+\frac{k}{3}+\frac{t_{F}}{5}-\frac{b}{5}I_{F}+\frac{4b}{15}I_{L},\\
p_{F}^{*}=&\frac{1}{15}+\frac{2c}{3}+\frac{k}{3}+\frac{t_{L}}{5}+\frac{b}{15}I_{L}+\frac{b}{5}I_{F}.
\end{aligned}
\end{equation}
if and only if  $I_{L}$ satisfies:
\begin{equation}\label{equ: condition 1}
I_{L}<\frac{4}{b}.
\end{equation}
\end{theorem}
\begin{proof}
In this case, every  NE by which $0\leq x_0\leq 1$, should satisfy the first order condition. Thus $p_{L}^{*}$ and $p_{F}^{*}$ should be such that 
$$\frac{d\pi_{L}}{dp_{L}}|_{p_{L}^{*}}=0,\, \frac{d\pi_{F}}{dp_{F}}|_{p_{F}^{*}}=0,$$ note that $t_{L}+t_{F}=1$,
then
\begin{align*}
p_{L}^{*}=&\frac{1}{15}+\frac{2c}{3}+\frac{k}{3}+\frac{t_{F}}{5}-\frac{b}{5}I_{F}+\frac{4b}{15}I_{L},\\
p_{F}^{*}=&\frac{1}{15}+\frac{2c}{3}+\frac{k}{3}+\frac{t_{L}}{5}+\frac{b}{15}I_{L}+\frac{b}{5}I_{F}.
\end{align*}

Take the second derivative of $\pi_{L}$ with respect to $p_{L}$,
$$\frac{d^{2}\pi_{L}}{d(p_{L}^{*})^{2}}=\frac{d^{2}\pi_{F}}{d(p_{F}^{*})^{2}}=-4\alpha<0,$$
then $p_{L}^{*}$ and $p_{F}^{*}$ are the unique maximal solutions of $\pi_{L}$ and $\pi_{F}$, respectively.

Thus, $p^*_F$ and $p^*_L$ are the unique interior NE strategies if and only if $0<x_0<1$. Substituting (\ref{equ: price}), $t_{L}=I_{F}/I_{L}$, and $t_{F}=(I_{L}-I_{F})/I_{L}$ into (\ref{equ: BM-indifferent location}) yields:
$$x_{0}=\frac{4}{5}-\frac{b}{5}I_{L}+(\frac{2b}{5}-\frac{3}{5I_{L}})I_{F}\triangleq \Psi(I_{F}).$$
Once $I_L$ is fixed, $\Psi(I_{F})$ would be  a linear function of $I_{F}$. Thus, $0<\Psi(I_F)<1$  for any values of $I_F$ such that $0\leq I_F\leq I_L$, if and only if 
\begin{align*}
&0<\Psi(0)<1\\
&0<\Psi(I_L)<1.
\end{align*}
Thus,
\begin{align*}
\Psi(I_{L})=&\frac{1}{5}+\frac{b}{5}I_{L}\in(0, 1)\\
\Psi(0)=&\frac{4}{5}-\frac{b}{5}I_{L}\in(0, 1)
\end{align*}  
if and only if $0<I_{L}<\frac{4}{b}$. 
\end{proof}

\noindent{\bf Stage 2:} 
 Based on the NE strategies of access fees, we obtain the optimum investment level of the MVNO.
 \begin{definition}
$g(I_{L})=\frac{b}{15}I_{L}+\frac{1}{15}-\frac{c}{3}+\frac{k}{3}$, $f(I_{L})=\frac{1}{5I_{L}}+\frac{b}{5}>0$
\end{definition}

\begin{theorem}\label{thm: out-I_F}
If $\pi_F(I_F;I_L)\geq 0$, and denote 
$$I_{F}^{0}=\frac{-2\alpha f(I_{L})g(I_{L})}{2\alpha f^{2}(I_{L})-s}.$$
Then, the unique optimal investment level of $\text{SP}_{F}$, $I_{F}^{*}$, is:
\begin{equation}\label{equ: I_F^*}
\footnotesize
I_{F}^{*}=\left\{\begin{aligned}
&I_{F}^{0}&\text{if}\quad& I_{L}\in\{s>2\alpha f^{2}(I_{L})+2\alpha f(I_{L})g(I_{L})/I_{L},\\
&\quad&\quad&g(I_{L})\geq0\}\\
&I_{L}&\text{if}\quad& I_{L}\in\{2\alpha f^{2}(I_{L})\leq s\leq2\alpha f^{2}(I_{L})\\
&\quad&\quad&+2\alpha f(I_{L})g(I_{L})/I_{L},g(I_{L})\geq0\}\\
&\quad&\quad&\cup\{2\alpha f^{2}(I_{L})+4\alpha f(I_{L})g(I_{L})/I_{L}\geq s,\\
&\quad&\quad&2\alpha f^{2}(I_{L})>s\}
\end{aligned}\right.
\end{equation}
\end{theorem}
\begin{proof}
First, we give the following the lemma
\begin{lemma}
The optimum investment level $I_{F}^{*}$ is obtained by
\begin{equation}\label{equ: stage 2}
\footnotesize
\begin{aligned}
\max_{I_{F}}\quad&\pi_{F}=(2\alpha f^{2}(I_{L})-s)I_{F}^{2}+4\alpha f(I_{L})g(I_{L})I_{F}+2\alpha g^{2}(I_{L})\\
s.t\quad &0\leq I_{F}\leq I_{L}.
\end{aligned}
\end{equation}
\end{lemma}

\begin{proof}
Substituting (\ref{equ: price}) into $\pi_{F}$ in (\ref{equ: out-stage 3}), we get the objective function. The constraints come from the model assumptions directly.
\end{proof}
We consider different cases. First, we consider the case that $2\alpha f^{2}(I_{L})-s=0$ (Step (i)). Then, we consider the case that $2\alpha f^{2}(I_{L})-s\neq0$ and $\pi_F$ is a quadratic function of $I_F$ (Step (ii)). In Step (iii), we prove that $I^*_F\neq 0$. Combining the steps yields the result of the theorem.

\noindent{\textbf{Step (\romannumeral1)}:} If $2\alpha f^{2}(I_{L})-s=0$, $\pi_{F}$ is linear function of $I_F$, i.e.,
$\pi_{F}=4\alpha f(I_{L})g(I_{L})I_{F}+2\alpha g^{2}(I_{L})$
Thus,
\begin{align*}
\left\{
\begin{aligned}
&I_{F}^{*}=0& \,\,&\text{if}\quad g(I_{L})<0\\
&I_{F}^{*}=I_{L}& \,\,&\text{if}\quad g(I_{L})\geq0
\end{aligned}
\right..
\end{align*}

\noindent{\textbf{Step (\romannumeral2)}:} Now, consider the case that $2\alpha f^{2}(I_{L})-s\neq0$ and $\pi_F$ is a quadratic function of $I_F$. We characterize the optimum answer in two cases: (a) if $2\alpha f^{2}(I_{L})-s>0$, and (b) if $2\alpha f^{2}(I_{L})-s<0$, $\pi_{F}(I_{F}; I_{L})$.

For the case that $\pi_{F}$ is a quadratic function, we use   the solution to the first order condition ($I_{F}^{0}$),
$$\frac{d\pi_{F}}{dI_{F}}|_{I_{F}^{0}}=0\Rightarrow I_{F}^{0}=\frac{-2\alpha f(I_{L})g(I_{L})}{2\alpha f^{2}(I_{L})-s}.$$

\noindent{\textbf{Case (ii-a)}}: If $2\alpha f^{2}(I_{L})-s>0$, then $\pi_{F}$ is convex function. From Lemma~\ref{lem: quadratic function}~(1),
\begin{align*}
\footnotesize
\left\{
\begin{aligned}
&I_{F}^{0}-\frac{I_{L}}{2}\leq0 &\text{if}\quad 2\alpha I_{L}f^{2}(I_{L})+4\alpha f(I_{L})g(I_{L})-I_{L}s\geq0\\
&I_{F}^{0}-\frac{I_{L}}{2}>0 &\text{if}\quad 2\alpha I_{L}f^{2}(I_{L})+4\alpha f(I_{L})g(I_{L})-I_{L}s<0
\end{aligned}
\right.,
\end{align*}
thus
\begin{align*}
\left\{
\begin{aligned}
&I_{F}^{*}=I_{L}& \,\,&\text{if}\, 2\alpha I_{L}f^{2}(I_{L})+4\alpha f(I_{L})g(I_{L})-I_{L}s\geq0\\
&I_{F}^{*}=0& \,\,&\text{if}\, 2\alpha I_{L}f^{2}(I_{L})+4\alpha f(I_{L})g(I_{L})-I_{L}s<0
\end{aligned}
\right..
\end{align*}

\noindent{\textbf{Case (ii-b}):} If $2\alpha f^{2}(I_{L})-s<0$, then $\pi_{F}$ is a concave function. Thus, from Lemma~\ref{lem: quadratic function}~(2),
\begin{align*}
\footnotesize
\left\{
\begin{aligned}
&I_{F}^{0}-0<0& \,\,&\text{if}\quad g(I_{L})<0\\
&0\leq I_{F}^{0}< I_{L}& \,\,&\text{if}\quad 2\alpha I_{L}f^{2}(I_{L})+2\alpha f(I_{L})g(I_{L})-I_{L}s<0,\\
&\quad&\,\,&g(I_{L})\geq0\\
&I_{F}^{0}\geq I_{L}& \,\,&\text{if}\quad
2\alpha I_{L}f^{2}(I_{L})+2\alpha f(I_{L})g(I_{L})-I_{L}s\geq0,\\
&\quad&\,\,&g(I_{L})\geq0
\end{aligned}
\right.,
\end{align*}
Thus,
\begin{align*}
\footnotesize
\left\{
\begin{aligned}
&I_{F}^{*}=0& \,\,&\text{if}\quad g(I_{L})<0\\
&I_{F}^{*}=I_{F}^{0}& \,\,&\text{if}\quad2\alpha I_{L}f^{2}(I_{L})+2\alpha f(I_{L})g(I_{L})-I_{L}s<0,\\
&\quad&\,\,& g(I_{L})\geq0\\
&I_{F}^{*}=I_{L}& \,\,&\text{if}\quad
2\alpha I_{L}f^{2}(I_{L})+2\alpha f(I_{L})g(I_{L})-I_{L}s\geq0,\\
&\quad&\,\,& g(I_{L})\geq0
\end{aligned}
\right..
\end{align*}

\noindent{\textbf{Step (\romannumeral3)}:} We now prove $I_{F}^{*}\neq0$. From Case (ii-a), if $I_{F}^{*}=0$, then 
$$2\alpha I_{L}f^{2}(I_{L})+4\alpha f(I_{L})g(I_{L})-I_{L}s<0,$$ i.e., 
$$s>2\alpha f^{2}(I_{L})+4\alpha f(I_{L})g(I_{L})/I_{L},$$ which implies $g(I_{L})<0$ since $2\alpha f^{2}(I_{L})-s>0$. Thus from Step (i), and Cases (ii-a) and (ii-b), if $I_{F}^{*}=0$,  then $g(I_{L})<0$.

Since $t_{L}^{*}=0$ and $t_{F}^{*}=1$, when $I_{F}^{*}=0$, then 
$$p_{F}^{*}-c=\frac{1}{15}-\frac{c}{3}+\frac{k}{3}+\frac{b}{15}I_{L}=g(I_{L})<0.$$ For an equilibrium solution $p_{F}^{*}$, $p_{F}^{*}\geq c$, otherwise 
$$\pi_{L}^{*}=\tilde{n}^{*}_{F}(p_{F}^{*}-c)-s(I_{F}^{*})^{2}<0.$$ Hence $I_{F}^{*}=0$ can not be an equilibrium solution for $\text{SP}_{F}$.

Combining Steps (\romannumeral1), (\romannumeral2),  and (\romannumeral3), we obtain the desired results.
\end{proof}

\noindent{\bf Stage 1:} 
Finally, we characterize the optimum investment level of SP$_L$, $I_L^*$.
\begin{theorem}\label{thm: out-generalization_stage1} The unique optimum investment level of $\text{SP}_{L}$, $I_{L}^{*}$, a solution of the following optimization problem:
\begin{equation}\label{equ: stage 1}
\small
\begin{aligned}
\max_{I_{L}}\quad&\pi_{L}(I_{L})=2\alpha(\frac{b}{5}I_{L}+\frac{1}{5}+g(I_{L})-f(I_{L})I_{F}^{*})^{2}\\
&+s(I_{F}^{*})^{2}-\gamma I_{L}^{2}\\
s.t\quad&\delta\leq I_{L}\\
&I_{L}<4/b.
\end{aligned}
\end{equation}
\end{theorem}

\begin{proof}
Substituting (\ref{equ: price}) into $\pi_{L}$ in (\ref{equ: out-stage 3}), we get the objective function. The constraints come from the model assumptions directly. 
\end{proof}

 We define functions $f(I_L)$, $g(I_L)$, $\pi_{L}(I_{F})$ and sets $\mathbb{L}_{1}$, $\mathbb{L}_{2}$ as follows:
\begin{equation*}
\footnotesize
\begin{aligned}
&g(I_{L})=\frac{b}{15}I_{L}+\frac{1}{15}-\frac{c}{3}+\frac{k}{3}, \,\, f(I_{L})=\frac{1}{5I_{L}}+\frac{b}{5},\\
&\theta(y)=2\alpha \big(\frac{b}{5}I_{L}+\frac{1}{5}+g(I_{L})-f(I_{L})y\big)^{2}+sy^{2}-\gamma I_{L}^{2},	
\end{aligned}
\end{equation*}
\begin{equation*}
\footnotesize
\begin{aligned}
\mathbb{L}_{1}=&\{s>2\alpha f^{2}(I_{L})+2\alpha f(I_{L})g(I_{L})/I_{L},\, g(I_{L})\geq0,\\
&\delta\leq I_{L}, I_L<4/b\},	
\end{aligned}
\end{equation*}
\begin{equation*}
\footnotesize
\begin{aligned}
\mathbb{L}_{2}=&\{0\leq I_{L}, I_L<4/b\}\cap\Big(\{g(I_{L})\geq0,\\
&\,2\alpha f^{2}(I_{L})\leq s\leq2\alpha f^{2}(I_{L})+2\alpha f(I_{L})g(I_{L})/I_{L}\}\\
\cup&\{2\alpha f^{2}(I_{L})+4\alpha f(I_{L})g(I_{L})/I_{L}\geq s,\,2\alpha f^{2}(I_{L})>s\}\Big).	
\end{aligned}
\end{equation*}
Collecting results in Stages 1$\sim$4, we have 
\begin{corollary}
The interior SPNE strategies are:
\begin{itemize}
  \item [(1)]  $I_{L}^{*}$ is characterized in
 \begin{equation*}
\footnotesize
\begin{aligned}
I_{L}^{*}=\argmax_{I_{L}}\Big(\max_{I_{L}\in\mathbb{L}_{1}}\theta(\frac{-2\alpha f(I_{L})g(I_{L})}{2\alpha f^{2}(I_{L})-s}),\max_{I_{L}\in\mathbb{L}_{2}}\theta(I_{L})\Big)
\end{aligned}
\end{equation*}

  \item [(2)]  $I_{F}^{*}$ is characterized in
\begin{equation*}
 \footnotesize
 \begin{aligned}
  I_{F}^{*}=\left\{\begin{aligned}
  &\frac{-2\alpha f(I_{L})g(I_{L})}{2\alpha f^{2}(I_{L})-s}\,\,& \text{if}\,\, I_{L}\in\mathbb{L}_{1}\\
  &I_{L}\,\, & \text{if}\,\, I_{L}\in\mathbb{L}_{2}\end{aligned}\right.
\end{aligned}
\end{equation*}
 \item [(3)]
$p_{L}^{*}=\frac{1}{15}+\frac{2c}{3}+\frac{k}{3}+\frac{I_{L}^{*}-I_{F}^{*}}{5I_{L}^{*}}-\frac{b}{5}I_{F}^{*}+\frac{4b}{15}I_{L}^{*}$,
$p_{F}^{*}=\frac{1}{15}+\frac{2c}{3}+\frac{k}{3}+\frac{I_{F}^{*}}{5I_{L}^{*}}+\frac{b}{15}I_{L}^{*}+\frac{b}{5}I_{F}^{*}$.

\item [(4)] $\tilde{n}_{L}^{*}=\frac{I_{L}^{*}-I_{F}^{*}}{I_{L}^{*}}+p_{F}^{*}-2p_{L}^{*}+k+bI_{L}^{*}-bI_{F}^{*}$,
  $\tilde{n}_{F}^{*}=\frac{I_{F}^{*}}{I_{L}^{*}}+p_{L}^{*}-2p_{F}^{*}+k+bI_{F}^{*}$
  \end{itemize}
\end{corollary}

\section{Proof of Corollary~1}\label{Appendix: Corollary}
\noindent{\bf Stage 4:}
Similar with Definition~\ref{def: BM-indifferent location},
$u_{F}(x_{0})=v-t(2\pi-x_{0})-p_{F}=v-tx_{0}-p_{L}=u_{L}(x_{0})$,
thus,
\begin{align}\label{equ: variant x_0}
x_0=\pi+\frac{p_F-p_L}{2t}.
\end{align}
Since EUs are distributed uniformly along $[0,2\pi]$,
the fraction of EUs with each SP  is:
\begin{equation}\label{equ: BM-variant-demand}
\begin{aligned}
&n_{L}=\left\{\begin{aligned}
&0,&\,\text{if}\quad&x_{0}\leq0\\
&x_{0},&\,\text{if}\quad&0<x_{0}<2\pi\\
&2\pi,&\,\text{if}\quad&x_{0}\geq2\pi\\
\end{aligned}\right.,\, n_{F}=2\pi-n_{L},
\end{aligned}
\end{equation}
where $x_{0}$ is defined in \eqref{equ: variant x_0} and $n_F = 2\pi-n_L$.

Only ``interior'' strategies may be SPNE, as:

\begin{theorem}\label{thm: Un-variant-3sp}
In the SPNE it must be that $0<x_{0}<2\pi.$
\end{theorem}
\begin{proof}
Let $(p_{L}^{*},p_{F}^{*},I_{L}^{*},I_{F}^{*})$ be a corner SPNE strategy. Thus, 1) $x_{0} \geq2\pi$, or 2)  $x_{0}\leq0$. We arrive at a contradiction for  1)  {\bf Step 1}  and 2) in {\bf Step 2}
respectively.
 \begin{lemma}\label{lem_variant1}
 $\pi_{F}^{*}\geq0$. If $n_F^* > 0,$  $p_{F}^{*} \geq c.$
  \end{lemma}
\begin{proof}
Let $\pi_{F}^*<0$. Consider a unilateral deviation in which $I_F = 0, p_F \geq c.$ From \eqref{equ: BM-F-payoff}, $\pi_F \geq 0$, leading to a contradiction. Now, let $n_F^* > 0$ and  $p_{F}^{*} < c$. Thus, $\pi_{F}^*<0$ which is a contradiction.
\end{proof}

\noindent{\bf Step 1}. Let $x_{0}^*\geq2\pi$. Clearly,  $n_{F}^*=0$ and $n_{L}^*=2\pi$.
From \eqref{equ: BM-F-payoff}, $\pi_{F}^{*}=-sI_{F}^{*2}.$
 From Lemma~\ref{lem_variant1}, $I_F^* = 0.$ Thus, $\pi_F^*=0.$
From \eqref{equ: variant x_0}, $2\pi \leq x_{0}^*=\pi+\frac{p_F^*-p_L^*}{2t}$. Thus,  $p_{F}^*\geq p_{L}^*+2\pi t$.

From \eqref{equ: BM-L-payoff}, $ \pi_{L}^{*}=2\pi(p_{L}^{*}-c)-\gamma I_{L}^{*2}.$
If $p_{L}^{*}<c$, then $\pi_{L}^*<-\gamma \delta^{2} < 0$ since $I_{L}^{*}\geq\delta$. Consider a unilateral deviation by which  $I_{L}=\delta, p_{L}=c$, then
$\pi_{L} = -\gamma\delta^{2}$, which is beneficial for SP$_L$.   Thus, $p_{L}^{*}\geq c$.

Now, let $p_L^* > c.$  Thus, $p_{F}^{*} \geq p_{L}^{*}+2\pi t > c+2\pi t>c$. Recall that $x_0^* = \pi+\frac{p_F^*-p_L^*}{2t}.$
 Consider a unilateral deviation by which $p_F = p_L^*+2\pi t-\epsilon$. Now, by \eqref{equ: variant x_0}, $x_0 < 2\pi$, and hence $n_F > 0.$  Now, from \eqref{equ: BM-F-payoff}, $\pi_F > 0 = \pi_F^*$. Thus, $(I_F^*, p_F^*)$ is not  SP$_F$'s best response  to SP$_L$'s choices $(I_L^*, p_L^*)$,    which is a contradiction. Hence, $p_L^* = c.$

Now consider another unilateral deviation of SP$_{L}$, $p_{L}'=p_{F}^{*}-2\pi t+\epsilon$, where $0 < \epsilon < \min(1,t)$,  with all the rest the same. Since $p_L^* \leq p_F^*-t$, $p_L' > p_L^* = c.$

\begin{align*}
&n_{L}'=x_{0}'=\pi+\frac{p_F^*-p_L'}{2t}=2\pi-\frac{\epsilon}{2t}.
\end{align*}

Then
\[\pi_{L}'-\pi_{L}^{*}=n_{L}'(p_{L}'-c) - (p_{L}^{*}-c) =
(2\pi-\frac{\epsilon}{2t})(p_L' - c)>0.\]
The last inequality follows because $p_L' > c$ and $\epsilon < \min(1,t).$ Thus, we again arrive at a contradiction.

\noindent{\bf Step 2}. Let $x_0^* \leq 0.$ Clearly, $n_{F}^*=2\pi, n_{L}^*=0$. Since $n_F^* > 0 $, by Lemma~\ref{lem_variant1},  $p_{F}^{*}\geq c$.
From \eqref{equ: BM-indifferent location}, $x_{0}^*=\pi+\frac{p_F^*-p_L^*}{2t}\leq0$. Thus,
  $p_L^* \geq p_F^* + 2\pi t.$ Now, from  \eqref{equ: BM-L-payoff},
\begin{align} \label{fallback}
 &\pi_{L}^*=sI_{F}^{*2}-\gamma I_{L}^{*2}.
 \end{align}
 Consider a unilateral deviation by SP$_{L}$, by which  $p_{L}'=2\pi t + p_F^* -\epsilon$, $0 < \epsilon < \min(1,t)$. Then 
\begin{align*}
&n_{L}'=x_{0}'=\pi+\frac{p_F^*-p_L'}{2t}=\frac{\epsilon}{2t}>0
\end{align*}
Therefore, by \eqref{fallback},
\begin{align*}
\pi_{L}'-\pi_{L}^{*}=n_{L}'(p_{L}'-c)=\frac{\epsilon}{2t}(p_{F}^{*}-\epsilon+2\pi t-c)
\end{align*}
Since $p_{F}^*\geq c$, and $\epsilon<\min(1,t)$. Then,  $\pi_{L}'-\pi_{L}^{*}>0$. We again arrive at a contradiction. \qed

By Theorem~\ref{thm: Un-variant-3sp} proved above henceforth we only consider interior SPNE in which $0 < x_0^* < 2\pi.$
\end{proof}

\noindent{{\bf Stage 3:}}
SP$_L$ and SP$_F$ determine their access fees for EUs, $p_L$ and $p_F$, respectively, to maximize their payoffs.

\begin{lemma}\label{lem: BM-variant-3-payoffs}
The payoffs of SPs are:

\begin{equation}\label{equ: BM-variant-payoffs}
\begin{aligned}
\pi_{L}=&\frac{1}{2t}(2\pi t+p_{F}-p_{L})(p_{L}-c)+sI_{F}^{2}-\gamma I_{L}^{2}\\
\pi_{F}=&\frac{1}{2t}(2\pi t+p_{L}-p_{F})(p_{F}-c)-sI_{F}^{2}
\end{aligned}
\end{equation}
\end{lemma}

\begin{proof}
From \eqref{equ: variant x_0} and \eqref{equ: BM-variant-demand}, substitute $(n_{L}, n_{F})=(\pi+\frac{p_{F}-p_{L}}{2t}, 2\pi-n_{L})$ into (\ref{equ: BM-L-payoff}) and (\ref{equ: BM-F-payoff}), and get (\ref{equ: BM-variant-payoffs}).
\end{proof}

We next obtain the SPNE $p_{F}^{*}$ and $p_{L}^{*}$ which maximize the payoffs $\pi_{L}$ and $\pi_{F}$ of the SPs respectively. 

\begin{theorem}\label{thm: BM-variant-prices}
The SPNE pricing strategies  are:
\begin{equation}\label{equ: BM-variant-prices}
p_{L}^{*}=c+2\pi t,\quad p_{F}^{*}=c+2\pi t
\end{equation}
\end{theorem}
\begin{proof}
$p^*_F$ and $p^*_L$ must satisfy the first order condition, i.e., $\frac{d \pi_F}{dp_F}=0$ and   $\frac{d \pi_L}{dp_L}=0$.  Thus,
$p^*_F=p_L^*=c+2\pi t$.
 $p^*_F$ and $p^*_L$ are the unique SPNE strategies if they yield $0< x_0< 2\pi$ and no unilateral deviation is profitable for SPs. We establish these respectively in  Parts A and B.

\noindent{{\bf Part A}.}
From (\ref{equ: BM-variant-prices}),
$x_0=\pi+\frac{p_F^*-p_L^*}{2t}=\pi\in(0,2\pi)$ since $p_L^*=p_F^*=2\pi t+c$.

\noindent{{\bf Part B}.}
Since $\frac{d^2 \pi_F}{dp^2_F}<0, \frac{d^2 \pi_L}{dp^2_L}<0$,  a local maxima is also a global maximum, and any solution to the first order conditions  maximize the payoffs   when $0<x_0<2\pi$,  and no unilateral deviation by which $0<x_0<1$ would be profitable for the SPs.
Now, we show that unilateral deviations of the SPs leading to  $n_L=0, n_F=2\pi$ and $n_L=2\pi, n_F=0$  is not profitable. Note that the payoffs of the SPs, \eqref{equ: BM-L-payoff} and \eqref{equ: BM-F-payoff}, are continuous  as $n_L\downarrow 0$, and $n_L\uparrow 2\pi$ (which subsequently yields  $n_F\uparrow 2\pi$ and $n_F\downarrow 0$, respectively). Thus, the payoffs of both SPs when selecting $p_L$ and $p_F$ as the  solutions of the first order conditions are greater than or equal to the payoffs when $n_L=0$ and $n_L=2\pi$. Thus, the unilateral deviations under consideration are not profitable for the SPs.
\end{proof}

\noindent{{\bf Stage 2:}}
SP$_F$ decides on  the amount of spectrum to be leased from SP$_L$, $I_F$, with the condition that $0\leq I_F\leq I_L$,  to maximize $\pi_F$.

\begin{theorem}\label{thm: BM-varaint-I_F}
The SPNE spectrum acquired by $\text{SP}_{F}$ is: $I_{F}^{*}=0.$
\end{theorem}

\begin{proof}

Substituting $p_{F}$ and $p_{L}$ from \eqref{equ: BM-variant-prices} into \eqref{equ: BM-variant-payoffs},  $\text{SP}_{F}$'s payoff  becomes,
\begin{equation}\label{equ: BM-stage-2-varaint-F-payoff}
\pi_{F}(I_{F}; I_{L})=2\pi^2 t-sI_F^2.	
 \end{equation}
Since $\pi_{F}(I_{F}; I_{L})$ is a decreasing function of $I_F$ and $0\leq I_{F}\leq I_{L}$, so $I_F^*=0$.
\end{proof}

\noindent{{\bf Stage 1:}}
 SP$_L$ chooses the amount of spectrum $I_L$ to lease from the regulator, to maximize $\pi_L$.
\begin{theorem}\label{thm: stage 1-unequal}
The SPNE spectrum acquired by $\text{SP}_{L}$ is: $I_{L}^{*}=\delta.$
\end{theorem}

\begin{proof} Substituting $p_{L}$ and $p_{F}$ from \eqref{equ: BM-variant-prices} into \eqref{equ: BM-variant-payoffs}, $\text{SP}_{L}$'s payoff becomes:
\begin{equation}\label{equ: BM-L-payoff1}
\pi_{L}(I_{L}; I_{F}^{*})=2\pi^2t-\gamma I_{L}^{2}.
 \end{equation}
since from Theorem~\ref{thm: BM-varaint-I_F}, $I_F^*=0.$
Note that $\pi_L$ is a decreasing function of $I_L$, and $\delta\leq I_L\leq M$, so $I_L^*=\delta.$
\end{proof}

Collecting all SPNE from {\bf Stages~1$\sim$4}, the unique SPNE strategies
 are:
 \[I_{L}^{*}=\delta, \ I_{F}^{*}=0, \ p_{L}^{*}=p_{F}^{*}=2t\pi+c, \ n_{F}^{*}= n_{L}^{*}=\pi .\]

\section{Limited Spectrum: SPNE Analysis}\label{Appendix: Limited Spectrum}

\noindent{\bf Proof of Theorem~\ref{thm: Un-conclusion-sectionA, M}.}
\begin{proof}
The proofs of in {\bf Stage 2} (finding $I_F^*$), {\bf Stage 3} (finding $p_L^*, p_F^*$) and {\bf Stage 4} (finding $n_L^*, n_F^*$) are the same as proofs of Theorems~\ref{thm: Un-A-indiffenrent location}, \ref{thm: Un-A-prices}, \ref{thm: Un-A-I_F-sectionA}, respectively. Now we only consider {\bf Stage 1} (finding $I_L^*$). Similar with the proof of Theorem~\ref{thm: Un-payoff-L-sectionA},
substituting $I_{F}^{*}$ in \eqref{equ: Un-I_F1} into (\ref{equ: Un-L-payoff1}), the optimal investment level of $\text{SP}_{L}$, $I_{L}^{*}$, is a solution of the following optimization problem,
\begin{equation}\label{equ: Un-interior-pi_L-proof1, M}
\small
\begin{aligned}
\max_{I_{L}}\quad&\pi_{L}(I_{L})=(\frac{2+\Delta}{3}-\frac{I_{F}^{*}}{3I_{L}})^{2}+s(I_{F}^{*})^{2}-\gamma I_{L}^{2}\\
s.t\quad&\delta\leq I_{L}\leq M.
\end{aligned}
\end{equation}

 If $M\leq \sqrt{\frac{2-\Delta}{9s}}$, from (\ref{equ: Un-I_F1}) in Theorem~\ref{thm: Un-A-I_F-sectionA}, $I_{F}^{*}=I_{L}$, thus (\ref{equ: Un-interior-pi_L-proof1, M}) is equivalent to
\begin{align*}
\max_{I_{L}}\quad&\pi_{L}(I_{L})=\frac{(1+\Delta)^{2}}{9}+(s-\gamma)I_{L}^{2}\\
&\delta\leq I_{L}\leq M
\end{align*}
Since $s>\gamma$, then $\pi_{L}(I_{L})$ is an increasing function of $I_{L}$, thus $I_{L}^{*}=M$. In this case, $I_F^*=M$, and from \eqref{equ: Un-A-indifferent location}, \eqref{equ: Un-A-demand} and \eqref{equ: Un-A-interior-prices}, $n_L^*=p_L^*-c=\frac{\Delta+1}{3}$ and $n_F^*=p_F^*-c=\frac{2-\Delta}{3}$.

 If $M> \sqrt{\frac{2-\Delta}{9s}}$, the proof are the same with that of Theorem~\ref{thm: Un-payoff-L-sectionA}.

\end{proof}

\noindent{\bf Proof of Theorem~\ref{thm: Un-corner-sectionA, M}.}
\begin{proof}
The proofs of in {\bf Stage 3} (finding $p_L^*, p_F^*$) are the same as proofs of Lemmas~\ref{lem: no negative-corner SPNE} \ref{lem: no positive-corner SPNE} and Theorems~\ref{thm: negative-corner} and \ref{thm: positive-corner}. Now we only consider {\bf Stages 1, 2} (finding $I_L^*, I_F^*$). 

{\bf (A)} We consider $\Delta\leq-1$. Similar with the proof of Theorem~\ref{thm: negative-corner},	
 substituting $I_{F}^{*}$ from {\bf Step~5} in Theorem~\ref{thm: negative-corner} into (\ref{equ: BM-L-payoff}), the optimum investment level of SP$_{L}$, $I_{L}^{*}$, is a solution of the following optimization problem,
\begin{equation}\label{equ: Un-L-payoff2-sectionC, M}
\begin{aligned}
\max_{I_{L}}\quad&\pi_{L}(I_{L}; I_{F}^{*})=sI_{F}^{*2}-\gamma I_{L}^{2}\\
s.t\quad&\delta\leq I_{L}\leq M
\end{aligned}
\end{equation} 
Then, we consider two sub-cases: $\delta\leq M\leq\frac{1}{\sqrt{2s}}$ and $M>\frac{1}{\sqrt{2s}}$. 

If $\delta\leq M\leq\frac{1}{\sqrt{2s}}$,  Since $I_L\leq M\leq\frac{1}{\sqrt{2s}}$, then $I_{F}^{*}=I_{L}$, thus the optimization (\ref{equ: Un-L-payoff2-sectionC, M}) is equivalent to
\begin{align*}
\max_{I_{L}}\quad&\pi_{L}=(s-\gamma)I_{L}^{2}\\
s.t\quad&\delta\leq I_{L}\leq M.
\end{align*}
Note that $s>\gamma$, then $\pi_{L}$ is an increasing function of $I_{L}$, thus $I_{L}^{*}=I_{F}^{*}=M$. If $\frac{1}{\sqrt{2s}}<M$, the proof is the same as the proof in Theorem~\ref{thm: negative-corner}.

{\bf (B)} We consider $\Delta\geq1$. The proof is the same as the proof in Theorem~\ref{thm: positive-corner}.

\end{proof}

\noindent{\bf Proof of Theorem~\ref{cor: 3-p model, M}.}
\begin{proof}
The proofs of in {\bf Stage 2} (finding $I_F^*$), {\bf Stage 3} (finding $p_L^*, p_F^*$) and {\bf Stage 4} (finding $n_L^*, n_F^*$) are the same as proofs of  Theorems~\ref{Thm: 3p-interior-prices},  \ref{thm: 3p-demand-division} and \ref{thm: 3p-interior-I_F}. Now we only consider {\bf Stage 1} (finding $I_L^*$). Similar with the proof of Theorem~\ref{thm: 3p-payoff-L-sectionA},	each MNO chooses its $I_L$ as the solution of the following maximization: \begin{equation}\label{equ: 3p-pi_L-proof1, M}
\begin{aligned}
\max_{I_{L}}\quad&\pi_{L}(I_{L})=\frac{t\pi^{2}}{18}(\frac{7I_{L}-I_{F}^{*}}{2I_{L}})^{2}+sI_{F}^{*2}-\gamma I_{L}^{2}\\
s.t\quad&\delta\leq I_{L}\leq M.
\end{aligned}
\end{equation}
The objective function follows by substituting (\ref{equ: 3p-interior-prices}) into (\ref{equ: 3p-interior-L-payoff}). The constraint follows from the modeling assumption. 

 If $M\leq\frac{\pi}{2}\sqrt{\frac{t}{3s}}$, from (\ref{equ: 3p-I_F_star}), $I_{F}^{*}=I_{L}$, thus the objective function of  (\ref{equ: 3p-pi_L-proof1, M}) is $\frac{t\pi^{2}}{2}+(s-\gamma)I_{L}^{2}.$
This is an increasing function of $I_{L}$ since $s>\gamma.$ Thus the optimum solution in this range is $M$.

 If $M>\frac{\pi}{2}\sqrt{\frac{t}{3s}}$, the proof is the same as that of Theorem~{thm: 3p-payoff-L-sectionA}.
\end{proof}

\end{document}